\documentclass[aps,prl,twocolumn,superscriptaddress,longbibliography,10pt,footinbib]{revtex4-2}
\usepackage[T1]{fontenc}
\usepackage{amsmath,amssymb,amsfonts,amsthm}
\usepackage{graphicx}
\usepackage{placeins}
\usepackage{booktabs}
\usepackage{microtype}
\usepackage{enumitem}
\usepackage{bm}
\usepackage{rotating}
\usepackage{pbox}
\usepackage{array}
\usepackage{mathtools}
\usepackage{leftidx}
\usepackage{lmodern}
\usepackage[normalem]{ulem}
\usepackage{braket}
\usepackage{tensor}
\usepackage[dvipsnames]{xcolor}
\usepackage{footnote}
\usepackage{setspace}
\usepackage{CJKutf8}
\usepackage{dsfont}
\usepackage{mathrsfs}
\usepackage[many]{tcolorbox}
\usepackage{palatino}
\usepackage{tikz}
\usetikzlibrary{arrows.meta, positioning, calc, fit, backgrounds}
\definecolor{agentaccent}{HTML}{2F6FAF}
\tikzset{
  agentbox/.style    = {draw=agentaccent, line width=0.6pt, fill=agentaccent!12},
  humanbox/.style    = {draw=black, fill=black!22},
  resourcebox/.style = {draw=black, dashed, fill=black!2},
  layerbox/.style    = {draw=black, fill=black!4},
}
\definecolor{bpVerified}{HTML}{008000}
\definecolor{bpProofDone}{HTML}{9CEC8B}
\definecolor{bpChainDone}{HTML}{1CAC78}
\usepackage[colorlinks,citecolor=blue,linkcolor=blue,urlcolor=blue]{hyperref}

\DeclareUrlCommand\leanid{\urlstyle{tt}}

\usepackage{listings}
\usepackage{textcomp}
\usepackage{upquote}
\lstdefinestyle{agentprompt}{%
  basicstyle=\ttfamily\footnotesize,
  breaklines=true,
  breakindent=0pt,
  breakautoindent=false,
  columns=fullflexible,
  keepspaces=true,
  upquote=true,
  showstringspaces=false,
  xleftmargin=0pt,
  aboveskip=0.4\baselineskip,
  belowskip=0.8\baselineskip,
  literate=%
    {—}{{\textrm{\textemdash}}}1
    {–}{{\textrm{\textendash}}}1
    {→}{{$\rightarrow$}}1
    {↔}{{$\leftrightarrow$}}1
    {∀}{{$\forall$}}1
    {α}{{$\alpha$}}1
    {β}{{$\beta$}}1
    {λ}{{$\lambda$}}1
    {·}{{$\cdot$}}1
    {≤}{{$\leq$}}1
    {ℕ}{{$\mathbb{N}$}}1
    {ℝ}{{$\mathbb{R}$}}1
    {ℤ}{{$\mathbb{Z}$}}1
    {⟨}{{$\langle$}}1
    {⟩}{{$\rangle$}}1
    {ₗ}{{$_{\mathrm{l}}$}}1
    {…}{{$\ldots$}}1
    {Č}{{\v{C}}}1
}

\usepackage[capitalize]{cleveref}
\crefname{figure}{Fig.}{Figs.}
\crefname{equation}{Eq.}{Eqs.}
\crefname{section}{Sec.}{Secs.}
\crefname{subsection}{Sec.}{Secs.}
\crefname{appendix}{App.}{Apps.}
\crefname{table}{Table}{Tables}
\Crefname{figure}{Fig.}{Figs.}
\Crefname{equation}{Eq.}{Eqs.}
\Crefname{section}{Sec.}{Secs.}
\Crefname{subsection}{Sec.}{Secs.}
\Crefname{appendix}{App.}{Apps.}
\Crefname{table}{Table}{Tables}

\newcommand{\nLOC}{62{,}000}
\newcommand{\nFiles}{233}
\newcommand{\nDecls}{2{,}300}
\newcommand{\nLOCfull}{227{,}000}
\newcommand{\nFilesfull}{705}
\newcommand{\nChapters}{12}

\newcommand{\nReviewRounds}{six}
\newcommand{\nPagesBluePrint}{150}

\newcommand{\nCost}{20{,}206}

\newcommand{\nCostPerKLOC}{89}
\newcommand{\nCostFT}{5{,}548}
\newcommand{\nCodexEq}{1{,}594}
\newcommand{\costMdlClaudeOpusFourSix}{6{,}393}
\newcommand{\costMdlGPTFiveFive}{4{,}567}
\newcommand{\costMdlGPTFiveFour}{3{,}484}
\newcommand{\costMdlClaudeOpusFourSeven}{2{,}760}
\newcommand{\costMdlGPTFiveTwo}{1{,}318}
\newcommand{\costMdlClaudeSonnetFourSix}{768}
\newcommand{\costMdlClaudeOpusFourEight}{589}
\newcommand{\costMdlDeepSeekpro}{235}
\newcommand{\costMdlGeminiThreeOnePro}{45}
\newcommand{\costMdlGPTFiveFourmini}{28}
\newcommand{\costMdlOther}{20}
\newcommand{\costRoleOrchestrator}{8{,}702}
\newcommand{\costRoleOrchestratorPerCall}{40}
\newcommand{\costRoleProofwriter}{7{,}545}
\newcommand{\costRoleProofwriterPerCall}{8}
\newcommand{\costRoleSimplifier}{1{,}136}
\newcommand{\costRoleSimplifierPerCall}{5}
\newcommand{\costRoleLibraryscout}{1{,}043}
\newcommand{\costRoleLibraryscoutPerCall}{4}
\newcommand{\costRoleBlueprintsync}{938}
\newcommand{\costRoleBlueprintsyncPerCall}{10}
\newcommand{\costRoleWorkflow}{160}
\newcommand{\costRoleOtherTU}{681}
\newcommand{\costRoleWorkflowPerCall}{2}
\newcommand{\costRoleOtherTUPerCall}{1}

\newcommand{\C}{\mathbb{C}}

\DeclareMathOperator{\Tr}{Tr}
\DeclareMathOperator{\tr}{tr}

\renewcommand{\tilde}{\widetilde}
\renewcommand{\bar}{\overline}
\DeclareMathOperator{\Span}{span}

\theoremstyle{plain}
\newtheorem{appthm}{Theorem}
\newtheorem{applem}{Lemma}
\newtheorem{theorem}{Theorem}

\newcommand{\prlsection}[1]{\emph{#1}.---}

\newif\ifshowcomments{}
    \showcommentstrue{} % Show comments
\pdfoutput=1

\begin{document}

\title{Multi-agent Autoformalization of Tensor Network Theory}

\author{Sirui Lu}
\affiliation{Max-Planck-Institut f\"ur Quantenoptik, Hans-Kopfermann-Stra\ss{}e 1, D-85748 Garching, Germany}
\affiliation{Munich Center for Quantum Science and Technology (MCQST), Schellingstraße 4, D-80799 Munich, Germany}
\email{sirui.lu@mpq.mpg.de}

\author{Erickson Tjoa}
\affiliation{Max-Planck-Institut f\"ur Quantenoptik, Hans-Kopfermann-Stra\ss{}e 1, D-85748 Garching, Germany}
\affiliation{Munich Center for Quantum Science and Technology (MCQST), Schellingstraße 4, D-80799 Munich, Germany}
\email{erickson.tjoa@mpq.mpg.de}

\author{J. Ignacio Cirac}
\affiliation{Max-Planck-Institut f\"ur Quantenoptik, Hans-Kopfermann-Stra\ss{}e 1, D-85748 Garching, Germany}
\affiliation{Munich Center for Quantum Science and Technology (MCQST), Schellingstraße 4, D-80799 Munich, Germany}
\email{ignacio.cirac@mpq.mpg.de}

\begin{abstract}
    We build a team of specialized large language-model agents and present an agent-driven workflow for research-level formalization in theoretical physics, with the autoformalization of the fundamental theorem of matrix-product states as a demonstration. The agents, coordinated through a structured mathematical blueprint and periodic human review,  orchestrated and executed the full formalization autonomously. For some statements,
    the agents were able to explore new proof routes that are not part of the standard literature. Along the way the agents produced extensive tensor-network and quantum-information libraries not previously available in Mathlib, Lean's mathematical library. As a physical application, the formalization also extends towards symmetry-protected topological phases in one dimension. We find that the main bottleneck in large-scale autoformalization is enforcing mathematical intent and we provide a detailed study of the full process and various subtleties involved. We release the codebase as the library \href{https://github.com/LionSR/TNLean}{TNLean}, together with a \nChapters{}-chapter \href{https://lionsr.github.io/TNLean/blueprint/}{blueprint} of the formalization effort.
\end{abstract}

\maketitle

\let\SMrealaddcontentsline\addcontentsline{}
\renewcommand{\addcontentsline}[3]{}

%======================================================================
\prlsection{Introduction}%
%======================================================================
Proof assistants are software systems that check every logical step of a mathematical argument. They have certified results from the four-color theorem~\cite{Gonthier2008Formal} to the Kepler conjecture~\cite{Hales2017Formal}, and AI systems can now find Lean~4 proofs automatically for many mathematical targets~\cite{Moura2021Lean,Hubert2026Olympiadlevel,Ren2025DeepSeekproverV2,Wang2025Kiminaprover,Lin2025GoedelproverV2}.
Most automated targets, however, remain competition problems or isolated lemmas, with research-level and textbook formalizations starting to grow in mathematics~\cite{MathInc.2025Strong,Rammal2026Formalizing}. Research-level theoretical physics poses a different challenge: since physics results outside mathematical physics are often not subject to the same mathematical rigor as pure mathematics~\cite{Lu2026Position}, formalization would need to fill in the informal parts as rigorous, Lean-checkable statements.

Large-scale formalizations have so far mostly been human-led, with projects such as the Liquid Tensor Experiment and Fermat's Last Theorem requiring teams of experts over months or years~\cite{Commelin2023Liquid,Gowers2025Conjecture,Best2025Complete,Buzzard2024Lean}. AI systems have begun to reduce this burden, from competition-style proving~\cite{Hubert2026Olympiadlevel,Axiom2025AxiomProver,Harmonic2025Aristotle} to autoformalization, local proof assistance, and short reasoning chains~\cite{Wu2022Autoformalization,Wang2025MALoT,Yang2023LeanDojo}; the AI system Gauss~\cite{2025Introducing,MathInc.2025Strong} recently produced a blueprint-guided formalization of the strong Prime Number Theorem.
In physics, PhysLib provides early Lean infrastructure~\cite{PhyslibContributors2024PhysLib}, while in quantum information the generalized quantum Stein's lemma~\cite{Brandao2010Generalization,Berta2023Gap,Hayashi2025Generalized,Lami2025Solution} has been partially formalized in a human-led effort~\cite{Meiburg2024LeanQuantumInfo,Meiburg2025Formalization}. These developments motivate the question we address here: how far can an agent-driven system carry a research-level formalization in theoretical physics?

In this work we build and apply an agent-driven workflow in which a team of specialized AI agents (language-model programs with distinct roles and tools), coordinated through a shared mathematical `blueprint' with sporadic human supervision, autonomously formalizes the fundamental theorem of matrix-product states (FT-MPS)~\cite{Perez-Garcia2007Matrix,Cirac2017Matrixa,Cirac2021Matrix} via the Lean~4 proof assistant. FT-MPS is chosen as it is one of the core results of tensor-network theory~\cite{Cirac2021Matrix} that lies at the intersection of quantum many-body physics, quantum information, and condensed-matter physics. Consequently, this task required us to develop a library well beyond FT-MPS itself. Our work is, to our knowledge, the first autonomous formalization that combines multi-agent coordination, persistent memory, and blueprint-guided planning for a research-level theorem in mathematical physics. As a byproduct, we have produced a \nPagesBluePrint{}-page formalization \textit{blueprint} of these \nChapters{} chapters~\cite{Massot2021Leanblueprint,TNLeanBlueprint} (a human-readable translation of the proof linked to the Lean code), developed several libraries adjacent to quantum information theory, and applied the fundamental theorem to physically relevant questions, such as the classification of one-dimensional symmetry-protected topological phases~\cite{Chen2011Classification,Schuch2011Classifying,Pollmann2010Entanglement,Pollmann2012Symmetry}. The codebase is released as a Lean~4 library, TNLean~\cite{TNLean}.

Our work reveals that, for research-level theoretical physics, the core challenge in autonomous formalization is no longer only the proof of individual lemmas, but also the faithful management of a large argument: breaking a published proof into hundreds of formal steps, deciding where to deviate from the literature's route, and tracking dependencies that span quantum mechanics, operator algebras, and spectral theory. In particular, the bounded context window of large language models (LLMs; the fixed amount of text a model can process at once) and the reliability with which agents follow instructions constrain how far LLM-based formalization can scale at present. We address these constraints through orchestration, the blueprint, automated review, and regular human review, as illustrated in \cref{fig:alignment}.

% fig_alignment.tex  --  literature <-> blueprint <-> Lean alignment loop.
%
%   Forward (solid): several sources, in different notations, are transcribed
%   into one blueprint statement (with \lean / \leanok links) and then proved
%   in Lean.  leanSearch scouts Mathlib; lean closes the proof.
%   Audit (dashed): leanBlueprint checks the mechanical blueprint<->Lean links
%   (checkdecls); the reviewer checks the two semantic correspondences,
%   blueprint-vs-Lean and literature-vs-Lean.  A mismatch sends the statement
%   back for restatement and re-proof.
%
%   Layout: compressed vertically (tier gaps tightened, legend flattened to a
%   single row in the empty bottom-right) and spread wider horizontally,
%   trading unused page width for column height.
%
% Requires:  \usepackage{tikz}
%            \usetikzlibrary{arrows.meta, positioning, calc, fit, backgrounds}

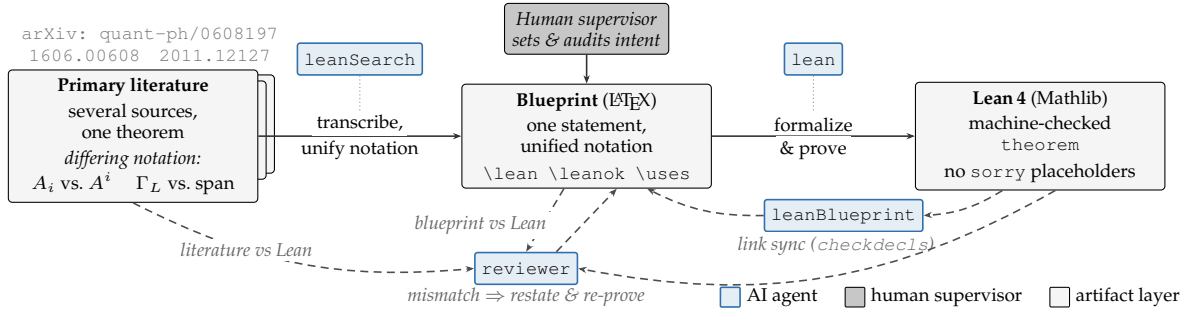
\begin{figure*}[t]
    \centering
    \begin{tikzpicture}[
            >={Stealth[length=3.2pt]},
            every node/.style={font=\footnotesize},
            layer/.style ={layerbox, rounded corners=1.5pt, align=center, inner sep=3pt,
                    font=\scriptsize, minimum height=1.3cm,
                    minimum width=3.3cm},
            human/.style ={humanbox, rounded corners=1.5pt, align=center,
                    font=\scriptsize\itshape, inner sep=3pt},
            agent/.style ={agentbox, rounded corners=1pt, align=center,
                    font=\scriptsize\ttfamily, inner sep=2.5pt, minimum height=0.42cm},
            fwd/.style   ={->, semithick, draw=black!80},
            aud/.style   ={->, semithick, draw=black!70, densely dashed},
            tick/.style  ={-, thin, draw=black!45, densely dotted},
            legsw/.style ={draw, minimum size=0.26cm, inner sep=0pt,
                    rounded corners=0.5pt},
            elbl/.style  ={font=\scriptsize, fill=white, inner sep=1.5pt},
            sub/.style   ={font=\scriptsize\itshape, text=black!60,
                    fill=white, inner sep=1.5pt},
            x=1cm, y=1cm
        ]

        %% --- Layer 1: primary literature (several sources) ----------------
        \node[layer] (litB) at (0.22,0.16) {};
        \node[layer] (litA) at (0.11,0.08) {};
        \node[layer] (lit)  at (0,0) {%
            \textbf{Primary literature}\\[1.5pt]
            several sources,\\ one theorem\\[2pt]
            \scriptsize\itshape differing notation:\\
            $A_i$ vs.\ $A^{i}$ \quad $\Gamma_L$ vs.\ span};
        \node[font=\scriptsize\ttfamily, text=black!55, align=center,
                above=1.5pt of litB.north]
        {arXiv:~quant-ph/0608197\\ 1606.00608\quad 2011.12127};

        %% --- Layer 2: blueprint -------------------------------------------
        \node[layer] (bp) at (6.0,0) {%
            \textbf{Blueprint} (\LaTeX)\\[1.5pt]
            one statement,\\ unified notation\\[2pt]
            {\ttfamily\textbackslash{}lean} \ {\ttfamily\textbackslash{}leanok} \ {\ttfamily\textbackslash{}uses}};

        %% --- Layer 3: Lean -------------------------------------------------
        \node[layer] (lean) at (12.0,0) {%
            \textbf{Lean\,4} (Mathlib)\\[1.5pt]
            machine-checked\\ {\ttfamily theorem}\\[2pt]
            \scriptsize no {\ttfamily sorry} placeholders};

        %% --- Forward construction (solid) ---------------------------------
        \draw[fwd] (lit.east) -- node[elbl, above]{transcribe,}
        node[elbl, below]{unify notation} (bp.west);
        \draw[fwd] (bp.east) -- node[elbl, above]{formalize}
        node[elbl, below]{\& prove} (lean.west);

        \node[agent] (scout) at (3.0,1.02) {leanSearch};
        \node[agent] (prove) at (9.0,1.02) {lean};
        % ticks on the background layer, so the white-filled edge labels
        % ("transcribe,"/"formalize") mask them instead of being overprinted
        \begin{pgfonlayer}{background}
            \draw[tick] (scout.south) -- (3.0,0.04);
            \draw[tick] (prove.south) -- (9.0,0.04);
        \end{pgfonlayer}

        \node[human] (human) at (6.0,1.42)
        {Human supervisor\\\scriptsize sets \& audits intent};
        \draw[fwd] (human) -- (bp.north);

        %% --- Audit loop (dashed) ------------------------------------------
        % leanBlueprint: mechanical blueprint <-> Lean link check
        \node[agent] (sync) at (9.4,-1.05) {leanBlueprint};
        \draw[aud] ([xshift=-5mm]lean.south) to[bend left=14] (sync.east);
        \draw[aud] (sync.west)  to[bend left=14] ([xshift=8mm]bp.south);
        \node[sub] at (9.3,-1.45) {link sync (\texttt{checkdecls})};

        % reviewer: two semantic checks (blueprint-vs-Lean, literature-vs-Lean)
        \node[agent] (rev) at (5.2,-1.75) {reviewer};
        \draw[aud] ([xshift=-3mm]bp.south) -- (rev.north);
        \draw[aud] ([xshift=2mm]lean.south) to[bend left=20] (rev.east);
        \draw[aud] (lit.south)  to[bend right=16] (rev.west);
        \node[sub] at (4.6,-1.18) {blueprint vs Lean};
        \node[sub] at (1.5,-1.50) {literature vs Lean};
        % return path: a mismatch sends the statement back to the blueprint,
        % drawn as an explicit reviewer -> blueprint arrow; the note sits
        % under the reviewer box, in open space
        \draw[aud] ([xshift=4mm]rev.north) -- ([xshift=4mm]bp.south);
        \node[sub, fill=none] at (5.2,-2.13)
        {mismatch $\Rightarrow$ restate \& re-prove};

        %% --- Legend: one row in the empty bottom-right (box fill encodes
        %%     role); flattened from a framed 3-row block to save height ----
        \node[legsw, agentbox] (lg1) at (7.9,-2.13) {};
        \node[anchor=west, font=\scriptsize, inner sep=2pt] (lg1t) at (lg1.east) {AI agent};
        \node[legsw, humanbox, right=8pt of lg1t] (lg2) {};
        \node[anchor=west, font=\scriptsize, inner sep=2pt] (lg2t) at (lg2.east) {human supervisor};
        \node[legsw, layerbox, right=8pt of lg2t] (lg3) {};
        \node[anchor=west, font=\scriptsize, inner sep=2pt] (lg3t) at (lg3.east) {artifact layer};

    \end{tikzpicture}
    \caption{%
        Maintaining consistency across the literature, Lean, and the blueprint.
        A theorem stated in several sources under
        different notations (the matrices written $A_i$~\cite{Perez-Garcia2007Matrix}
        or $A^{i}$~\cite{Cirac2017Matrixa,Cirac2021Matrix}, injectivity phrased
        through the map $\Gamma_L$~\cite{Perez-Garcia2007Matrix} or as a spanning
        condition~\cite{Cirac2017Matrixa}) is transcribed once into
        the blueprint, in a single notation and with machine-readable links to
        the Lean declaration, and then proved in Lean (solid arrows;
        \texttt{leanSearch} scouts Mathlib, \texttt{lean} closes the proof;
        dotted lines attach each agent to the step it carries out).
        Two agents keep the layers consistent (dashed). \texttt{leanBlueprint}
        checks the mechanical links (every \texttt{\textbackslash{}lean}
        reference resolves to a declaration and every
        \texttt{\textbackslash{}leanok} to a finished proof) via
        \texttt{checkdecls}. An automated reviewer checks the two semantic
        correspondences: blueprint versus Lean (the Lean theorem proves the
        stated result) and literature versus Lean (its hypotheses are no
        stronger than the cited source's). A mismatch returns the statement
        for restatement and re-proof.
    }\label{fig:alignment}
\end{figure*}

%======================================================================
\prlsection{The fundamental theorem and its formalization}%
%======================================================================
A translationally invariant matrix-product state (hereafter simply an MPS) is a family of quantum states on a one-dimensional chain~\cite{Verstraete2006Matrix}, parameterized by a single rank-3 tensor $A$ with $A = \sum_{i=0}^{d-1}A^{i}\ket{i}$ and $A^i\in M_D(\C)$ for each $i=0,1,\ldots,d-1$. Here $d$ is the local Hilbert-space dimension and $D$ is the bond dimension. An MPS on $N$ sites with periodic boundary conditions reads

\begin{equation}\label{eq:mps}
    \ket{\psi_N(A)} = \sum_{i_1,\ldots,i_N=0}^{d-1} \tr\!\left(A^{i_1} A^{i_2} \cdots A^{i_N}\right) \ket{i_1 \cdots i_N}.
\end{equation}
Eq.~\eqref{eq:mps} is not normalized as they are also relevant for studying matrix-product operators (see, e.g.,~\cite{Cirac2017Matrix,Liu2026Parent}), thus they are also called \textit{matrix-product vectors} (MPVs)~\cite{Cirac2017Matrixa}. To each tensor $A$ we associate a completely-positive (CP) map called the \textit{transfer operator} $E_A \colon M_D(\C) \to M_D(\C)$, defined by $E_A(X) = \sum_{i=0}^{d-1} A^{i} X {(A^{i})}^\dagger$, whose spectral information govern the properties of the MPV family.

Since $\ket{\psi_N(A)}$ depends only on traces of products of the matrices $A^i$, the map $A \mapsto \ket{\psi_N(A)}$ is many-to-one. A natural question is when two tensors generate MPV families that are proportional or equal to one another. The fundamental theorem of MPS states that a full characterization is possible after putting the tensors in canonical form~\cite{Perez-Garcia2007Matrix,Cirac2017Matrixa}.
We say that a tensor $A$ is \textit{normal} if its transfer operator $E_A$ has spectral radius $1$ and a unique eigenvalue $\lambda$ with $|\lambda|=1$ (all other eigenvalues have $|\lambda|<1$). Then $A$ is said to be in a canonical form~(CF) if $A^i = \bigoplus_{k=1}^r \mu_k A^i_k$, where $\mu_{k}\in\C$ and $A_k$ are normal tensors.
One of the central theorems we formalized is the following~\cite{Perez-Garcia2007Matrix,Cirac2017Matrixa,Cirac2021Matrix}:
\begin{theorem}[Fundamental theorem of MPS, equal case]\label{thm:bnt-equivalence}
    Let $A$ and $B$ be two tensors in CF\@. Then $A,B$ generate the same MPV family, $\ket{\psi_N(A)} = \ket{\psi_N(B)}$ for every $N\geq 1$, if and only if there exists an invertible matrix $X$ such that for all $i=0,1,\ldots,d-1$:
    \begin{equation}
        B^i=X A^i X^{-1}.\label{eq:equal-MPV-tensor}
    \end{equation}
\end{theorem}
\noindent Note that given any tensor $A$ generating an MPV $\ket{\psi_N(A)}$, it is always possible to bring it into a canonical form after blocking a sufficient number of sites to remove the so-called $p$-periodic subspaces~\cite{Cirac2017Matrixa,Fannes1992Finitely}. Thus our starting point is that the tensors are in CF\@.

We began the formalization effort by querying the orchestrator agent to prove FT-MPS without specifying any literature. Interestingly, in the early phase the system proved the special case of \cref{thm:bnt-equivalence} (when $A$ is \textit{injective}) through the \emph{Skolem-Noether theorem} (Sec.~\ref{app:ftsb_walkthrough} of Supplementary Material), a classical result about automorphisms of central simple algebras. This approach is not a standard in the literature and was chosen autonomously as the agents can directly make use of the existing Mathlib's ring-theory library~\cite{Supp}. This suggests that autoformalization is able to accommodate novel approaches based on the system's ability to make wide connections between different fields.

The injective case is the simplest instance of FT-MPS, as normal tensors only become injective after blocking sufficiently many sites%.
\footnote{The formalization follows the $(D^2-d+1)D^2 = O(D^4)$ blocking bound of the quantum Wielandt inequality~\cite{Sanz2010Quantum}; whether this can be improved to the optimal scaling remains an interesting open question~\cite{Michalek2019Quantum,Shitov2023Growth}.}.
To prove \cref{thm:bnt-equivalence}, we need to first put the tensors into CF after sufficient \emph{blocking} (\cref{fig:blocking}) and then find a `minimal' collection of normal tensors $A_k$, called \textit{basis of normal tensors}, or BNT, ${\{A_k^i\}}_{k=1}^g$, that always exists for any tensor and accounts for repeated blocks in $A$. The global gauge transformation $X$ can then be constructed by proving the following result: two tensors $A,B$ generating the same MPV must have the same number of BNT elements and, up to permutations, the BNTs are related by gauge transformation $B_k=X_k A_{\pi(k)}X_k^{-1}$ where $\pi$ is some permutation.

Our earliest attempt used only the review paper~\cite{Cirac2021Matrix}, which proved insufficient to provide the full proof for \cref{thm:bnt-equivalence} fully autonomously. Thus we supplied the agents with the arXiv preprint \LaTeX{} source, rather than the typeset PDF, of the primary references~\cite{Perez-Garcia2007Matrix,Cirac2017Matrixa}, together with the review~\cite{Cirac2021Matrix} and the most relevant quantum-information sources~\cite{Wolf2012Quantum,Sanz2010Quantum}; the system regenerated its proof strategy.
This regenerated strategy made clear that the task was not just a short linear-algebra formalization. FT-MPS depends on a substantial quantum-information infrastructure around transfer operators, such as CP maps, quantum Perron-Frobenius theory, and the quantum Wielandt bound~\footnote{After the formalization for the fundamental theorem has been completed, the positive and completely positive maps between $C^*$-algebras were added to Mathlib 4.31; TNLean has since been updated to build on them, where it had previously constructed these notions itself.}, itself the subject of a separate paper~\cite{Sanz2010Quantum}.
As these dependencies grew, the development had to be split into smaller proof files and bounded agent tasks, with the blueprint and persistent memory tracking the global structure; the detailed dependency graph and workflow are given in the Supplementary Material~\cite{Supp}.

% fig_blocking.tex  --  Blocking transformation A -> A^{(L)} for FT-MPS
% Requires:  \usepackage{tikz}
%            \usetikzlibrary{calc}
%
% Visualises the coarse-graining step that takes a Matrix Product State
% tensor A = {A^i}_{i=1}^d on Mat_D(C) to the L-blocked tensor A^{(L)}.
% The diagram style follows the tensor-network conventions of the project
% blueprint (small filled circles for tensor dots, thin lines for bonds
% and physical legs, a thin rectangle marking the blocked region).

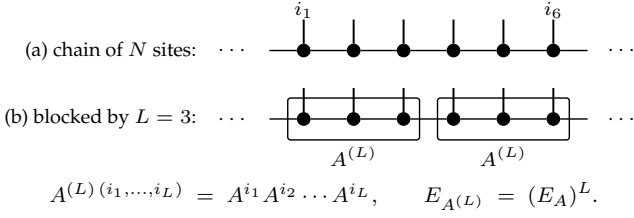
\begin{figure}[t]
    \centering
    \resizebox{\columnwidth}{!}{%
    \begin{tikzpicture}[
            baseline=-0.1cm,
            transform shape,
            font=\scriptsize,
            tn dot/.style   ={draw, fill, shape=circle, inner sep=0.055cm},
            tn bond/.style  ={line width=0.5pt},
            tn phys/.style  ={line width=0.7pt},
            tn block/.style ={line width=0.5pt, rounded corners=1pt},
            x=0.62cm
        ]
        %% --- Panel (a): the original chain of N copies of A -------------------
        \begin{scope}[shift={(0,0)}]
            \foreach \i in {0,...,5} {
                    \node[tn dot] (A\i) at (\i,0) {};
                }
            \foreach \i [evaluate=\i as \j using int(\i+1)] in {0,...,4} {
                    \draw[tn bond] (A\i.east) -- (A\j.west);
                }
            \foreach \i in {0,...,5} {
                    \draw[tn phys] (A\i.north) -- ++(0,0.30);
                }
            \draw[tn bond] (A5.east) -- ++(0.55,0);
            \node           at (6.40,0)    {$\cdots$};
            \draw[tn bond] (A0.west) -- ++(-0.55,0);
            \node           at (-1.40,0)   {$\cdots$};
            \node           at (0,0.50)    {$i_1$};
            \node           at (5,0.50)    {$i_6$};
            \node[anchor=east] at (-1.95,0) {(a)\ chain of $N$ sites:};
        \end{scope}
        %% --- Panel (b): the same chain, grouped into L-site blocks ------------
        \begin{scope}[shift={(0,-0.85)}]
            \foreach \i in {0,...,5} {
                    \node[tn dot] (B\i) at (\i,0) {};
                }
            \foreach \i [evaluate=\i as \j using int(\i+1)] in {0,...,4} {
                    \draw[tn bond] (B\i.east) -- (B\j.west);
                }
            \foreach \i in {0,...,5} {
                    \draw[tn phys] (B\i.north) -- ++(0,0.30);
                }
            \draw[tn bond] (B5.east) -- ++(0.55,0);
            \node           at (6.40,0)    {$\cdots$};
            \draw[tn bond] (B0.west) -- ++(-0.55,0);
            \node           at (-1.40,0)   {$\cdots$};
            \foreach \g in {0,1} {
                    \pgfmathsetmacro{\xL}{\g*3 - 0.32}
                    \pgfmathsetmacro{\xR}{\g*3 + 2.32}
                    \pgfmathsetmacro{\xC}{\g*3 + 1}
                    \draw[tn block] (\xL,-0.26) rectangle (\xR,0.26);
                    \node          at (\xC,-0.46) {$A^{(L)}$};
                }
            \node[anchor=east] at (-1.95,0) {(b)\ blocked by $L=3$:};
        \end{scope}
    \end{tikzpicture}%
    }
    \\[0.15em]
    {\footnotesize
        $\displaystyle
            A^{(L)\,(i_1,\ldots,i_L)} \;=\; A^{i_1} A^{i_2} \cdots A^{i_L},
            \qquad
            E_{A^{(L)}} \;=\; {(E_A)}^L$.}
    \caption{Blocking in tensor networks.
        Given an MPS tensor $A$, grouping $L$ consecutive sites yields a single coarse-grained tensor $A^{(L)}$ whose physical index is the composite index $(i_1,\ldots,i_L)$ and whose transfer operator is the $L$-th power ${(E_A)}^L$ of the transfer operator $E_A$ of $A$.
        The spectrum of ${(E_A)}^L$ concentrates onto its dominant eigenvalues as $L$ grows: each block of the resulting canonical form becomes
        \emph{normal} after a finite blocking length and \emph{injective}
        (the matrices of block $k$ span the full algebra
        $M_{D_k}(\C)$)
        after $L = O(D^4)$ further sites, by the quantum Wielandt bound.
        The agent pool reached the single-block injective conjugation autonomously (\cref{eq:equal-MPV-tensor}); the extension to the multi-block regime indicated here was recovered only after the primary references~\cite{Perez-Garcia2007Matrix,Cirac2017Matrixa} were supplied.}\label{fig:blocking}
\end{figure}

This splitting was carried out by a team of specialized agents, coordinated by an orchestrator that assigns tasks to sub-agents and combines their results; the agent roles, tools, and coordination patterns are described in the End Matter.
As the project grew, it also became important to periodically refactor the code, extracting reusable proof strategies and retiring unsuccessful attempts, to keep the agents from getting stuck; this is the responsibility of the simplifier agent.

The project produced a formalization codebase, accompanied by a \nChapters{}-chapter blueprint spanning \nPagesBluePrint pages (see Supplementary Material~\cite{Supp}).
Every step of the FT-MPS proof chain was verified, with no placeholders (Lean's \texttt{sorry} markers for unfinished steps) remaining in the core argument.
Statements occupying a few lines in the literature can require several hundred lines of Lean~\cite{Supp}.
In total, the recorded model-API cost of the interactive sessions was \$\nCost{} for the whole library (an estimated \$\nCostFT{} for the FT-MPS alone, scaled by code size), chiefly
on proof writing and orchestration (see End Matter).

%======================================================================
\prlsection{Catching unintended formalizations}%
%======================================================================
There are at least two aspects of the autonomous Lean formalization process that require some care. First, the agents may only formalize a statement whose proof can pass the Lean kernel but deviates from the intended one, e.g., by weakening the theorem, as we discuss below, or by proving it vacuously from hypotheses that are never satisfied.
Second, physics arguments rely on many conventions and implicit assumptions, which must be made explicit under rigorous formalization.
While the second problem can many times be dealt with autonomously, at present the first problem requires more human intervention in the form of a review.

The blueprint is a human-readable mathematical specification that links definitions, theorem statements, and dependency structure to the underlying Lean code, since auditing all of the Lean proofs directly is impractical for a human reviewer.
We performed \nReviewRounds{} rounds of review of the blueprint for the FT-MPS formalization~\cite{TNLeanBlueprint}, from which reports were generated. These reports were fed back to the orchestrator agent, which dispatched suitable sub-agents to correct the issues they identified.
We emphasize that the human intervention is \textit{strategic} rather than \textit{tactical}: the human supervisor pointed out the unintended version of the statement being proved, but not \textit{how} it should be formalized by the agents. The agents then re-derived the intended statement and removed the incorrect or outdated versions (somewhat reluctantly), without further guidance. Below we mention some instances where this intervention was required.

\paragraph*{Weakening of the theorem.}%
The most expensive corrections were those in which an unintended hypothesis entered early and many later lemmas were built on top of it.
In the first formalization round, the agents attempted to prove FT-MPS under what they called a \textit{doubly stochastic (DS) ``gauge''}. In MPS theory, one may normalize the transfer operator $E_A$ to have spectral radius $1$. One can then choose a right-canonical form, in which $E_A$ is unital, or a left-canonical form, in which $E_A$ is trace-preserving. For a generic MPS tensor, however, one cannot impose both conditions simultaneously. The DS assumption therefore gave a formally correct theorem which is easier to prove and formalize, but only for a much smaller class of tensors. Removing this extra assumption required reorganizing many earlier arguments around one-sided canonical forms.
This example shows that Lean checks the proof but not the choice of statement: the theorem proved was correct, only narrower than the one intended and do not support the later proof of the full theorem.
After the corrected statement was specified, the agents carried out this restructuring autonomously.

\paragraph*{Finite vs asymptotic limit.}%
A second class of issues arose when the autoformalization replaced the intended definitions with a superficially related statement that is not merely a weakening, but may not hold in the setting of interest. In our case, the finite statement of FT-MPS was replaced by an asymptotic one. The vectors in Eq.~\eqref{eq:mps} are unnormalized MPVs: their coefficients are traces of products of tensor matrices at a fixed length $N$. The theorem compares these vectors for every finite $N$, either by equality or by nonzero proportionality. It does not compare normalized states, nor does it assume that the norms of the MPVs have any limiting behavior as $N\to\infty$. Early formalization attempts based on asymptotic norm comparisons therefore proved statements with the wrong hypotheses and failed to recover the needed finite information. The same finite issue appears when comparing multiplicities of canonical blocks, where coefficients at finite length can oscillate rather than converge.
Here the team could not patch the proof directly and only found the correct route after being reminded by us after reading the blueprint that the underlying asymptotic definition, rather than the proof strategy, was at fault.
The final proof keeps the comparison at fixed finite $N$, matching the statement of FT-MPS\@.

\paragraph*{Edge cases.}%
Some corrections were mathematically minor but still costly because proof assistants require every edge case to be stated explicitly. The physically relevant MPS statements always have positive bond dimension and nonempty chains.
If this is not included in the theorem statement, Lean may have to consider artificial cases such as $D=0$ or $N=0$. The case $D=0$ carries no physical state, while for $N=0$ the empty-word convention produces the coefficient $\tr(\openone_D)=D$, which is not part of the usual FT-MPS comparison. Adding the intended assumptions $D\geq 1$ and $N\geq 1$ removed large irrelevant branches of the proof.

\paragraph*{Lean code-Blueprint alignment.}%
Several reviews found discrepancies between the blueprint and the Lean declarations: in some cases the Lean statement was correct, while the blueprint prose had drifted across successive edits. We then introduced explicit synchronization and review steps to keep the blueprint aligned with the formal statements (see~\cref{fig:alignment}). We also adjusted the prompts and review criteria so that the blueprint followed the source papers more closely and uses reader-friendly prose to facilitate auditing \cite{Supp}.

% fig_spt.tex -- From the physical symmetry assumption to the SPT cohomology invariant.
% Requires: \usepackage{tikz}; \usetikzlibrary{arrows.meta, positioning, calc}.
% Top tier: the physical-layer assumption -- applying the on-site symmetry
%   U(g) to the physical leg of A (defining \tilde A_g) is assumed to generate
%   the same MPV family as A itself ("Sym").
% Middle tier: FT-MPS turns that physical-layer assumption into the
%   virtual-layer gauge identity \tilde A_g^i = X(g) A^i X(g)^{-1}.
%   Tensor-network conventions follow fig_blocking (small filled circles for
%   tensors, thin lines for bonds and physical legs).
% Bottom tier: the chain turning that gauge into a topological invariant.
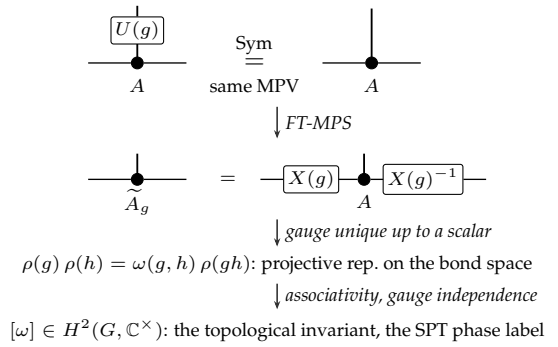
\begin{figure}[t]
    \centering
    \begin{tikzpicture}[
            font=\scriptsize,
            x=1cm, y=1cm,
            >={Stealth[length=3.2pt]},
            tn dot/.style  ={draw, fill, shape=circle, inner sep=0.055cm},
            tn bond/.style ={line width=0.5pt},
            tn phys/.style ={line width=0.7pt},
            gbox/.style    ={draw, fill=white, rounded corners=1pt, inner sep=1.6pt},
            arr/.style     ={->, semithick, draw=black!80},
            albl/.style    ={font=\scriptsize\itshape, inner sep=1.5pt}
        ]
        %% ---- top tier: physical-layer assumption -----------------------------
        % LHS: tensor A with U(g) on the physical leg, defining \tilde A_g
        \node[tn dot] (P1) at (0,2.62) {};
        \node at (0,2.32) {$A$};
        \draw[tn bond] (P1) -- ++(-0.65,0);
        \draw[tn bond] (P1) -- ++(0.65,0);
        \draw[tn phys] (0,2.68) -- (0,2.88);
        \node[gbox] (Pg) at (0,3.05) {$U(g)$};
        \draw[tn phys] (Pg.north) -- ++(0,0.14);

        \node[align=center] at (1.55,2.62) {\normalsize$\overset{\text{Sym}}{=}$\\[1pt]\scriptsize same MPV};

        % RHS: plain tensor A, for comparison
        \node[tn dot] (P2) at (3.1,2.62) {};
        \node at (3.1,2.32) {$A$};
        \draw[tn bond] (P2) -- ++(-0.65,0);
        \draw[tn bond] (P2) -- ++(0.65,0);
        \draw[tn phys] (3.1,2.68) -- (3.1,3.34);

        % labels sit beside the arrows (not on them, so no white box occludes
        % the shaft or arrowhead)
        \draw[arr] (1.85,2.04) -- node[albl, right=1.5pt]{FT-MPS} (1.85,1.62);

        %% ---- middle tier: virtual gauge identity  \tilde A_g^i = X(g) A^i X(g)^{-1} ----
        \node[tn dot] (L) at (0,1.08) {};
        \node at (0,0.78) {$\tilde A_g$};
        \draw[tn bond] (L) -- ++(-0.65,0);
        \draw[tn bond] (L) -- ++(0.65,0);
        \draw[tn phys] (0,1.14) -- (0,1.42);

        \node at (1.2,1.08) {$=$};

        \node[tn dot] (R) at (3,1.08) {};
        \node at (3,0.78) {$A$};
        \draw[tn phys] (3,1.14) -- (3,1.42);
        \draw[tn bond] (R) -- ++(-0.38,0);
        \node[gbox] (Xg) at (2.30,1.08) {$X(g)$};
        \draw[tn bond] (Xg.west) -- ++(-0.30,0);
        \draw[tn bond] (R) -- ++(0.38,0);
        \node[gbox] (Xgi) at (3.78,1.08) {$X(g)^{-1}$};
        \draw[tn bond] (Xgi.east) -- ++(0.30,0);

        %% ---- bottom tier: gauge -> projective rep -> cohomology class ---------
        % single text lines (no boxes), arrows short with side labels: the same
        % content at less than half the former height
        \draw[arr] (1.85,0.55) -- node[albl, right=1.5pt]{gauge unique up to a scalar} (1.85,0.19);
        \node (proj) at (1.85,-0.05)
            {$\rho(g)\,\rho(h) = \omega(g,h)\,\rho(gh)$: projective rep.\ on the bond space};
        \draw[arr] (1.85,-0.30) -- node[albl, right=1.5pt]{associativity, gauge independence} (1.85,-0.64);
        \node (coh) at (1.85,-0.91)
            {$[\omega]\in H^2(G,\C^\times)$: the topological invariant, the SPT phase label};
    \end{tikzpicture}
    \caption{From the fundamental theorem to the classification of symmetry protected topological phases.
        \emph{Top}: the on-site symmetry $U(g)$ acting on the physical leg of
        an injective MPS tensor $A$ defines $\tilde A_g$; the physical-layer
        assumption is that $\tilde A_g$ generates the same MPV family as $A$
        itself.
        \emph{Middle}: FT-MPS turns this assumption into the virtual-layer
        gauge identity $\tilde A_g^i = X(g)\,A^i\,{X(g)}^{-1}$
        (\cref{thm:spt-invariant}; physical legs drawn upward, bond legs
        horizontal, as in \cref{fig:blocking}).
        \emph{Bottom}: $X(g)$ is unique up to a scalar, so the gauges form a
        projective representation $\rho$ with 2-cocycle $\omega$; its class
        $[\omega]\in H^2(G,\C^\times)$ is the gauge invariant labeling the
        symmetry-protected topological (SPT) phase, $U(1)$-valued after the
        unitary normalization.
        }\label{fig:spt}
\end{figure}%

%======================================================================
\prlsection{Physical application}%
%======================================================================
FT-MPS provides the mechanism behind virtual symmetry actions in 1D tensor networks and the cohomology invariant used in the classification of 1D bosonic SPT phases~\cite{Chen2011Classification,Schuch2011Classifying,Pollmann2010Entanglement,Pollmann2012Symmetry,Perez-Garcia2008String} (\cref{fig:spt}). After the formalization of FT-MPS, we fed the agents with the following article \cite{Perez-Garcia2008String} about string order and symmetries in spin chains, and they autonomously pursued the formalization that led to the following result.

\begin{theorem}[A cohomological invariant from injective symmetric MPS]\label{thm:spt-invariant}
    Let $A$ be an injective MPS tensor of bond dimension $D\geq 1$, and let
    $U:G\to GL_d(\C)$ be an on-site linear representation of a finite group
    $G$. Define $\tilde A_g^i \coloneqq \sum_j {U(g)}_{ij}A^j$.
    Suppose that $A$ is invariant under $U$, in the sense that
    $\ket{\psi_N(A)}=\ket{\psi_N(\tilde A_g)}$ for every $g\in G$ and every
    $N\geq 1$. Then for each $g\in G$ there is an invertible matrix
    $X(g)\in \mathrm{GL}_D(\mathbb C)$, unique up to a nonzero scalar, such that
    \begin{equation}
        \tilde A_g^i = X(g)\,A^i\,{X(g)}^{-1}
    \end{equation}
    for all physical indices $i$. With the convention $\rho(g)\coloneqq X(g^{-1})$,
    the maps $\rho(g)$ form a projective representation of $G$: $\rho(g)\rho(h)=\omega(g,h)\rho(gh)$, $\omega(g,h)\in\mathbb C^\times$.
    The cohomology class $[\omega]\in H^2(G,\mathbb C^\times)$ is unchanged
    by rescaling the gauges $X(g)$, and therefore depends only on the
    symmetric tensor $A$ together with the on-site action $U$.
\end{theorem}
\noindent We note that the autoformalized theorem is proven for a linear representation $U:G\to {GL}_d(\C)$, since the formal argument only uses the equality of MPV families and does not require unitarity. In practice, the physical on-site symmetry is taken to be a unitary representation $U:G\to\mathrm{U}(d)$, in which case the virtual gauges can be chosen unitary up to phase. Therefore,  $\omega$ can then be taken to be $U(1)$-valued and the standard invariant $[\omega]\in H^2(G,U(1))$ is recovered. The full SPT classification also involves symmetric parent Hamiltonians and gapped paths; while we do not formalize the full classification here, the TNLean library and existing Mathlib provide very natural starting points to complete the autoformalization in this setting.

%======================================================================
\prlsection{Discussion and outlook}%
%======================================================================
In this work we have built and presented an agent-driven workflow for research-level formalization in theoretical physics, with the autoformalization of the FT-MPS to demonstrate the system's capabilities. The agents are coordinated through a shared mathematical `blueprint' intermittent human review. The core development consists of approximately \nLOC{} lines of code, \nDecls{} declarations, and \nFiles{} files, within a broader Lean~4 tensor-network library of about \nLOCfull{} lines. Building on Mathlib's foundation~\cite{2026Leanprovercommunity,TheMathlibCommunity2020Lean}, TNLean develops verified statements within tensor-network theory as well as quantum information theory required for the FT-MPS, e.g., completely positive maps and quantum Perron-Frobenius theory. We observed that enforcement of the correct mathematical intent is the main bottleneck in large-scale autoformalization that requires us to navigate several technical complications.

These observations clarify why the blueprint and review steps were not auxiliary to the proof search. In a project of this size, correctness depends not only on proving lemmas, but on controlling how the target statement is formulated. A single agent session cannot reliably carry the source papers, the current Lean library, the proof state, and the intended theorem statement at once (see End Matter). It is therefore necessary for the project to be split into separate proof tasks, and the blueprint is used to record the global mathematical specification (e.g., the canonical-form assumptions, the normalization conventions, and consistency with the literature). Between sessions, the persistent memory carried the corrected conventions and the rejected proof routes. The review stage checks whether the formal declarations still expressed the intended theorem faithfully before the Lean kernel checks the proof. This is the setting in which Buzzard's warning that ``definitions are more dangerous than proofs'' is most concrete~\cite{Buzzard2025Formal}.

Along with TNLean library, we release the blueprint, review reports, and perhaps most importantly, ``technique notes'' produced during the project~\cite{TNLean,TNLeanBlueprint}. These materials document the choices that are not visible from the final Lean files alone and facilitate audits: which hypotheses were required, which plausible formulations had to be rejected, and where the formal statements differ from the informal presentation in the literature. We also release a curated selection of the distilled memory files used during the project. We note that the distilled memory files and the agent role specifications are not specific to tensor network formalizations: they record project-independent formalization techniques and conventions, and can seed other autoformalization targets.

We expect that a reusable formal library (like ours), together with the multi-agent workflow that we developed, will help in both the formalization of existing modern results, such as the low individual degree test and the quantum-complexity theorem $\mathrm{MIP}^*=\mathrm{RE}$~\cite{Ji2020LowIndividualDegree,LDTLeanPaper,Ji2021Mip}, as well as attacking conjectures in quantum many-body physics and quantum information and computation~\cite{Meiburg2024LeanQuantumInfo,Meiburg2025Formalization,Ren2026MerLean,Li2026MerLeanprover,Ehatamm2026Endtoend}. Within tensor-network theory itself, the current TNLean also supports broader formalization targets: FT-MPS with boundaries~\cite{Florido-Llinas2025Uniform}, FT for a subclass of higher-dimensional tensor-network states~\cite{Molnar2018Normal}, parent Hamiltonians~\cite{Garre-Rubio2025MPS,Schuch2025Simple}, matrix-product operators, renormalization fixed points~\cite{Cirac2017Matrixa}, and classification of phases of matter~\cite{Chen2011Classification,Schuch2011Classifying,Pollmann2012Symmetry,Chen2012Symmetryprotected}. We leave these further developments of TNLean for future work.\\

\begin{acknowledgments}
    \paragraph*{Data availability}%
    The Lean~4 source code, the formalization blueprint, and the review reports and technique notes that support this study are openly available in the TNLean repository~\cite{TNLean,TNLeanBlueprint}. The multi-agent system was built within the TeXRA software~\cite{TeXRA}, through which it can be used; the system prompts of the five interactive agent (see End Matter) are reproduced in the Supplementary Material~\cite{Supp}, and the reviewer's prompts are released with TNLean.

    \paragraph*{AI Disclosure}
    The Lean~4 code and the blueprint have been fully generated by the multi-agent AI system we developed under our supervision.
    We used LLM-based tools for editing parts of the manuscript text and figures;
    all scientific content was checked by the authors, who take responsibility for it.

    \paragraph*{Acknowledgments}%
    We thank X.-L. Qi for helpful discussion.
    E.T.\ acknowledges support from the Alexander von Humboldt Foundation.
    The work is partially supported by the Deutsche Forschungsgemeinschaft (DFG, German Research Foundation) under Germany's Excellence Strategy -- EXC-2111 -- 390814868.
    This research is part of the Munich Quantum Valley, which is supported by the Bavarian state government with funds from the Hightech Agenda Bayern Plus.
\end{acknowledgments}

\FloatBarrier
\bibliography{additional_references,library}

\begin{thebibliography}{67}%
\makeatletter
\providecommand \@ifxundefined [1]{%
 \@ifx{#1\undefined}
}%
\providecommand \@ifnum [1]{%
 \ifnum #1\expandafter \@firstoftwo
 \else \expandafter \@secondoftwo
 \fi
}%
\providecommand \@ifx [1]{%
 \ifx #1\expandafter \@firstoftwo
 \else \expandafter \@secondoftwo
 \fi
}%
\providecommand \natexlab [1]{#1}%
\providecommand \enquote  [1]{``#1''}%
\providecommand \bibnamefont  [1]{#1}%
\providecommand \bibfnamefont [1]{#1}%
\providecommand \citenamefont [1]{#1}%
\providecommand \href@noop [0]{\@secondoftwo}%
\providecommand \href [0]{\begingroup \@sanitize@url \@href}%
\providecommand \@href[1]{\@@startlink{#1}\@@href}%
\providecommand \@@href[1]{\endgroup#1\@@endlink}%
\providecommand \@sanitize@url [0]{\catcode `\\12\catcode `\$12\catcode
  `\&12\catcode `\#12\catcode `\^12\catcode `\_12\catcode `\%12\relax}%
\providecommand \@@startlink[1]{}%
\providecommand \@@endlink[0]{}%
\providecommand \url  [0]{\begingroup\@sanitize@url \@url }%
\providecommand \@url [1]{\endgroup\@href {#1}{\urlprefix }}%
\providecommand \urlprefix  [0]{URL }%
\providecommand \Eprint [0]{\href }%
\providecommand \doibase [0]{https://doi.org/}%
\providecommand \selectlanguage [0]{\@gobble}%
\providecommand \bibinfo  [0]{\@secondoftwo}%
\providecommand \bibfield  [0]{\@secondoftwo}%
\providecommand \translation [1]{[#1]}%
\providecommand \BibitemOpen [0]{}%
\providecommand \bibitemStop [0]{}%
\providecommand \bibitemNoStop [0]{.\EOS\space}%
\providecommand \EOS [0]{\spacefactor3000\relax}%
\providecommand \BibitemShut  [1]{\csname bibitem#1\endcsname}%
\let\auto@bib@innerbib\@empty
%</preamble>
\bibitem [{\citenamefont {Gonthier}(2008)}]{Gonthier2008Formal}%
  \BibitemOpen
  \bibfield  {author} {\bibinfo {author} {\bibfnamefont {G.}~\bibnamefont
  {Gonthier}},\ }\bibfield  {title} {\bibinfo {title} {Formal proof---the
  four-color theorem},\ }\href@noop {} {\bibfield  {journal} {\bibinfo
  {journal} {Notices Amer. Math. Soc.}\ }\textbf {\bibinfo {volume} {55}},\
  \bibinfo {pages} {1382} (\bibinfo {year} {2008})}\BibitemShut {NoStop}%
\bibitem [{\citenamefont {Hales}\ \emph {et~al.}(2017)\citenamefont {Hales},
  \citenamefont {Adams}, \citenamefont {Bauer}, \citenamefont {Dang},
  \citenamefont {Harrison}, \citenamefont {Hoang}, \citenamefont {Kaliszyk},
  \citenamefont {Magron}, \citenamefont {Mclaughlin}, \citenamefont {Nguyen},
  \citenamefont {Nguyen}, \citenamefont {Nipkow}, \citenamefont {Obua},
  \citenamefont {Pleso}, \citenamefont {Rute}, \citenamefont {Solovyev},
  \citenamefont {Ta}, \citenamefont {Tran}, \citenamefont {Trieu},
  \citenamefont {Urban}, \citenamefont {Vu},\ and\ \citenamefont
  {Zumkeller}}]{Hales2017Formal}%
  \BibitemOpen
  \bibfield  {author} {\bibinfo {author} {\bibfnamefont {T.}~\bibnamefont
  {Hales}}, \bibinfo {author} {\bibfnamefont {M.}~\bibnamefont {Adams}},
  \bibinfo {author} {\bibfnamefont {G.}~\bibnamefont {Bauer}}, \bibinfo
  {author} {\bibfnamefont {T.~D.}\ \bibnamefont {Dang}}, \bibinfo {author}
  {\bibfnamefont {J.}~\bibnamefont {Harrison}}, \bibinfo {author}
  {\bibfnamefont {L.~T.}\ \bibnamefont {Hoang}}, \bibinfo {author}
  {\bibfnamefont {C.}~\bibnamefont {Kaliszyk}}, \bibinfo {author}
  {\bibfnamefont {V.}~\bibnamefont {Magron}}, \bibinfo {author} {\bibfnamefont
  {S.}~\bibnamefont {Mclaughlin}}, \bibinfo {author} {\bibfnamefont {T.~T.}\
  \bibnamefont {Nguyen}}, \bibinfo {author} {\bibfnamefont {Q.~T.}\
  \bibnamefont {Nguyen}}, \bibinfo {author} {\bibfnamefont {T.}~\bibnamefont
  {Nipkow}}, \bibinfo {author} {\bibfnamefont {S.}~\bibnamefont {Obua}},
  \bibinfo {author} {\bibfnamefont {J.}~\bibnamefont {Pleso}}, \bibinfo
  {author} {\bibfnamefont {J.}~\bibnamefont {Rute}}, \bibinfo {author}
  {\bibfnamefont {A.}~\bibnamefont {Solovyev}}, \bibinfo {author}
  {\bibfnamefont {T.~H.~A.}\ \bibnamefont {Ta}}, \bibinfo {author}
  {\bibfnamefont {N.~T.}\ \bibnamefont {Tran}}, \bibinfo {author}
  {\bibfnamefont {T.~D.}\ \bibnamefont {Trieu}}, \bibinfo {author}
  {\bibfnamefont {J.}~\bibnamefont {Urban}}, \bibinfo {author} {\bibfnamefont
  {K.}~\bibnamefont {Vu}},\ and\ \bibinfo {author} {\bibfnamefont
  {R.}~\bibnamefont {Zumkeller}},\ }\bibfield  {title} {\bibinfo {title} {A
  formal proof of the {{Kepler}} conjecture},\ }\href
  {https://doi.org/10.1017/fmp.2017.1} {\bibfield  {journal} {\bibinfo
  {journal} {Forum Math. Pi}\ }\textbf {\bibinfo {volume} {5}},\ \bibinfo
  {pages} {e2} (\bibinfo {year} {2017})}\BibitemShut {NoStop}%
\bibitem [{\citenamefont {Moura}\ and\ \citenamefont
  {Ullrich}(2021)}]{Moura2021Lean}%
  \BibitemOpen
  \bibfield  {author} {\bibinfo {author} {\bibfnamefont {L.~D.}\ \bibnamefont
  {Moura}}\ and\ \bibinfo {author} {\bibfnamefont {S.}~\bibnamefont
  {Ullrich}},\ }\bibfield  {title} {\bibinfo {title} {The lean 4 theorem prover
  and programming language},\ }in\ \href
  {https://doi.org/10.1007/978-3-030-79876-5_37} {\emph {\bibinfo {booktitle}
  {Automated Deduction -- {{CADE}} 28}}},\ Vol.\ \bibinfo {volume} {12699},\
  \bibinfo {editor} {edited by\ \bibinfo {editor} {\bibfnamefont
  {A.}~\bibnamefont {Platzer}}\ and\ \bibinfo {editor} {\bibfnamefont
  {G.}~\bibnamefont {Sutcliffe}}}\ (\bibinfo  {publisher} {Springer
  International Publishing},\ \bibinfo {address} {Cham},\ \bibinfo {year}
  {2021})\ pp.\ \bibinfo {pages} {625--635}\BibitemShut {NoStop}%
\bibitem [{\citenamefont {Hubert}\ \emph {et~al.}(2026)\citenamefont {Hubert},
  \citenamefont {Mehta}, \citenamefont {Sartran}, \citenamefont {Horv{\'a}th},
  \citenamefont {{\v Z}u{\v z}i{\'c}}, \citenamefont {Wieser}, \citenamefont
  {Huang}, \citenamefont {Schrittwieser}, \citenamefont {Schroecker},
  \citenamefont {Masoom}, \citenamefont {Bertolli}, \citenamefont {Zahavy},
  \citenamefont {Mandhane}, \citenamefont {Yung}, \citenamefont {Beloshapka},
  \citenamefont {Ibarz}, \citenamefont {Veeriah}, \citenamefont {Yu},
  \citenamefont {Nash}, \citenamefont {Lezeau}, \citenamefont {Mercuri},
  \citenamefont {S{\"o}nne}, \citenamefont {Mehta}, \citenamefont {Davies},
  \citenamefont {Zheng}, \citenamefont {Pedregosa}, \citenamefont {Li},
  \citenamefont {Von~Glehn}, \citenamefont {Rowland}, \citenamefont {Albanie},
  \citenamefont {Velingker}, \citenamefont {Schmitt}, \citenamefont {Lockhart},
  \citenamefont {Hughes}, \citenamefont {Michalewski}, \citenamefont
  {Sonnerat}, \citenamefont {Hassabis}, \citenamefont {Kohli},\ and\
  \citenamefont {Silver}}]{Hubert2026Olympiadlevel}%
  \BibitemOpen
  \bibfield  {author} {\bibinfo {author} {\bibfnamefont {T.}~\bibnamefont
  {Hubert}}, \bibinfo {author} {\bibfnamefont {R.}~\bibnamefont {Mehta}},
  \bibinfo {author} {\bibfnamefont {L.}~\bibnamefont {Sartran}}, \bibinfo
  {author} {\bibfnamefont {M.~Z.}\ \bibnamefont {Horv{\'a}th}}, \bibinfo
  {author} {\bibfnamefont {G.}~\bibnamefont {{\v Z}u{\v z}i{\'c}}}, \bibinfo
  {author} {\bibfnamefont {E.}~\bibnamefont {Wieser}}, \bibinfo {author}
  {\bibfnamefont {A.}~\bibnamefont {Huang}}, \bibinfo {author} {\bibfnamefont
  {J.}~\bibnamefont {Schrittwieser}}, \bibinfo {author} {\bibfnamefont
  {Y.}~\bibnamefont {Schroecker}}, \bibinfo {author} {\bibfnamefont
  {H.}~\bibnamefont {Masoom}}, \bibinfo {author} {\bibfnamefont
  {O.}~\bibnamefont {Bertolli}}, \bibinfo {author} {\bibfnamefont
  {T.}~\bibnamefont {Zahavy}}, \bibinfo {author} {\bibfnamefont
  {A.}~\bibnamefont {Mandhane}}, \bibinfo {author} {\bibfnamefont
  {J.}~\bibnamefont {Yung}}, \bibinfo {author} {\bibfnamefont {I.}~\bibnamefont
  {Beloshapka}}, \bibinfo {author} {\bibfnamefont {B.}~\bibnamefont {Ibarz}},
  \bibinfo {author} {\bibfnamefont {V.}~\bibnamefont {Veeriah}}, \bibinfo
  {author} {\bibfnamefont {L.}~\bibnamefont {Yu}}, \bibinfo {author}
  {\bibfnamefont {O.}~\bibnamefont {Nash}}, \bibinfo {author} {\bibfnamefont
  {P.}~\bibnamefont {Lezeau}}, \bibinfo {author} {\bibfnamefont
  {S.}~\bibnamefont {Mercuri}}, \bibinfo {author} {\bibfnamefont
  {C.}~\bibnamefont {S{\"o}nne}}, \bibinfo {author} {\bibfnamefont
  {B.}~\bibnamefont {Mehta}}, \bibinfo {author} {\bibfnamefont
  {A.}~\bibnamefont {Davies}}, \bibinfo {author} {\bibfnamefont
  {D.}~\bibnamefont {Zheng}}, \bibinfo {author} {\bibfnamefont
  {F.}~\bibnamefont {Pedregosa}}, \bibinfo {author} {\bibfnamefont
  {Y.}~\bibnamefont {Li}}, \bibinfo {author} {\bibfnamefont {I.}~\bibnamefont
  {Von~Glehn}}, \bibinfo {author} {\bibfnamefont {M.}~\bibnamefont {Rowland}},
  \bibinfo {author} {\bibfnamefont {S.}~\bibnamefont {Albanie}}, \bibinfo
  {author} {\bibfnamefont {A.}~\bibnamefont {Velingker}}, \bibinfo {author}
  {\bibfnamefont {S.}~\bibnamefont {Schmitt}}, \bibinfo {author} {\bibfnamefont
  {E.}~\bibnamefont {Lockhart}}, \bibinfo {author} {\bibfnamefont
  {E.}~\bibnamefont {Hughes}}, \bibinfo {author} {\bibfnamefont
  {H.}~\bibnamefont {Michalewski}}, \bibinfo {author} {\bibfnamefont
  {N.}~\bibnamefont {Sonnerat}}, \bibinfo {author} {\bibfnamefont
  {D.}~\bibnamefont {Hassabis}}, \bibinfo {author} {\bibfnamefont
  {P.}~\bibnamefont {Kohli}},\ and\ \bibinfo {author} {\bibfnamefont
  {D.}~\bibnamefont {Silver}},\ }\bibfield  {title} {\bibinfo {title}
  {Olympiad-level formal mathematical reasoning with reinforcement learning},\
  }\href {https://doi.org/10.1038/s41586-025-09833-y} {\bibfield  {journal}
  {\bibinfo  {journal} {Nature}\ }\textbf {\bibinfo {volume} {651}},\ \bibinfo
  {pages} {607} (\bibinfo {year} {2026})}\BibitemShut {NoStop}%
\bibitem [{\citenamefont {Ren}(2025)}]{Ren2025DeepSeekproverV2}%
  \BibitemOpen
  \bibfield  {author} {\bibinfo {author} {\bibfnamefont {Z.}~\bibnamefont
  {Ren}},\ }\href {https://arxiv.org/abs/2504.21801} {\bibinfo {title}
  {{{DeepSeek-prover-V2}}: Advancing formal mathematical reasoning via
  reinforcement learning for subgoal decomposition}} (\bibinfo {year} {2025}),\
  \Eprint {https://arxiv.org/abs/2504.21801} {arXiv:2504.21801} \BibitemShut
  {NoStop}%
\bibitem [{\citenamefont {Wang}(2025)}]{Wang2025Kiminaprover}%
  \BibitemOpen
  \bibfield  {author} {\bibinfo {author} {\bibfnamefont {H.}~\bibnamefont
  {Wang}},\ }\href {https://arxiv.org/abs/2504.11354} {\bibinfo {title}
  {Kimina-prover preview: Towards large formal reasoning models with
  reinforcement learning}} (\bibinfo {year} {2025}),\ \Eprint
  {https://arxiv.org/abs/2504.11354} {arXiv:2504.11354} \BibitemShut {NoStop}%
\bibitem [{\citenamefont {Lin}\ \emph {et~al.}(2025)\citenamefont {Lin},
  \citenamefont {Tang},\ and\ \citenamefont {Lyu}}]{Lin2025GoedelproverV2}%
  \BibitemOpen
  \bibfield  {author} {\bibinfo {author} {\bibfnamefont {Y.}~\bibnamefont
  {Lin}}, \bibinfo {author} {\bibfnamefont {S.}~\bibnamefont {Tang}},\ and\
  \bibinfo {author} {\bibfnamefont {B.}~\bibnamefont {Lyu}},\ }\href
  {https://arxiv.org/abs/2508.03613} {\bibinfo {title} {Goedel-prover-{{V2}}:
  Scaling formal theorem proving with scaffolded data synthesis and
  self-correction}} (\bibinfo {year} {2025}),\ \Eprint
  {https://arxiv.org/abs/2508.03613} {arXiv:2508.03613} \BibitemShut {NoStop}%
\bibitem [{\citenamefont {{Math Inc.}}(2025)}]{MathInc.2025Strong}%
  \BibitemOpen
  \bibfield  {author} {\bibinfo {author} {\bibnamefont {{Math Inc.}}},\ }\href
  {https://github.com/math-inc/strongpnt} {\bibinfo {title} {Strong prime
  number theorem: A lean formalization}},\ \bibinfo {howpublished} {Math Inc.}
  (\bibinfo {year} {2025})\BibitemShut {NoStop}%
\bibitem [{\citenamefont {Rammal}\ \emph {et~al.}(2026)\citenamefont {Rammal},
  \citenamefont {Patel}, \citenamefont {Gloeckle}, \citenamefont {Hayat},
  \citenamefont {Kempe}, \citenamefont {Munos}, \citenamefont {Arnal},\ and\
  \citenamefont {Cabannes}}]{Rammal2026Formalizing}%
  \BibitemOpen
  \bibfield  {author} {\bibinfo {author} {\bibfnamefont {A.}~\bibnamefont
  {Rammal}}, \bibinfo {author} {\bibfnamefont {N.}~\bibnamefont {Patel}},
  \bibinfo {author} {\bibfnamefont {F.}~\bibnamefont {Gloeckle}}, \bibinfo
  {author} {\bibfnamefont {A.}~\bibnamefont {Hayat}}, \bibinfo {author}
  {\bibfnamefont {J.}~\bibnamefont {Kempe}}, \bibinfo {author} {\bibfnamefont
  {R.}~\bibnamefont {Munos}}, \bibinfo {author} {\bibfnamefont
  {C.}~\bibnamefont {Arnal}},\ and\ \bibinfo {author} {\bibfnamefont
  {V.}~\bibnamefont {Cabannes}},\ }\href {https://arxiv.org/abs/2605.29955}
  {\bibinfo {title} {Formalizing mathematics at scale}} (\bibinfo {year}
  {2026}),\ \Eprint {https://arxiv.org/abs/2605.29955} {arXiv:2605.29955
  [cs.AI]} \BibitemShut {NoStop}%
\bibitem [{\citenamefont {Lu}\ \emph {et~al.}(2026{\natexlab{a}})\citenamefont
  {Lu}, \citenamefont {Jin}, \citenamefont {Zhang}, \citenamefont {Kos},
  \citenamefont {Cirac},\ and\ \citenamefont {Sch{\"o}lkopf}}]{Lu2026Position}%
  \BibitemOpen
  \bibfield  {author} {\bibinfo {author} {\bibfnamefont {S.}~\bibnamefont
  {Lu}}, \bibinfo {author} {\bibfnamefont {Z.}~\bibnamefont {Jin}}, \bibinfo
  {author} {\bibfnamefont {T.~J.}\ \bibnamefont {Zhang}}, \bibinfo {author}
  {\bibfnamefont {P.}~\bibnamefont {Kos}}, \bibinfo {author} {\bibfnamefont
  {J.~I.}\ \bibnamefont {Cirac}},\ and\ \bibinfo {author} {\bibfnamefont
  {B.}~\bibnamefont {Sch{\"o}lkopf}},\ }\bibfield  {title} {\bibinfo {title}
  {Position: {{LLM}} for physics research requires domain-specialized training
  and tooling},\ }in\ \href {https://openreview.net/forum?id=lHtHR6aZKZ} {\emph
  {\bibinfo {booktitle} {Forty-Third International Conference on Machine
  Learning Position Paper Track}}}\ (\bibinfo {year} {2026})\BibitemShut
  {NoStop}%
\bibitem [{\citenamefont {Commelin}\ and\ \citenamefont
  {Topaz}(2023)}]{Commelin2023Liquid}%
  \BibitemOpen
  \bibfield  {author} {\bibinfo {author} {\bibfnamefont {J.}~\bibnamefont
  {Commelin}}\ and\ \bibinfo {author} {\bibfnamefont {A.}~\bibnamefont
  {Topaz}},\ }\href {https://arxiv.org/abs/2309.14870} {\bibinfo {title}
  {Liquid tensor experiment}} (\bibinfo {year} {2023}),\ \Eprint
  {https://arxiv.org/abs/2309.14870} {arXiv:2309.14870} \BibitemShut {NoStop}%
\bibitem [{\citenamefont {Gowers}\ \emph {et~al.}(2025)\citenamefont {Gowers},
  \citenamefont {Green}, \citenamefont {Manners},\ and\ \citenamefont
  {Tao}}]{Gowers2025Conjecture}%
  \BibitemOpen
  \bibfield  {author} {\bibinfo {author} {\bibfnamefont {W.~T.}\ \bibnamefont
  {Gowers}}, \bibinfo {author} {\bibfnamefont {B.}~\bibnamefont {Green}},
  \bibinfo {author} {\bibfnamefont {F.}~\bibnamefont {Manners}},\ and\ \bibinfo
  {author} {\bibfnamefont {T.}~\bibnamefont {Tao}},\ }\bibfield  {title}
  {\bibinfo {title} {On a conjecture of marton},\ }\href
  {https://doi.org/10.4007/annals.2025.201.2.5} {\bibfield  {journal} {\bibinfo
   {journal} {Ann. Math.}\ }\textbf {\bibinfo {volume} {201}},\ \bibinfo
  {pages} {515} (\bibinfo {year} {2025})},\ \Eprint
  {https://arxiv.org/abs/2311.05762} {arXiv:2311.05762 [math.NT]} \BibitemShut
  {NoStop}%
\bibitem [{\citenamefont {Best}\ \emph {et~al.}(2025)\citenamefont {Best},
  \citenamefont {Birkbeck}, \citenamefont {Brasca}, \citenamefont
  {Rodriguez~Boidi}, \citenamefont {Velde},\ and\ \citenamefont
  {Yang}}]{Best2025Complete}%
  \BibitemOpen
  \bibfield  {author} {\bibinfo {author} {\bibfnamefont {A.}~\bibnamefont
  {Best}}, \bibinfo {author} {\bibfnamefont {C.}~\bibnamefont {Birkbeck}},
  \bibinfo {author} {\bibfnamefont {R.}~\bibnamefont {Brasca}}, \bibinfo
  {author} {\bibfnamefont {E.}~\bibnamefont {Rodriguez~Boidi}}, \bibinfo
  {author} {\bibfnamefont {R.~V.~D.}\ \bibnamefont {Velde}},\ and\ \bibinfo
  {author} {\bibfnamefont {A.}~\bibnamefont {Yang}},\ }\bibfield  {title}
  {\bibinfo {title} {A complete formalization of {{Fermat}}'s last theorem for
  regular primes in lean},\ }\href {https://doi.org/10.46298/afm.14586}
  {\bibfield  {journal} {\bibinfo  {journal} {Ann. Formaliz. Math.}\ }\textbf
  {\bibinfo {volume} {Volume 1}},\ \bibinfo {pages} {14586} (\bibinfo {year}
  {2025})}\BibitemShut {NoStop}%
\bibitem [{\citenamefont {Buzzard}\ and\ \citenamefont {{Flt Project
  Contributors}}(2024)}]{Buzzard2024Lean}%
  \BibitemOpen
  \bibfield  {author} {\bibinfo {author} {\bibfnamefont {K.}~\bibnamefont
  {Buzzard}}\ and\ \bibinfo {author} {\bibnamefont {{Flt Project
  Contributors}}},\ }\href {https://imperialcollegelondon.github.io/FLT/}
  {\bibinfo {title} {Towards a lean proof of {{Fermat}}'s last theorem}},\
  \bibinfo {howpublished} {Imperial College London} (\bibinfo {year}
  {2024})\BibitemShut {NoStop}%
\bibitem [{\citenamefont {{Axiom}}(2025)}]{Axiom2025AxiomProver}%
  \BibitemOpen
  \bibfield  {author} {\bibinfo {author} {\bibnamefont {{Axiom}}},\ }\href
  {https://axiommath.ai} {\bibinfo {title} {{{AxiomProver}}}} (\bibinfo {year}
  {2025})\BibitemShut {NoStop}%
\bibitem [{\citenamefont {{Harmonic}}(2025)}]{Harmonic2025Aristotle}%
  \BibitemOpen
  \bibfield  {author} {\bibinfo {author} {\bibnamefont {{Harmonic}}},\ }\href
  {https://harmonic.fun} {\bibinfo {title} {Aristotle}} (\bibinfo {year}
  {2025})\BibitemShut {NoStop}%
\bibitem [{\citenamefont {Wu}\ \emph {et~al.}(2022)\citenamefont {Wu},
  \citenamefont {Jiang}, \citenamefont {Li}, \citenamefont {Rabe},
  \citenamefont {Staats}, \citenamefont {Jamnik},\ and\ \citenamefont
  {Szegedy}}]{Wu2022Autoformalization}%
  \BibitemOpen
  \bibfield  {author} {\bibinfo {author} {\bibfnamefont {Y.}~\bibnamefont
  {Wu}}, \bibinfo {author} {\bibfnamefont {A.~Q.}\ \bibnamefont {Jiang}},
  \bibinfo {author} {\bibfnamefont {W.}~\bibnamefont {Li}}, \bibinfo {author}
  {\bibfnamefont {M.~N.}\ \bibnamefont {Rabe}}, \bibinfo {author}
  {\bibfnamefont {C.}~\bibnamefont {Staats}}, \bibinfo {author} {\bibfnamefont
  {M.}~\bibnamefont {Jamnik}},\ and\ \bibinfo {author} {\bibfnamefont
  {C.}~\bibnamefont {Szegedy}},\ }\bibfield  {title} {\bibinfo {title}
  {Autoformalization with large language models},\ }in\ \href@noop {} {\emph
  {\bibinfo {booktitle} {Adv. {{Neural}} Inf. {{Process}}. {{Syst}}. 35
  ({{NeurIPS}} 2022)}}}\ (\bibinfo {year} {2022})\ pp.\ \bibinfo {pages}
  {32353--32368},\ \Eprint {https://arxiv.org/abs/2205.12615} {arXiv:2205.12615
  [cs.LG]} \BibitemShut {NoStop}%
\bibitem [{\citenamefont {Wang}\ \emph {et~al.}(2025)\citenamefont {Wang},
  \citenamefont {Pan}, \citenamefont {Li}, \citenamefont {Zhang}, \citenamefont
  {Jia}, \citenamefont {Diao}, \citenamefont {Pi}, \citenamefont {Hu},\ and\
  \citenamefont {Zhang}}]{Wang2025MALoT}%
  \BibitemOpen
  \bibfield  {author} {\bibinfo {author} {\bibfnamefont {R.}~\bibnamefont
  {Wang}}, \bibinfo {author} {\bibfnamefont {R.}~\bibnamefont {Pan}}, \bibinfo
  {author} {\bibfnamefont {Y.}~\bibnamefont {Li}}, \bibinfo {author}
  {\bibfnamefont {J.}~\bibnamefont {Zhang}}, \bibinfo {author} {\bibfnamefont
  {Y.}~\bibnamefont {Jia}}, \bibinfo {author} {\bibfnamefont {S.}~\bibnamefont
  {Diao}}, \bibinfo {author} {\bibfnamefont {R.}~\bibnamefont {Pi}}, \bibinfo
  {author} {\bibfnamefont {J.}~\bibnamefont {Hu}},\ and\ \bibinfo {author}
  {\bibfnamefont {T.}~\bibnamefont {Zhang}},\ }\href@noop {} {\bibinfo {title}
  {{{MA-LoT}}: Model-collaboration lean-based long chain-of-thought reasoning
  enhances formal theorem proving}} (\bibinfo {year} {2025}),\ \Eprint
  {https://arxiv.org/abs/2503.03205} {arXiv:2503.03205} \BibitemShut {NoStop}%
\bibitem [{\citenamefont {Yang}\ \emph {et~al.}(2023)\citenamefont {Yang},
  \citenamefont {Swope}, \citenamefont {Gu}, \citenamefont {Chalamala},
  \citenamefont {Song}, \citenamefont {Yu}, \citenamefont {Godil},
  \citenamefont {Prenger},\ and\ \citenamefont
  {Anandkumar}}]{Yang2023LeanDojo}%
  \BibitemOpen
  \bibfield  {author} {\bibinfo {author} {\bibfnamefont {K.}~\bibnamefont
  {Yang}}, \bibinfo {author} {\bibfnamefont {A.}~\bibnamefont {Swope}},
  \bibinfo {author} {\bibfnamefont {A.}~\bibnamefont {Gu}}, \bibinfo {author}
  {\bibfnamefont {R.}~\bibnamefont {Chalamala}}, \bibinfo {author}
  {\bibfnamefont {P.}~\bibnamefont {Song}}, \bibinfo {author} {\bibfnamefont
  {S.}~\bibnamefont {Yu}}, \bibinfo {author} {\bibfnamefont {S.}~\bibnamefont
  {Godil}}, \bibinfo {author} {\bibfnamefont {R.}~\bibnamefont {Prenger}},\
  and\ \bibinfo {author} {\bibfnamefont {A.}~\bibnamefont {Anandkumar}},\
  }\bibfield  {title} {\bibinfo {title} {{{LeanDojo}}: Theorem proving with
  retrieval-augmented language models},\ }in\ \href@noop {} {\emph {\bibinfo
  {booktitle} {Adv. {{Neural}} Inf. {{Process}}. {{Syst}}.}}}\ (\bibinfo {year}
  {2023})\BibitemShut {NoStop}%
\bibitem [{202(2025)}]{2025Introducing}%
  \BibitemOpen
  \href {https://math.inc/gauss} {\bibinfo {title} {Introducing {{Gauss}}, an
  agent for autoformalization}} (\bibinfo {year} {2025})\BibitemShut {NoStop}%
\bibitem [{\citenamefont {{Physlib
  Contributors}}(2024)}]{PhyslibContributors2024PhysLib}%
  \BibitemOpen
  \bibfield  {author} {\bibinfo {author} {\bibnamefont {{Physlib
  Contributors}}},\ }\href {https://github.com/leanprover-community/physlib}
  {\bibinfo {title} {{{PhysLib}}: A lean 4 library for physics}},\ \bibinfo
  {howpublished} {leanprover-community} (\bibinfo {year} {2024})\BibitemShut
  {NoStop}%
\bibitem [{\citenamefont {Brand{\~a}o}\ and\ \citenamefont
  {Plenio}(2010)}]{Brandao2010Generalization}%
  \BibitemOpen
  \bibfield  {author} {\bibinfo {author} {\bibfnamefont {F.~G. S.~L.}\
  \bibnamefont {Brand{\~a}o}}\ and\ \bibinfo {author} {\bibfnamefont {M.~B.}\
  \bibnamefont {Plenio}},\ }\bibfield  {title} {\bibinfo {title} {A
  generalization of quantum stein's lemma},\ }\href
  {https://doi.org/10.1007/s00220-010-1005-z} {\bibfield  {journal} {\bibinfo
  {journal} {Commun. Math. Phys.}\ }\textbf {\bibinfo {volume} {295}},\
  \bibinfo {pages} {791} (\bibinfo {year} {2010})}\BibitemShut {NoStop}%
\bibitem [{\citenamefont {Berta}\ \emph {et~al.}(2023)\citenamefont {Berta},
  \citenamefont {Brand{\~a}o}, \citenamefont {Gour}, \citenamefont {Lami},
  \citenamefont {Plenio}, \citenamefont {Regula},\ and\ \citenamefont
  {Tomamichel}}]{Berta2023Gap}%
  \BibitemOpen
  \bibfield  {author} {\bibinfo {author} {\bibfnamefont {M.}~\bibnamefont
  {Berta}}, \bibinfo {author} {\bibfnamefont {F.~G. S.~L.}\ \bibnamefont
  {Brand{\~a}o}}, \bibinfo {author} {\bibfnamefont {G.}~\bibnamefont {Gour}},
  \bibinfo {author} {\bibfnamefont {L.}~\bibnamefont {Lami}}, \bibinfo {author}
  {\bibfnamefont {M.~B.}\ \bibnamefont {Plenio}}, \bibinfo {author}
  {\bibfnamefont {B.}~\bibnamefont {Regula}},\ and\ \bibinfo {author}
  {\bibfnamefont {M.}~\bibnamefont {Tomamichel}},\ }\bibfield  {title}
  {\bibinfo {title} {On a gap in the proof of the generalised quantum stein's
  lemma and its consequences for the reversibility of quantum resources},\
  }\href {https://doi.org/10.22331/q-2023-09-07-1103} {\bibfield  {journal}
  {\bibinfo  {journal} {Quantum}\ }\textbf {\bibinfo {volume} {7}},\ \bibinfo
  {pages} {1103} (\bibinfo {year} {2023})}\BibitemShut {NoStop}%
\bibitem [{\citenamefont {Hayashi}\ and\ \citenamefont
  {Yamasaki}(2025)}]{Hayashi2025Generalized}%
  \BibitemOpen
  \bibfield  {author} {\bibinfo {author} {\bibfnamefont {M.}~\bibnamefont
  {Hayashi}}\ and\ \bibinfo {author} {\bibfnamefont {H.}~\bibnamefont
  {Yamasaki}},\ }\bibfield  {title} {\bibinfo {title} {The generalized quantum
  stein's lemma and the second law of quantum resource theories},\ }\href
  {https://doi.org/10.1038/s41567-025-03047-9} {\bibfield  {journal} {\bibinfo
  {journal} {Nat. Phys.}\ }\textbf {\bibinfo {volume} {21}},\ \bibinfo {pages}
  {1988} (\bibinfo {year} {2025})}\BibitemShut {NoStop}%
\bibitem [{\citenamefont {Lami}(2025)}]{Lami2025Solution}%
  \BibitemOpen
  \bibfield  {author} {\bibinfo {author} {\bibfnamefont {L.}~\bibnamefont
  {Lami}},\ }\bibfield  {title} {\bibinfo {title} {A solution of the
  generalized quantum stein's lemma},\ }\href
  {https://doi.org/10.1109/TIT.2025.3543610} {\bibfield  {journal} {\bibinfo
  {journal} {IEEE Trans. Inf. Theory}\ }\textbf {\bibinfo {volume} {71}},\
  \bibinfo {pages} {4454} (\bibinfo {year} {2025})}\BibitemShut {NoStop}%
\bibitem [{\citenamefont {Meiburg}(2024)}]{Meiburg2024LeanQuantumInfo}%
  \BibitemOpen
  \bibfield  {author} {\bibinfo {author} {\bibfnamefont {A.}~\bibnamefont
  {Meiburg}},\ }\href {https://github.com/Timeroot/Lean-QuantumInfo} {\bibinfo
  {title} {Lean-{{QuantumInfo}}}} (\bibinfo {year} {2024})\BibitemShut
  {NoStop}%
\bibitem [{\citenamefont {Meiburg}\ \emph {et~al.}(2025)\citenamefont
  {Meiburg}, \citenamefont {Lessa},\ and\ \citenamefont
  {Soldati}}]{Meiburg2025Formalization}%
  \BibitemOpen
  \bibfield  {author} {\bibinfo {author} {\bibfnamefont {A.}~\bibnamefont
  {Meiburg}}, \bibinfo {author} {\bibfnamefont {L.~A.}\ \bibnamefont {Lessa}},\
  and\ \bibinfo {author} {\bibfnamefont {R.~R.}\ \bibnamefont {Soldati}},\
  }\href@noop {} {\bibinfo {title} {A formalization of the generalized quantum
  stein's lemma in lean}} (\bibinfo {year} {2025})\BibitemShut {NoStop}%
\bibitem [{\citenamefont {{Perez-Garcia}}\ \emph {et~al.}(2007)\citenamefont
  {{Perez-Garcia}}, \citenamefont {Verstraete}, \citenamefont {Wolf},\ and\
  \citenamefont {Cirac}}]{Perez-Garcia2007Matrix}%
  \BibitemOpen
  \bibfield  {author} {\bibinfo {author} {\bibfnamefont {D.}~\bibnamefont
  {{Perez-Garcia}}}, \bibinfo {author} {\bibfnamefont {F.}~\bibnamefont
  {Verstraete}}, \bibinfo {author} {\bibfnamefont {M.}~\bibnamefont {Wolf}},\
  and\ \bibinfo {author} {\bibfnamefont {J.}~\bibnamefont {Cirac}},\ }\bibfield
   {title} {\bibinfo {title} {Matrix product state representations},\ }\href
  {https://doi.org/10.26421/QIC7.5-6-1} {\bibfield  {journal} {\bibinfo
  {journal} {Quantum Inf. Comput.}\ }\textbf {\bibinfo {volume} {7}},\ \bibinfo
  {pages} {401} (\bibinfo {year} {2007})}\BibitemShut {NoStop}%
\bibitem [{\citenamefont {Cirac}\ \emph
  {et~al.}(2017{\natexlab{a}})\citenamefont {Cirac}, \citenamefont
  {{P{\'e}rez-Garc{\'i}a}}, \citenamefont {Schuch},\ and\ \citenamefont
  {Verstraete}}]{Cirac2017Matrixa}%
  \BibitemOpen
  \bibfield  {author} {\bibinfo {author} {\bibfnamefont {J.}~\bibnamefont
  {Cirac}}, \bibinfo {author} {\bibfnamefont {D.}~\bibnamefont
  {{P{\'e}rez-Garc{\'i}a}}}, \bibinfo {author} {\bibfnamefont {N.}~\bibnamefont
  {Schuch}},\ and\ \bibinfo {author} {\bibfnamefont {F.}~\bibnamefont
  {Verstraete}},\ }\bibfield  {title} {\bibinfo {title} {Matrix product density
  operators: Renormalization fixed points and boundary theories},\ }\href
  {https://doi.org/10.1016/j.aop.2016.12.030} {\bibfield  {journal} {\bibinfo
  {journal} {Ann. Phys.}\ }\textbf {\bibinfo {volume} {378}},\ \bibinfo {pages}
  {100} (\bibinfo {year} {2017}{\natexlab{a}})}\BibitemShut {NoStop}%
\bibitem [{\citenamefont {Cirac}\ \emph {et~al.}(2021)\citenamefont {Cirac},
  \citenamefont {{P{\'e}rez-Garc{\'i}a}}, \citenamefont {Schuch},\ and\
  \citenamefont {Verstraete}}]{Cirac2021Matrix}%
  \BibitemOpen
  \bibfield  {author} {\bibinfo {author} {\bibfnamefont {J.~I.}\ \bibnamefont
  {Cirac}}, \bibinfo {author} {\bibfnamefont {D.}~\bibnamefont
  {{P{\'e}rez-Garc{\'i}a}}}, \bibinfo {author} {\bibfnamefont {N.}~\bibnamefont
  {Schuch}},\ and\ \bibinfo {author} {\bibfnamefont {F.}~\bibnamefont
  {Verstraete}},\ }\bibfield  {title} {\bibinfo {title} {Matrix product states
  and projected entangled pair states: Concepts, symmetries, theorems},\ }\href
  {https://doi.org/10.1103/RevModPhys.93.045003} {\bibfield  {journal}
  {\bibinfo  {journal} {Rev. Mod. Phys.}\ }\textbf {\bibinfo {volume} {93}},\
  \bibinfo {pages} {045003} (\bibinfo {year} {2021})}\BibitemShut {NoStop}%
\bibitem [{\citenamefont {Massot}(2021)}]{Massot2021Leanblueprint}%
  \BibitemOpen
  \bibfield  {author} {\bibinfo {author} {\bibfnamefont {P.}~\bibnamefont
  {Massot}},\ }\href {https://github.com/PatrickMassot/leanblueprint} {\bibinfo
  {title} {Leanblueprint}} (\bibinfo {year} {2021})\BibitemShut {NoStop}%
\bibitem [{\citenamefont {Lu}\ \emph {et~al.}(2026{\natexlab{b}})\citenamefont
  {Lu}, \citenamefont {Tjoa},\ and\ \citenamefont {Cirac}}]{TNLeanBlueprint}%
  \BibitemOpen
  \bibfield  {author} {\bibinfo {author} {\bibfnamefont {S.}~\bibnamefont
  {Lu}}, \bibinfo {author} {\bibfnamefont {E.}~\bibnamefont {Tjoa}},\ and\
  \bibinfo {author} {\bibfnamefont {J.~I.}\ \bibnamefont {Cirac}},\ }\href@noop
  {} {\bibinfo {title} {A formalization blueprint for the fundamental theorem
  of matrix product states}},\ \bibinfo {howpublished}
  {https://lionsr.github.io/TNLean/blueprint/} (\bibinfo {year}
  {2026}{\natexlab{b}})\BibitemShut {NoStop}%
\bibitem [{\citenamefont {Chen}\ \emph {et~al.}(2011)\citenamefont {Chen},
  \citenamefont {Gu},\ and\ \citenamefont {Wen}}]{Chen2011Classification}%
  \BibitemOpen
  \bibfield  {author} {\bibinfo {author} {\bibfnamefont {X.}~\bibnamefont
  {Chen}}, \bibinfo {author} {\bibfnamefont {Z.-C.}\ \bibnamefont {Gu}},\ and\
  \bibinfo {author} {\bibfnamefont {X.-G.}\ \bibnamefont {Wen}},\ }\bibfield
  {title} {\bibinfo {title} {Classification of gapped symmetric phases in
  one-dimensional spin systems},\ }\href
  {https://doi.org/10.1103/PhysRevB.83.035107} {\bibfield  {journal} {\bibinfo
  {journal} {Phys. Rev. B}\ }\textbf {\bibinfo {volume} {83}},\ \bibinfo
  {pages} {035107} (\bibinfo {year} {2011})}\BibitemShut {NoStop}%
\bibitem [{\citenamefont {Schuch}\ \emph {et~al.}(2011)\citenamefont {Schuch},
  \citenamefont {{P{\'e}rez-Garc{\'i}a}},\ and\ \citenamefont
  {Cirac}}]{Schuch2011Classifying}%
  \BibitemOpen
  \bibfield  {author} {\bibinfo {author} {\bibfnamefont {N.}~\bibnamefont
  {Schuch}}, \bibinfo {author} {\bibfnamefont {D.}~\bibnamefont
  {{P{\'e}rez-Garc{\'i}a}}},\ and\ \bibinfo {author} {\bibfnamefont
  {I.}~\bibnamefont {Cirac}},\ }\bibfield  {title} {\bibinfo {title}
  {Classifying quantum phases using matrix product states and projected
  entangled pair states},\ }\href {https://doi.org/10.1103/PhysRevB.84.165139}
  {\bibfield  {journal} {\bibinfo  {journal} {Phys. Rev. B}\ }\textbf {\bibinfo
  {volume} {84}},\ \bibinfo {pages} {165139} (\bibinfo {year}
  {2011})}\BibitemShut {NoStop}%
\bibitem [{\citenamefont {Pollmann}\ \emph {et~al.}(2010)\citenamefont
  {Pollmann}, \citenamefont {Turner}, \citenamefont {Berg},\ and\ \citenamefont
  {Oshikawa}}]{Pollmann2010Entanglement}%
  \BibitemOpen
  \bibfield  {author} {\bibinfo {author} {\bibfnamefont {F.}~\bibnamefont
  {Pollmann}}, \bibinfo {author} {\bibfnamefont {A.~M.}\ \bibnamefont
  {Turner}}, \bibinfo {author} {\bibfnamefont {E.}~\bibnamefont {Berg}},\ and\
  \bibinfo {author} {\bibfnamefont {M.}~\bibnamefont {Oshikawa}},\ }\bibfield
  {title} {\bibinfo {title} {Entanglement spectrum of a topological phase in
  one dimension},\ }\href {https://doi.org/10.1103/physrevb.81.064439}
  {\bibfield  {journal} {\bibinfo  {journal} {Phys. Rev. B}\ }\textbf {\bibinfo
  {volume} {81}},\ \bibinfo {pages} {064439} (\bibinfo {year}
  {2010})}\BibitemShut {NoStop}%
\bibitem [{\citenamefont {Pollmann}\ \emph {et~al.}(2012)\citenamefont
  {Pollmann}, \citenamefont {Berg}, \citenamefont {Turner},\ and\ \citenamefont
  {Oshikawa}}]{Pollmann2012Symmetry}%
  \BibitemOpen
  \bibfield  {author} {\bibinfo {author} {\bibfnamefont {F.}~\bibnamefont
  {Pollmann}}, \bibinfo {author} {\bibfnamefont {E.}~\bibnamefont {Berg}},
  \bibinfo {author} {\bibfnamefont {A.~M.}\ \bibnamefont {Turner}},\ and\
  \bibinfo {author} {\bibfnamefont {M.}~\bibnamefont {Oshikawa}},\ }\bibfield
  {title} {\bibinfo {title} {Symmetry protection of topological phases in
  one-dimensional quantum spin systems},\ }\href
  {https://doi.org/10.1103/physrevb.85.075125} {\bibfield  {journal} {\bibinfo
  {journal} {Phys. Rev. B}\ }\textbf {\bibinfo {volume} {85}},\ \bibinfo
  {pages} {075125} (\bibinfo {year} {2012})}\BibitemShut {NoStop}%
\bibitem [{TNL(2026)}]{TNLean}%
  \BibitemOpen
  \href@noop {} {\bibinfo {title} {{TNLean}: A {Lean~4} tensor-network
  library}},\ \bibinfo {howpublished} {https://github.com/lionsr/TNLean}
  (\bibinfo {year} {2026})\BibitemShut {NoStop}%
\bibitem [{\citenamefont {Verstraete}\ and\ \citenamefont
  {Cirac}(2006)}]{Verstraete2006Matrix}%
  \BibitemOpen
  \bibfield  {author} {\bibinfo {author} {\bibfnamefont {F.}~\bibnamefont
  {Verstraete}}\ and\ \bibinfo {author} {\bibfnamefont {J.~I.}\ \bibnamefont
  {Cirac}},\ }\bibfield  {title} {\bibinfo {title} {Matrix product states
  represent ground states faithfully},\ }\href
  {https://doi.org/10.1103/PhysRevB.73.094423} {\bibfield  {journal} {\bibinfo
  {journal} {Phys. Rev. B}\ }\textbf {\bibinfo {volume} {73}},\ \bibinfo
  {pages} {094423} (\bibinfo {year} {2006})}\BibitemShut {NoStop}%
\bibitem [{\citenamefont {Cirac}\ \emph
  {et~al.}(2017{\natexlab{b}})\citenamefont {Cirac}, \citenamefont
  {{P{\'e}rez-Garc{\'i}a}}, \citenamefont {Schuch},\ and\ \citenamefont
  {Verstraete}}]{Cirac2017Matrix}%
  \BibitemOpen
  \bibfield  {author} {\bibinfo {author} {\bibfnamefont {J.~I.}\ \bibnamefont
  {Cirac}}, \bibinfo {author} {\bibfnamefont {D.}~\bibnamefont
  {{P{\'e}rez-Garc{\'i}a}}}, \bibinfo {author} {\bibfnamefont {N.}~\bibnamefont
  {Schuch}},\ and\ \bibinfo {author} {\bibfnamefont {F.}~\bibnamefont
  {Verstraete}},\ }\bibfield  {title} {\bibinfo {title} {Matrix product
  unitaries: Structure, symmetries, and topological invariants},\ }\href
  {https://doi.org/10.1088/1742-5468/aa7e55} {\bibfield  {journal} {\bibinfo
  {journal} {J. Stat. Mech: Theory Exp.}\ }\textbf {\bibinfo {volume} {2017}},\
  \bibinfo {pages} {083105} (\bibinfo {year} {2017}{\natexlab{b}})},\ \Eprint
  {https://arxiv.org/abs/1710.04799} {arXiv:1710.04799 [math.CT]} \BibitemShut
  {NoStop}%
\bibitem [{\citenamefont {Liu}\ \emph {et~al.}(2026)\citenamefont {Liu},
  \citenamefont {{Ruiz-de-Alarc{\'o}n}}, \citenamefont {Styliaris},
  \citenamefont {Sun}, \citenamefont {{P{\'e}rez-Garc{\'i}a}},\ and\
  \citenamefont {Cirac}}]{Liu2026Parent}%
  \BibitemOpen
  \bibfield  {author} {\bibinfo {author} {\bibfnamefont {Y.}~\bibnamefont
  {Liu}}, \bibinfo {author} {\bibfnamefont {A.}~\bibnamefont
  {{Ruiz-de-Alarc{\'o}n}}}, \bibinfo {author} {\bibfnamefont {G.}~\bibnamefont
  {Styliaris}}, \bibinfo {author} {\bibfnamefont {X.-Q.}\ \bibnamefont {Sun}},
  \bibinfo {author} {\bibfnamefont {D.}~\bibnamefont
  {{P{\'e}rez-Garc{\'i}a}}},\ and\ \bibinfo {author} {\bibfnamefont {J.~I.}\
  \bibnamefont {Cirac}},\ }\bibfield  {title} {\bibinfo {title} {Parent
  lindbladians for matrix product density operators},\ }\href
  {https://doi.org/10.1103/1qyd-59z7} {\bibfield  {journal} {\bibinfo
  {journal} {Phys. Rev. Res.}\ }\textbf {\bibinfo {volume} {8}},\ \bibinfo
  {pages} {013210} (\bibinfo {year} {2026})}\BibitemShut {NoStop}%
\bibitem [{\citenamefont {Fannes}\ \emph {et~al.}(1992)\citenamefont {Fannes},
  \citenamefont {Nachtergaele},\ and\ \citenamefont
  {Werner}}]{Fannes1992Finitely}%
  \BibitemOpen
  \bibfield  {author} {\bibinfo {author} {\bibfnamefont {M.}~\bibnamefont
  {Fannes}}, \bibinfo {author} {\bibfnamefont {B.}~\bibnamefont
  {Nachtergaele}},\ and\ \bibinfo {author} {\bibfnamefont {R.~F.}\ \bibnamefont
  {Werner}},\ }\bibfield  {title} {\bibinfo {title} {Finitely correlated states
  on quantum spin chains},\ }\href {https://doi.org/10.1007/BF02099178}
  {\bibfield  {journal} {\bibinfo  {journal} {Commun. Math. Phys.}\ }\textbf
  {\bibinfo {volume} {144}},\ \bibinfo {pages} {443} (\bibinfo {year}
  {1992})}\BibitemShut {NoStop}%
\bibitem [{Sup()}]{Supp}%
  \BibitemOpen
  \href@noop {} {}\bibinfo {note} {See {S}upplemental {M}aterial at [{URL} will
  be inserted by publisher] for details of the agent architecture and {Lean}
  interface, the blueprint, the persistent-memory system, orchestration
  patterns, cost analysis, a representative proof excerpt, system prompts, and
  a paper-to-formalization gap analysis}\BibitemShut {NoStop}%
\bibitem [{Note1()}]{Note1}%
  \BibitemOpen
  \bibinfo {note} {The formalization follows the $(D^2-d+1)D^2 = O(D^4)$
  blocking bound of the quantum Wielandt inequality~\cite {Sanz2010Quantum};
  whether this can be improved to the optimal scaling remains an interesting
  open question~\cite {Michalek2019Quantum,Shitov2023Growth}.}\BibitemShut
  {Stop}%
\bibitem [{\citenamefont {Wolf}(2012)}]{Wolf2012Quantum}%
  \BibitemOpen
  \bibfield  {author} {\bibinfo {author} {\bibfnamefont {M.~M.}\ \bibnamefont
  {Wolf}},\ }\href {https://mediatum.ub.tum.de/doc/1701036/1701036.pdf} {\emph
  {\bibinfo {title} {Quantum Channels \& Operations: Guided Tour}}},\ \bibinfo
  {type} {Lecture Notes}\ (\bibinfo  {institution} {Technische Universit\"at
  M\"unchen},\ \bibinfo {year} {2012})\BibitemShut {NoStop}%
\bibitem [{\citenamefont {Sanz}\ \emph {et~al.}(2010)\citenamefont {Sanz},
  \citenamefont {{Perez-Garcia}}, \citenamefont {Wolf},\ and\ \citenamefont
  {Cirac}}]{Sanz2010Quantum}%
  \BibitemOpen
  \bibfield  {author} {\bibinfo {author} {\bibfnamefont {M.}~\bibnamefont
  {Sanz}}, \bibinfo {author} {\bibfnamefont {D.}~\bibnamefont
  {{Perez-Garcia}}}, \bibinfo {author} {\bibfnamefont {M.~M.}\ \bibnamefont
  {Wolf}},\ and\ \bibinfo {author} {\bibfnamefont {J.~I.}\ \bibnamefont
  {Cirac}},\ }\bibfield  {title} {\bibinfo {title} {A quantum version of
  wielandt's inequality},\ }\href {https://doi.org/10.1109/TIT.2010.2054552}
  {\bibfield  {journal} {\bibinfo  {journal} {IEEE Trans. Inf. Theory}\
  }\textbf {\bibinfo {volume} {56}},\ \bibinfo {pages} {4668} (\bibinfo {year}
  {2010})}\BibitemShut {NoStop}%
\bibitem [{Note2()}]{Note2}%
  \BibitemOpen
  \bibinfo {note} {After the formalization for the fundamental theorem has been
  completed, the positive and completely positive maps between $C^*$-algebras
  were added to Mathlib 4.31; TNLean has since been updated to build on them,
  where it had previously constructed these notions itself.}\BibitemShut
  {Stop}%
\bibitem [{\citenamefont {{P{\'e}rez-Garc{\'i}a}}\ \emph
  {et~al.}(2008)\citenamefont {{P{\'e}rez-Garc{\'i}a}}, \citenamefont {Wolf},
  \citenamefont {Sanz}, \citenamefont {Verstraete},\ and\ \citenamefont
  {Cirac}}]{Perez-Garcia2008String}%
  \BibitemOpen
  \bibfield  {author} {\bibinfo {author} {\bibfnamefont {D.}~\bibnamefont
  {{P{\'e}rez-Garc{\'i}a}}}, \bibinfo {author} {\bibfnamefont {M.~M.}\
  \bibnamefont {Wolf}}, \bibinfo {author} {\bibfnamefont {M.}~\bibnamefont
  {Sanz}}, \bibinfo {author} {\bibfnamefont {F.}~\bibnamefont {Verstraete}},\
  and\ \bibinfo {author} {\bibfnamefont {J.~I.}\ \bibnamefont {Cirac}},\
  }\bibfield  {title} {\bibinfo {title} {String order and symmetries in quantum
  spin lattices},\ }\href {https://doi.org/10.1103/PhysRevLett.100.167202}
  {\bibfield  {journal} {\bibinfo  {journal} {Phys. Rev. Lett.}\ }\textbf
  {\bibinfo {volume} {100}},\ \bibinfo {pages} {167202} (\bibinfo {year}
  {2008})}\BibitemShut {NoStop}%
\bibitem [{\citenamefont {Chen}\ \emph {et~al.}(2013)\citenamefont {Chen},
  \citenamefont {Gu}, \citenamefont {Liu},\ and\ \citenamefont
  {Wen}}]{Chen2013Symmetry}%
  \BibitemOpen
  \bibfield  {author} {\bibinfo {author} {\bibfnamefont {X.}~\bibnamefont
  {Chen}}, \bibinfo {author} {\bibfnamefont {Z.-C.}\ \bibnamefont {Gu}},
  \bibinfo {author} {\bibfnamefont {Z.-X.}\ \bibnamefont {Liu}},\ and\ \bibinfo
  {author} {\bibfnamefont {X.-G.}\ \bibnamefont {Wen}},\ }\bibfield  {title}
  {\bibinfo {title} {Symmetry protected topological orders and the group
  cohomology of their symmetry group},\ }\href
  {https://doi.org/10.1103/physrevb.87.155114} {\bibfield  {journal} {\bibinfo
  {journal} {Phys. Rev. B}\ }\textbf {\bibinfo {volume} {87}},\ \bibinfo
  {pages} {155114} (\bibinfo {year} {2013})}\BibitemShut {NoStop}%
\bibitem [{202(2026)}]{2026Leanprovercommunity}%
  \BibitemOpen
  \href {https://github.com/leanprover-community/mathlib4} {\bibinfo {title}
  {Leanprover-community/mathlib4}},\ \bibinfo {howpublished}
  {leanprover-community} (\bibinfo {year} {2026})\BibitemShut {NoStop}%
\bibitem [{\citenamefont {{The Mathlib
  Community}}(2020)}]{TheMathlibCommunity2020Lean}%
  \BibitemOpen
  \bibfield  {author} {\bibinfo {author} {\bibnamefont {{The Mathlib
  Community}}},\ }\bibfield  {title} {\bibinfo {title} {The lean mathematical
  library},\ }in\ \href {https://doi.org/10.1145/3372885.3373824} {\emph
  {\bibinfo {booktitle} {Proceedings of the 9th {{ACM SIGPLAN}} International
  Conference on Certified Programs and Proofs ({{CPP}} 2020)}}}\ (\bibinfo
  {publisher} {ACM},\ \bibinfo {year} {2020})\ pp.\ \bibinfo {pages}
  {367--381},\ \Eprint {https://arxiv.org/abs/1910.09336} {arXiv:1910.09336
  [cs.LO]} \BibitemShut {NoStop}%
\bibitem [{\citenamefont {Buzzard}(2025)}]{Buzzard2025Formal}%
  \BibitemOpen
  \bibfield  {author} {\bibinfo {author} {\bibfnamefont {K.}~\bibnamefont
  {Buzzard}},\ }\href
  {https://xenaproject.wordpress.com/2025/10/22/formal-or-not-formal-that-is-the-question-in-ai-for-theorem-proving/}
  {\bibinfo {title} {Formal or not formal? {{That}} is the question in {{AI}}
  for theorem proving}} (\bibinfo {year} {2025})\BibitemShut {NoStop}%
\bibitem [{\citenamefont {{Florido-Llin{\`a}s}}\ \emph
  {et~al.}(2025)\citenamefont {{Florido-Llin{\`a}s}}, \citenamefont {Alhambra},
  \citenamefont {{P{\'e}rez-Garc{\'i}a}},\ and\ \citenamefont
  {Cirac}}]{Florido-Llinas2025Uniform}%
  \BibitemOpen
  \bibfield  {author} {\bibinfo {author} {\bibfnamefont {M.}~\bibnamefont
  {{Florido-Llin{\`a}s}}}, \bibinfo {author} {\bibfnamefont {{\'A}.~M.}\
  \bibnamefont {Alhambra}}, \bibinfo {author} {\bibfnamefont {D.}~\bibnamefont
  {{P{\'e}rez-Garc{\'i}a}}},\ and\ \bibinfo {author} {\bibfnamefont {J.~I.}\
  \bibnamefont {Cirac}},\ }\href@noop {} {\bibinfo {title} {Uniform matrix
  product states with a boundary}} (\bibinfo {year} {2025}),\ \Eprint
  {https://arxiv.org/abs/2512.11968} {arXiv:2512.11968 [quant-ph]} \BibitemShut
  {NoStop}%
\bibitem [{\citenamefont {Molnar}\ \emph {et~al.}(2018)\citenamefont {Molnar},
  \citenamefont {{Garre-Rubio}}, \citenamefont {{P{\'e}rez-Garc{\'i}a}},
  \citenamefont {Schuch},\ and\ \citenamefont {Cirac}}]{Molnar2018Normal}%
  \BibitemOpen
  \bibfield  {author} {\bibinfo {author} {\bibfnamefont {A.}~\bibnamefont
  {Molnar}}, \bibinfo {author} {\bibfnamefont {J.}~\bibnamefont
  {{Garre-Rubio}}}, \bibinfo {author} {\bibfnamefont {D.}~\bibnamefont
  {{P{\'e}rez-Garc{\'i}a}}}, \bibinfo {author} {\bibfnamefont {N.}~\bibnamefont
  {Schuch}},\ and\ \bibinfo {author} {\bibfnamefont {J.~I.}\ \bibnamefont
  {Cirac}},\ }\bibfield  {title} {\bibinfo {title} {Normal projected entangled
  pair states generating the same state},\ }\href
  {https://doi.org/10.1088/1367-2630/aae9fa} {\bibfield  {journal} {\bibinfo
  {journal} {New J. Phys.}\ }\textbf {\bibinfo {volume} {20}},\ \bibinfo
  {pages} {113017} (\bibinfo {year} {2018})},\ \Eprint
  {https://arxiv.org/abs/1804.04964} {arXiv:1804.04964 [cond-mat.str-el]}
  \BibitemShut {NoStop}%
\bibitem [{\citenamefont {{Garre-Rubio}}\ \emph {et~al.}(2025)\citenamefont
  {{Garre-Rubio}}, \citenamefont {Turzillo},\ and\ \citenamefont
  {Moln{\'a}r}}]{Garre-Rubio2025MPS}%
  \BibitemOpen
  \bibfield  {author} {\bibinfo {author} {\bibfnamefont {J.}~\bibnamefont
  {{Garre-Rubio}}}, \bibinfo {author} {\bibfnamefont {A.}~\bibnamefont
  {Turzillo}},\ and\ \bibinfo {author} {\bibfnamefont {A.}~\bibnamefont
  {Moln{\'a}r}},\ }\bibfield  {title} {\bibinfo {title} {{{MPS}} stability and
  the intersection property},\ }in\ \href@noop {} {\emph {\bibinfo {booktitle}
  {Annales Henri Poincar\'e}}}\ (\bibinfo  {publisher} {Springer},\ \bibinfo
  {year} {2025})\ pp.\ \bibinfo {pages} {1--23}\BibitemShut {NoStop}%
\bibitem [{\citenamefont {Schuch}\ \emph {et~al.}(2025)\citenamefont {Schuch},
  \citenamefont {Molnar},\ and\ \citenamefont
  {{Perez-Garcia}}}]{Schuch2025Simple}%
  \BibitemOpen
  \bibfield  {author} {\bibinfo {author} {\bibfnamefont {N.}~\bibnamefont
  {Schuch}}, \bibinfo {author} {\bibfnamefont {A.}~\bibnamefont {Molnar}},\
  and\ \bibinfo {author} {\bibfnamefont {D.}~\bibnamefont {{Perez-Garcia}}},\
  }\bibfield  {title} {\bibinfo {title} {Simple hamiltonians for matrix product
  state models},\ }\href@noop {} {\bibfield  {journal} {\bibinfo  {journal}
  {arXiv Prepr. arXiv:2503,10767}\ } (\bibinfo {year} {2025})},\ \Eprint
  {https://arxiv.org/abs/2503.10767} {arXiv:2503.10767} \BibitemShut {NoStop}%
\bibitem [{\citenamefont {Chen}\ \emph {et~al.}(2012)\citenamefont {Chen},
  \citenamefont {Gu}, \citenamefont {Liu},\ and\ \citenamefont
  {Wen}}]{Chen2012Symmetryprotected}%
  \BibitemOpen
  \bibfield  {author} {\bibinfo {author} {\bibfnamefont {X.}~\bibnamefont
  {Chen}}, \bibinfo {author} {\bibfnamefont {Z.-C.}\ \bibnamefont {Gu}},
  \bibinfo {author} {\bibfnamefont {Z.-X.}\ \bibnamefont {Liu}},\ and\ \bibinfo
  {author} {\bibfnamefont {X.-G.}\ \bibnamefont {Wen}},\ }\bibfield  {title}
  {\bibinfo {title} {Symmetry-protected topological orders in interacting
  bosonic systems},\ }\href {https://doi.org/10.1126/science.1227224}
  {\bibfield  {journal} {\bibinfo  {journal} {Science}\ }\textbf {\bibinfo
  {volume} {338}},\ \bibinfo {pages} {1604} (\bibinfo {year}
  {2012})}\BibitemShut {NoStop}%
\bibitem [{\citenamefont {Ji}\ \emph {et~al.}(2020)\citenamefont {Ji},
  \citenamefont {Natarajan}, \citenamefont {Vidick}, \citenamefont {Wright},\
  and\ \citenamefont {Yuen}}]{Ji2020LowIndividualDegree}%
  \BibitemOpen
  \bibfield  {author} {\bibinfo {author} {\bibfnamefont {Z.}~\bibnamefont
  {Ji}}, \bibinfo {author} {\bibfnamefont {A.}~\bibnamefont {Natarajan}},
  \bibinfo {author} {\bibfnamefont {T.}~\bibnamefont {Vidick}}, \bibinfo
  {author} {\bibfnamefont {J.}~\bibnamefont {Wright}},\ and\ \bibinfo {author}
  {\bibfnamefont {H.}~\bibnamefont {Yuen}},\ }\href@noop {} {\bibinfo {title}
  {Quantum soundness of the classical low individual degree test}} (\bibinfo
  {year} {2020}),\ \Eprint {https://arxiv.org/abs/2009.12982} {arXiv:2009.12982
  [quant-ph]} \BibitemShut {NoStop}%
\bibitem [{LDT(2026)}]{LDTLeanPaper}%
  \BibitemOpen
  \href@noop {} {\bibinfo {title} {Low degree test {Lean} formalization}}
  (\bibinfo {year} {2026}),\ \bibinfo {note} {in preparation}\BibitemShut
  {NoStop}%
\bibitem [{\citenamefont {Ji}\ \emph {et~al.}(2021)\citenamefont {Ji},
  \citenamefont {Natarajan}, \citenamefont {Vidick}, \citenamefont {Wright},\
  and\ \citenamefont {Yuen}}]{Ji2021Mip}%
  \BibitemOpen
  \bibfield  {author} {\bibinfo {author} {\bibfnamefont {Z.}~\bibnamefont
  {Ji}}, \bibinfo {author} {\bibfnamefont {A.}~\bibnamefont {Natarajan}},
  \bibinfo {author} {\bibfnamefont {T.}~\bibnamefont {Vidick}}, \bibinfo
  {author} {\bibfnamefont {J.}~\bibnamefont {Wright}},\ and\ \bibinfo {author}
  {\bibfnamefont {H.}~\bibnamefont {Yuen}},\ }\bibfield  {title} {\bibinfo
  {title} {Mip* = re},\ }\href {https://doi.org/10.1145/3485628} {\bibfield
  {journal} {\bibinfo  {journal} {Commun, ACM}\ }\textbf {\bibinfo {volume}
  {64}},\ \bibinfo {pages} {131} (\bibinfo {year} {2021})}\BibitemShut
  {NoStop}%
\bibitem [{\citenamefont {Ren}\ \emph {et~al.}(2026)\citenamefont {Ren},
  \citenamefont {Li},\ and\ \citenamefont {Qi}}]{Ren2026MerLean}%
  \BibitemOpen
  \bibfield  {author} {\bibinfo {author} {\bibfnamefont {Y.}~\bibnamefont
  {Ren}}, \bibinfo {author} {\bibfnamefont {J.}~\bibnamefont {Li}},\ and\
  \bibinfo {author} {\bibfnamefont {Y.}~\bibnamefont {Qi}},\ }\href
  {https://arxiv.org/abs/2602.16554} {\bibinfo {title} {{{MerLean}}: An agentic
  framework for autoformalization in quantum computation}} (\bibinfo {year}
  {2026}),\ \Eprint {https://arxiv.org/abs/2602.16554} {arXiv:2602.16554
  [cs.LO]} \BibitemShut {NoStop}%
\bibitem [{\citenamefont {Li}\ \emph {et~al.}(2026)\citenamefont {Li},
  \citenamefont {Zhu},\ and\ \citenamefont {Ren}}]{Li2026MerLeanprover}%
  \BibitemOpen
  \bibfield  {author} {\bibinfo {author} {\bibfnamefont {J.}~\bibnamefont
  {Li}}, \bibinfo {author} {\bibfnamefont {Z.}~\bibnamefont {Zhu}},\ and\
  \bibinfo {author} {\bibfnamefont {Y.}~\bibnamefont {Ren}},\ }\href
  {https://arxiv.org/abs/2605.26959} {\bibinfo {title} {{{MerLean-prover}}: A
  recursive looping harness for end-to-end lean 4 theorem proving}} (\bibinfo
  {year} {2026}),\ \Eprint {https://arxiv.org/abs/2605.26959} {arXiv:2605.26959
  [cs.LO]} \BibitemShut {NoStop}%
\bibitem [{\citenamefont {Ehatamm}\ \emph {et~al.}(2026)\citenamefont
  {Ehatamm}, \citenamefont {Lee}, \citenamefont {Wu},\ and\ \citenamefont
  {Tao}}]{Ehatamm2026Endtoend}%
  \BibitemOpen
  \bibfield  {author} {\bibinfo {author} {\bibfnamefont {M.}~\bibnamefont
  {Ehatamm}}, \bibinfo {author} {\bibfnamefont {Y.}~\bibnamefont {Lee}},
  \bibinfo {author} {\bibfnamefont {X.}~\bibnamefont {Wu}},\ and\ \bibinfo
  {author} {\bibfnamefont {R.}~\bibnamefont {Tao}},\ }\href
  {https://arxiv.org/abs/2605.16523} {\bibinfo {title} {End-to-end
  formalization of quantum error correction}} (\bibinfo {year} {2026}),\
  \Eprint {https://arxiv.org/abs/2605.16523} {arXiv:2605.16523 [quant-ph]}
  \BibitemShut {NoStop}%
\bibitem [{\citenamefont {Lu}(2025)}]{TeXRA}%
  \BibitemOpen
  \bibfield  {author} {\bibinfo {author} {\bibfnamefont {S.}~\bibnamefont
  {Lu}},\ }\href {https://texra.ai} {\bibinfo {title} {{TeXRA}: A multi-agent
  research assistant for theorists}} (\bibinfo {year} {2025})\BibitemShut
  {NoStop}%
\bibitem [{\citenamefont {Micha{\l}ek}\ and\ \citenamefont
  {Shitov}(2019)}]{Michalek2019Quantum}%
  \BibitemOpen
  \bibfield  {author} {\bibinfo {author} {\bibfnamefont {M.}~\bibnamefont
  {Micha{\l}ek}}\ and\ \bibinfo {author} {\bibfnamefont {Y.}~\bibnamefont
  {Shitov}},\ }\bibfield  {title} {\bibinfo {title} {Quantum version of
  wielandt's inequality revisited},\ }\href
  {https://doi.org/10.1109/TIT.2019.2897772} {\bibfield  {journal} {\bibinfo
  {journal} {IEEE Trans. Inf. Theory}\ }\textbf {\bibinfo {volume} {65}},\
  \bibinfo {pages} {5239} (\bibinfo {year} {2019})}\BibitemShut {NoStop}%
\bibitem [{\citenamefont {Shitov}(2023)}]{Shitov2023Growth}%
  \BibitemOpen
  \bibfield  {author} {\bibinfo {author} {\bibfnamefont {Y.}~\bibnamefont
  {Shitov}},\ }\href {https://vixra.org/abs/2308.0028} {\bibinfo {title}
  {Growth in matrix algebras and a conjecture of perez-garcia, verstraete, wolf
  and cirac}} (\bibinfo {year} {2023})\BibitemShut {NoStop}%
\bibitem [{\citenamefont {{Microsoft
  Corporation}}(2024)}]{MicrosoftCorporation2024Visual}%
  \BibitemOpen
  \bibfield  {author} {\bibinfo {author} {\bibnamefont {{Microsoft
  Corporation}}},\ }\href {https://code.visualstudio.com/} {\bibinfo {title}
  {Visual studio code}} (\bibinfo {year} {2024})\BibitemShut {NoStop}%
\bibitem [{\citenamefont {Ma}\ \emph {et~al.}(2025)\citenamefont {Ma},
  \citenamefont {Li},\ and\ \citenamefont {{the math-xmum
  contributors}}}]{MathXmum2025Brouwer}%
  \BibitemOpen
  \bibfield  {author} {\bibinfo {author} {\bibfnamefont {J.-J.}\ \bibnamefont
  {Ma}}, \bibinfo {author} {\bibfnamefont {K.}~\bibnamefont {Li}},\ and\
  \bibinfo {author} {\bibnamefont {{the math-xmum contributors}}},\ }\href
  {https://github.com/math-xmum/Brouwer} {\bibinfo {title} {Game theory
  formalization in {Lean}: {Brouwer}'s fixed-point theorem via {Scarf}'s
  lemma}} (\bibinfo {year} {2025})\BibitemShut {NoStop}%
\end{thebibliography}%

%======================================================================
% End Matter (part of the Letter) - kept in endmatter.tex
%======================================================================
\clearpage
\appendix
%======================================================================
% End Matter (part of the Letter, NOT the Supplementary Material).
% \input by the main file Draft3.tex.
%======================================================================
% The combined architecture/orchestration figure is pinned as a full-width
% block so it appears at the start of the End Matter; REVTeX figure* floats
% otherwise tend to move it onto a later or isolated float page.
% fig_architecture_combined.tex  -  Architecture (left, 7) + orchestration
% patterns (right, 3), combined into a single full-width (figure*) row so
% both fit one figure in the end matter, separated by a vertical rule
% (a tabular column rule spans the full row height automatically).  The
% architecture panel is the original three-tier layout from
% fig_architecture.tex (human on top, orchestrator, five agents), just
% resized to share the row with the orchestration-pattern panel.
% Requires:  \usepackage{tikz}
%            \usetikzlibrary{arrows.meta, positioning, calc, fit, backgrounds}

% Pinned page-wide block instead of a figure* float: placed at the very top
% of the End Matter's first page (see note in endmatter.tex). The caption
% still numbers as a figure via \def\@captype{figure}.
\onecolumngrid
\begin{center}
% 6pt gutters (8pt overflowed \textwidth by 1.1pt with the 0.63+0.34 panels)
\begin{tabular}{@{}c@{\hspace{12pt}}c@{}}
\resizebox{0.63\linewidth}{!}{%
\begin{tikzpicture}[
    >={Stealth[length=3.2pt]},
    every node/.style ={font=\footnotesize},
    box/.style    ={draw, rounded corners=1.5pt, inner sep=2pt,
                    align=center, font=\scriptsize, minimum height=0.50cm},
    agent/.style  ={box, agentbox, minimum width=1.95cm},
    human/.style  ={box, humanbox},
    sandbox/.style={resourcebox, rounded corners=3pt, inner sep=10pt},
    chip/.style   ={font=\small\itshape, fill=white, inner sep=1.5pt,
                    rounded corners=1pt, draw=black!40, dashed},
    legsw/.style  ={draw, minimum size=0.26cm, inner sep=0pt,
                    rounded corners=0.5pt},
    flow/.style   ={->, semithick, draw=black!80}
  ]

  %% --- Top: human supervisor (grey, not an agent) ---------------------
  \node[human, font=\footnotesize\itshape, minimum width=2.2cm] (human) at (0,2.55)
      {Human \scriptsize(mathematical oversight)};

  %% --- Orchestrator (lead AI agent) -----------------------------------
  \node[agent, minimum width=3.6cm] (orch) at (0,1.40)
      {\texttt{leanOrchestrator}\\\scriptsize planning, delegation, integration};

  %% --- Five front-line AI agents --------------------------------------
  \node[agent] (scout) at (-5.4,0)
      {\texttt{leanSearch}\\\scriptsize Mathlib search};
  \node[agent] (proof) at (-2.7,0)
      {\texttt{lean}\\\scriptsize proofs, sorry closure};
  \node[agent] (clean) at (0,0)
      {\texttt{leanSimplifier}\\\scriptsize lint, simplify};
  \node[agent] (bp) at (2.7,0)
      {\texttt{leanBlueprint}\\\scriptsize \LaTeX{} sync};
  \node[agent] (audit) at (5.4,0)
      {reviewer\\\scriptsize hypothesis audit};

  %% --- Shared sandbox (drawn behind so it does not overpaint nodes);
  %%     a phantom coordinate below the agents gives the chip extra
  %%     room without pushing the top edge into the human box ---------
  \coordinate (padbottom) at ($(scout.south)+(0,-16pt)$);
  \begin{pgfonlayer}{background}
    \node[sandbox, fit=(orch)(scout)(audit)(padbottom)] (sandbox) {};
  \end{pgfonlayer}
  \node[chip, anchor=south]
      at ([yshift=-5pt]sandbox.south)
      {shared by every agent: persistent memory $\cdot$ Lean 4 server (\texttt{diagnostics}, \texttt{inspect}, \texttt{loogle})};

  %% --- Legend: role is encoded by box fill/outline (sits in the empty
  %%     upper-right of the sandbox interior) ----------------------------
  \node[legsw, agentbox] (lg1) at (4.0,1.6) {};
  \node[anchor=west, font=\scriptsize] (lg1t) at (lg1.east) {AI agent};
  \node[legsw, humanbox,
        below=3pt of lg1.south, anchor=north] (lg2) {};
  \node[anchor=west, font=\scriptsize] (lg2t) at (lg2.east) {human supervisor};
  \node[legsw, resourcebox,
        below=3pt of lg2.south, anchor=north] (lg3) {};
  \node[anchor=west, font=\scriptsize] (lg3t) at (lg3.east) {shared resource};
  \begin{pgfonlayer}{background}
    \node[draw=black!20, rounded corners=2pt, fill=white, inner sep=4pt,
          fit=(lg1)(lg1t)(lg2t)(lg3)(lg3t)] {};
  \end{pgfonlayer}

  %% --- Control edges (solid): human -> orch -> agents -----------------
  \draw[flow] (human) -- (orch);
  \foreach \n in {scout,proof,clean,bp,audit}
      \draw[flow] (orch.south) -- (\n.north);
\end{tikzpicture}%
}
&
\resizebox{0.34\linewidth}{!}{%
\begin{tikzpicture}[
        >={Stealth[length=3.2pt]},
        every node/.style={font=\footnotesize},
        box/.style   ={draw, rounded corners=1pt, minimum width=1.05cm,
                minimum height=0.42cm, inner sep=1.5pt, font=\footnotesize},
        agent/.style ={box, agentbox},
        % white fill: a grey fill here read as the legend's "human" encoding
        op/.style    ={box, draw=black!70, fill=white, font=\footnotesize\itshape},
        flow/.style  ={->, semithick, draw=black!75},
        caplbl/.style={font=\footnotesize\bfseries},
        sub/.style   ={font=\footnotesize\itshape, text=black!65},
        edgelbl/.style={font=\scriptsize, inner sep=1pt, fill=white}
    ]

    %% --- Panel A: Parallel dispatch (left half) ------------------------
    \begin{scope}[local bounding box=panelA]
        \node[agent] (orchA) at (0,1.25) {orchestrator};
        % 1.22cm pitch: at the old 1.10 the three boxes nearly abutted after
        % the 0.34\linewidth resize and read as one subdivided bar
        \node[agent] (a1) at (-1.22,0.20) {$T_1$};
        \node[agent] (a2) at ( 0.00,0.20) {$T_2$};
        \node[agent] (a3) at ( 1.22,0.20) {$T_3$};
        \node[op] (mg) at ( 0.00,-0.85) {merge};
        \foreach \n in {a1,a2,a3} \draw[flow] (orchA) -- (\n);
        \foreach \n in {a1,a2,a3} \draw[flow] (\n) -- (mg);
        \node[caplbl] at (0,-1.55) {(a) Parallel dispatch};
        \node[sub]    at (0,-1.92) {disjoint scopes};
    \end{scope}

    %% --- Panel B: Scout-then-prove (top right); named as in the End
    %%     Matter text, which cites the panels individually ---------------
    \begin{scope}[shift={(3.95,0.95)}, local bounding box=panelB]
        \node[caplbl] at (0,0.70) {(b) Scout--then--prove};
        \node[agent] (sc) at (-1.05,0) {scout};
        \node[agent] (im) at ( 1.05,0) {prover};
        % unfilled label lifted clear of the box tops: a white-filled 'memo'
        % here erased the corners of the scout/implementer outlines
        \draw[flow] (sc) -- node[edgelbl, fill=none, above=3pt]{memo} (im);
        \node[sub] at (-1.05,-0.40) {cheap};
        \node[sub] at ( 1.05,-0.40) {expensive};
    \end{scope}

    %% --- Panel C: Review-repair loop (bottom right) --------------------
    \begin{scope}[shift={(3.95,-1.05)}, local bounding box=panelC]
        \node[caplbl] at (0,0.75) {(c) Review--repair loop};
        \node[agent] (aud) at (-0.95,0) {reviewer};
        \node[agent] (fix) at ( 0.95,0) {fixer};
        \draw[flow] (aud) to[bend left=28]
        node[edgelbl, above]{issues} (fix);
        \draw[flow] (fix) to[bend left=28]
        node[edgelbl, below]{patch} (aud);
        \draw[flow, dashed] (aud.south) -- ++(0,-0.55)
        node[edgelbl, below, fill=none]{accept};
        \node[sub] at (0.55,-0.85) {$\le 5$ iterations};
    \end{scope}

\end{tikzpicture}%
}
\end{tabular}
\end{center}
\begingroup
\makeatletter\def\@captype{figure}\makeatother
\caption{%
    System architecture and orchestration patterns.
    \emph{Left}: a human
    supervisor sets mathematical intent; the orchestrator dispatches to
    four specialized interactive agents: a library scout (\texttt{leanSearch}), a
    proof writer (\texttt{lean}), a simplifier (\texttt{leanSimplifier}), and
    a blueprint synchronizer (\texttt{leanBlueprint}), together with an
    automated reviewer that also runs on each change proposed to the shared
    repository. All interactive agents
    (including the orchestrator) share a persistent memory and the
    Lean~4 server's \texttt{diagnostics}, \texttt{inspect}, and
    \texttt{loogle} tools (Supplementary Material~\cite{Supp}), drawn as a single
    dashed region rather than one connection per agent per service.
    \emph{Right}: the three recurring orchestration patterns.
    \textbf{(a)} The orchestrator dispatches independent subtasks
    $T_1,\dots,T_n$ to workers on disjoint scopes and merges their
    outputs. \textbf{(b)} A cheap scout surveys the library and emits a
    design memo which a stronger prover consumes, so the expensive
    proof agent is invoked only after a feasible route is
    identified. \textbf{(c)} A
    reviewer flags issues and a fixer applies patches; the loop
    terminates when the reviewer accepts (dashed) or after a cap of
    five iterations.%
}\label{fig:architecture}\label{fig:orchestration}
\endgroup
\vspace{0.5\baselineskip}
\twocolumngrid

\section*{Method: Multi-agent autoformalization}
%======================================================================

The formalization was produced by a team of specialized language-model agents coordinated through the shared blueprint and a persistent memory; here an agent is a language-model program that edits files and responds to Lean's type-checking feedback until its task is complete. This End Matter describes the agent roles, the blueprint and memory that coordinate them, and the software built to run them; the Supplementary Material~\cite{Supp} expands each part, reproduces the agent prompts, and reports the full cost analysis.

\prlsection{Multi-agent architecture and implementation}
The proof development is organized around six specialized agent roles: proof writer, library scout, simplifier, blueprint synchronizer, orchestrator, and reviewer; the last runs automatically on each change proposed to the shared repository. The other five agents work in interactive sessions and share the blueprint and a persistent memory of notes on what has been tried, what worked, and what failed; \cref{fig:architecture} shows the resulting system architecture. Each of these five agents interacts with Lean~4 through file editing and real-time type-checking feedback, the same interface a human developer uses in editors such as VS Code~\cite{MicrosoftCorporation2024Visual}. Proof writing is assigned to the most capable, most expensive language models, since proof errors propagate downstream; scouting, cleanup, and auditing use less expensive models, since failed attempts are cheap to retry; orchestration uses either class of model, according to the complexity of the task. The resulting usage by role is shown in \cref{fig:role_model_main}, and \cref{fig:model_gantt_main} shows when each model was in use over the project; the full cost analysis appears in the Supplementary Material~\cite{Supp}.

These six roles were not fixed from the start: whenever a recurring bottleneck appeared, often reported by the agents themselves, the supervisor introduced a new role for the missing function, and the full set was in place within the first six weeks of the project.

The agents interact through a small set of recurring coordination patterns: parallel dispatch, shown in panel~(a) of \cref{fig:orchestration}, where the orchestrator sends out many independent tasks; scout-then-prove, shown in panel~(b) of \cref{fig:orchestration}, where a scout runs an inexpensive library search before the prover attempts a costly proof; and review-repair cycles, shown in panel~(c) of \cref{fig:orchestration}, where a reviewer-fixer pair cycles through review and repair. Scout-then-prove was reused most often: a fast, inexpensive agent first identifies the relevant existing results in Mathlib, and a more capable agent then uses them to construct the proof. The remaining agents audit and reorganize in review-repair cycles: the blueprint synchronizer keeps the blueprint and the Lean code in step, the simplifier compresses long proofs, and the reviewer cross-checks claims against the primary literature.

The system we build and use is not an off-the-shelf coding assistant like Claude Code or Codex.
The system is built from the language-model interfaces (API calls) as a software that runs the agents, TeXRA~\cite{TeXRA}; the Lean-specific roles, tools, and prompts are additions built for this project.
The agent prompts and role definitions are released with TNLean~\cite{TNLean}; TeXRA itself is available separately, and the system prompts of the interactive roles are reproduced in the Supplementary Material~\cite{Supp}.

\begin{figure}[!htb]
    \centering
    \includegraphics[width=\linewidth]{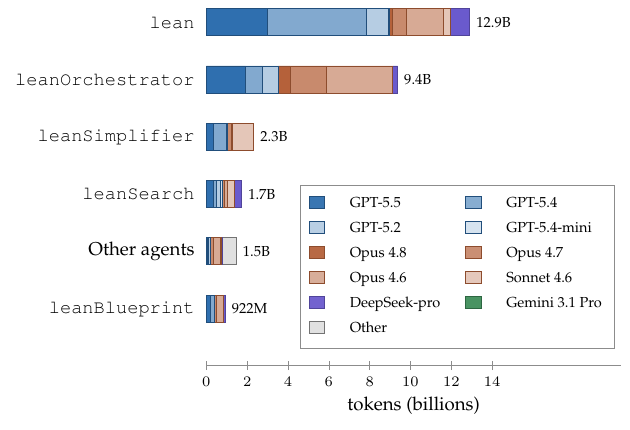}
    \caption{Token usage by agent role, each bar broken down by the
        language models that role used. Proof writing and orchestration
        dominate, consistent with assigning the most capable models to the
        tasks where an error is most expensive.}\label{fig:role_model_main}
\end{figure}

\begin{figure}[!htb]
    \centering
    \includegraphics[width=\linewidth]{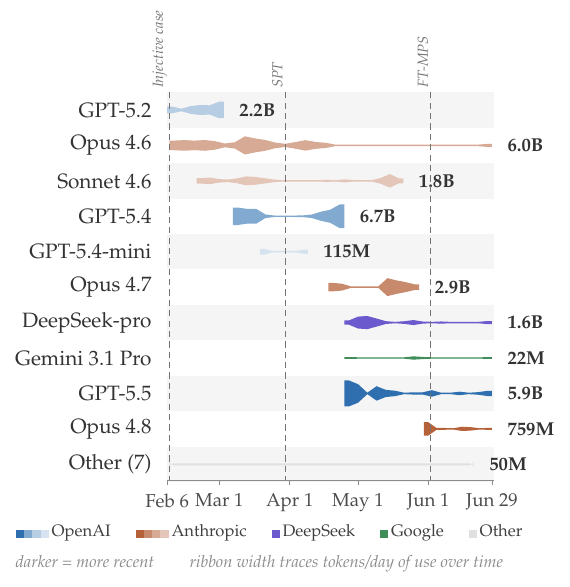}
    \caption{Model usage over the project. Newer model releases
        progressively replaced older ones within each provider family, and the
        token usage concentrated on the most capable models available at each
        time. Each model's period of use is shown as a ribbon, sorted by first
        use: one color family per provider, shade distinguishing models within a
        family, and ribbon width tracing the daily token usage on one scale
        shared across every model; the number at the end of each ribbon is that
        model's total. The dashed vertical lines mark the dates the injective
        case of FT-MPS, the SPT cohomological invariant, and the FT-MPS proof
        chains reported in this work became fully verified.}\label{fig:model_gantt_main}
\end{figure}

\prlsection{Blueprint, memory, and human oversight}
Two shared elements support this architecture: the blueprint and the persistent memory. Every definition, lemma, and theorem in the blueprint carries a proof sketch, a proof status, and its logical dependencies, and is linked to the Lean declarations that formalize it~\cite{Massot2021Leanblueprint}. A theorem stated across several sources in different notations is written once in the blueprint and then proved in Lean, so the blueprint sits between the primary literature and the Lean code (\cref{fig:alignment}). The blueprint is publicly available~\cite{TNLeanBlueprint}, and the Supplementary Material~\cite{Supp} reproduces the portion of its dependency graph around the fundamental theorem at print scale. The persistent memory separately records scouting results, proof strategies, counterexamples, and technique notes across sessions (the bounded working periods after which an agent's conversation is set aside and no longer read by later agents); without it, each dead end would be re-investigated from scratch.

The blueprint is also where the project separates \textit{tactical} from \textit{strategic} decisions. Tactical decisions belong to the agents: which Mathlib lemmas to invoke, when to abandon a failed approach, how to decompose a long argument into intermediate lemmas, and what definitions to try for MPS tensors and transfer operators, down to how the \nFilesfull{} source files of the full library are organized. Strategic decisions remain with the human supervisor: whether the formal statement is the theorem intended by the primary literature, and whether its hypotheses are correct. The supervisor acts through blueprint review and by steering the orchestrator (assigning stronger or more economical models to each task, redirecting the orchestrator away from unproductive routes), never by writing Lean code or choosing proof tactics; the supervision effort therefore does not grow with the length of individual proofs, since Lean already checks each proof. The supervisor reviewed the blueprint in \nReviewRounds{} rounds, each assisted by separate review agents that cross-checked chapters against the primary literature~\cite{Perez-Garcia2007Matrix,Cirac2017Matrixa,Cirac2021Matrix,Sanz2010Quantum}; each round surfaced 10--30 issues, chiefly missing hypotheses, mismatched definitions, and overly general claims, of which the doubly-stochastic gauge and the asymptotic reformulation traced in the main text are representative. This division of labor proved effective: most tactical decisions survived blueprint review, and where they did not, the review traced the error to a definition or an unintended hypothesis, which the agents then corrected.

As anticipated in the Introduction, no language-model context window can hold the source literature, the growing blueprint, and the expanding formal library at once; the work is therefore split into the many bounded sessions coordinated above. Within a session, too, the orchestrator's own conversation grows until it must be condensed: older exchanges are replaced by summaries while the list of running agents and open tasks is retained, so durable knowledge resides in the persistent memory rather than in any one conversation. This memory is what lets a fixed context window, and a fixed amount of supervisor effort, span a project far larger than either could otherwise hold.

%======================================================================
% Supplementary Material (separate document; standalone wrapper Draft3SM.tex)
%======================================================================
\clearpage
\onecolumngrid
\begin{center}
    {\large \bf Supplementary Material}
\end{center}

% !TEX root = Draft3SM.tex
%======================================================================
% Supplementary Material body. Shared (\input) by Draft3.tex and the
% standalone wrapper Draft3SM.tex. (End Matter is NOT here; see endmatter.tex.)
%
% Supplementary Material numbering: figures, equations, and tables get an
% "S" prefix (S1, S2, ...). Defining it here means both the combined
% Draft3.tex build and the standalone Draft3SM.tex build pick it up.
%======================================================================
\makeatletter
\setcounter{figure}{0}
\setcounter{equation}{0}
\setcounter{table}{0}
\renewcommand{\thefigure}{S\@arabic\c@figure}
\renewcommand{\theequation}{S\@arabic\c@equation}
\renewcommand{\thetable}{S\@arabic\c@table}
% Give the S-numbered floats their own hyperref anchors; otherwise the
% combined Draft3 build reuses the Letter's figure.N/table.N/equation.N
% destinations and links to S-figures jump to the main-text ones.
\renewcommand{\theHfigure}{S\@arabic\c@figure}
\renewcommand{\theHequation}{S\@arabic\c@equation}
\renewcommand{\theHtable}{S\@arabic\c@table}
\makeatother

This Supplementary Material expands the method summary given in the End Matter of the Letter: the agent
roles and their Lean~4 interface (\cref{app:architecture}), the blueprint that connects source theorems to Lean declarations (\cref{app:blueprint}), the
persistent memory that carries mathematical knowledge across bounded sessions
(\cref{app:memory}), the coordination patterns used by the agents
(\cref{app:orchestration}), and the paper-to-formalization gap analysis
(\cref{app:gaps}). It also records the computational cost (\cref{app:cost}), presents a representative proof excerpt (\cref{app:ftsb_walkthrough}), and reproduces the system prompts that defined the agent roles (\cref{app:prompts}).

% SM contents: list only the appendix sections. The wrappers (Draft3.tex /
% Draft3SM.tex) inject a tocdepth=-10 into the .toc before \maketitle to hide
% everything (title, abstract, acknowledgments, references, End Matter); the
% brackets below raise the depth for the appendices only, then hide again.
\tableofcontents

% Re-enable TOC entries for the appendices only (disabled in the wrapper after
% \maketitle), then disable again so any trailing references stay out.
\let\addcontentsline\SMrealaddcontentsline{}
\appendix

%======================================================================
\section{Agent architecture and Lean interface}\label{app:architecture}
%======================================================================

The formalization system is organized as a human-directed multi-agent architecture. Here an \emph{agent} is a large language model equipped with a system prompt that defines its role and a fixed set of tools that we programmed and supplied, such as editing a file or querying the Lean compiler. Once dispatched, the agent runs autonomously, calling commands and reading their results until it reports completion, failure, or a partial result. When the orchestrator dispatches tasks to sub-agents, each sub-agent reports its findings back to the orchestrator on completion.

The system is implemented in TeXRA~\cite{TeXRA}, which works as a command-line program and Visual~Studio~Code extension; the Lean-specific roles and commands described here are project-specific additions.
A human supervisor works with a pool of such agents through a shared repository of Lean code and \LaTeX{} sources, a persistent memory directory of notes, and the Lean~4 compiler accessed in real time.

\subsection{Agent roles}\label{sec:roles}

The specialized roles listed in \cref{tab:roles} were introduced incrementally over the first month of the project. The earliest sessions used generic, undifferentiated agents; the supervisor introduced each new specialization as the bottlenecks of that setup became visible.
Scouting, proof writing, and proof simplification became distinct roles by late February, and the full role list was codified in a pinned strategy memo in mid-March.
Each agent interacts with the Lean compiler and the file system through a fixed set of commands.

\begin{table*}[htbp]
    \caption{\label{tab:roles}Agent roles: the model tier each typically uses, its function, and the tools it is granted beyond the shared core. All are tool-using agents. The shared core of the interactive roles is the Lean~4 server interface of \cref{sec:lean_tools}, file editing and shell access, and the persistent memory of \cref{app:memory}; the last column lists what each role adds. ``Tier'' is the model class the role typically runs on, expensive (E) or inexpensive (I); roles were not preassigned to models, and the supervisor set the tier per task at dispatch (\cref{sec:model_selection}). The code names correspond to the roles named in the End Matter: \texttt{lean} is the proof writer, \texttt{leanSearch} the library scout, \texttt{leanSimplifier} the simplifier, and \texttt{leanBlueprint} the blueprint synchronizer. The reviewer runs as an automated review agent on each proposed repository change (\cref{sec:audit_fix}); its review prompts are released with TNLean.}
    \small
    \centering
    \setlength{\tabcolsep}{4pt}%
    \begin{tabular}{@{}lcll@{}}
        \toprule
        {Role}                  & {Tier} & {Function}                                                                       & {Tools beyond the shared core}                                                                                    \\
        \midrule
        \texttt{leanOrchestrator} & E      & \begin{tabular}[t]{@{}l@{}}Task decomposition,\\ dispatch\end{tabular}           & \begin{tabular}[t]{@{}l@{}}sub-agent dispatch, run management,\\ repository and code-review access\end{tabular} \\[0.9em]
        \texttt{lean}             & E      & \begin{tabular}[t]{@{}l@{}}Proof writing,\\ closure of placeholders\end{tabular} & ---                                                                                                             \\
        \texttt{leanSearch}       & I      & Mathlib scouting, feasibility                                                    & Web search, arXiv lookup                                                                                        \\
        \texttt{leanSimplifier}   & I      & Style checking, proof simplification                                             & ---                                                                                                             \\
        \texttt{leanBlueprint}    & I      & Blueprint-Lean sync                                                              & Web/arXiv, bibliography                                                                                         \\
        reviewer                  & I      & Literature-Lean cross-checking                                                   & runs on proposed repository changes                                                                             \\
        \bottomrule
    \end{tabular}
\end{table*}

All agents draw on a common core of tools: the Lean~4 server interface of \cref{sec:lean_tools}, file editing and shell access, and the persistent memory of \cref{app:memory}. They differ only in the additional tools they are granted (last column of \cref{tab:roles}), their system prompt, and their model tier.

\subsection{Delegation}\label{sec:delegation}

When the orchestrator, the \texttt{leanOrchestrator} role of \cref{tab:roles}, dispatches a task to a sub-agent (a separate agent session launched to carry out that one task), it fixes the properties of the session: the sub-agent's role, the task instructions and background notes it starts with, the files it works on, and the model tier it runs on. A sub-agent runs in isolation, with no access to the orchestrator's conversation, so it acts only on the instructions, notes, files, and model it is given.

\emph{Attach instructions and background notes.} Beyond the task instructions, the orchestrator may place one or more memory files in the sub-agent's initial context as read-only background: a style guide, a list of statements already known to be false together with their counterexamples, or the memo an earlier scouting agent left behind, so that a later proof-writing agent reuses the Mathlib lemmas the scout located instead of searching again. The life cycle that produces and reuses these notes is described in \cref{app:memory}.

\emph{Assign files.} Each sub-agent is assigned a set of files to edit and a set to read as background; vision-capable sub-agents additionally receive figures, screenshots, or PDFs. Parallel sub-agents are assigned disjoint sets of files to edit, a working convention rather than an enforced restriction, so that no two work on the same part of the Lean source tree (\cref{app:orchestration}). When the task finishes, the sub-agent's output is either copied back into the project by the orchestrator or, for sub-agents editing project files in place, already there.

\emph{Choose the model tier.} The orchestrator selects one model tier per task; the reasoning behind each choice, and the specific models used, are given in \cref{sec:model_selection}.

\subsection{Lean~4 integration and related tools}\label{sec:lean_tools}

Agents interact with Lean~4 through a programming interface that lets external tools query the compiler for errors and type information in real time, without rebuilding the project. The interface drives the same Lean~4 language server that human developers use in editors such as VS~Code~\cite{MicrosoftCorporation2024Visual}: inside the editor it routes through the running Lean~4 extension, and on the command line it spawns its own server process, so the agent sees the same diagnostics and proof state the developer does. It exposes the five tools listed below, each set in typewriter font.

\begin{enumerate}[label={(\roman*)},leftmargin=*]
    \item \texttt{lean\_diagnostics} returns the compilation errors, warnings, and hints for a given file, either in full or as severity counts only. It checks whether a proof step compiles, and is the first tool invoked after a proof attempt.

    \item \texttt{lean\_inspect} reports the proof state and available declarations at a specified position, given by line and column, in one of three modes. The \texttt{hover} mode reproduces the information shown when a user places the cursor over an identifier in the Lean editor: its full type and any associated documentation. For example, inspecting the single-block fundamental theorem returns a signature stating that an injective tensor $A$ and a tensor $B$ generating the same matrix product vectors are gauge-equivalent. The \texttt{goal} mode returns the current tactic state, namely the assumptions in scope and the goal that remains to be proved. The \texttt{term\_goal} mode returns the type expected at the position within a proof term, that is, within a proof written as an expression rather than as a tactic block.

    \item \texttt{lean\_loogle} searches Mathlib for a named result. It wraps Loogle, a community-developed search engine for Mathlib, querying it over its public interface, and accepts queries by declaration name, by the constants a result mentions, or by type pattern, such as lemmas of the form $f(x+y)=f(x)+f(y)$, several at once in one batched call.

    \item \texttt{lean\_file} runs file-scoped commands that refresh the server when the editor state becomes stale: restarting the Lean server for one file, or performing a lighter refresh of its dependencies, used when a file's diagnostics no longer reflect the latest edits.

    \item \texttt{lean\_project} runs project-wide commands that take no target file and manage the build and the language server: building the library through Lake, fetching the precompiled Mathlib cache (for the whole project or for the current file's imports alone), cleaning build artifacts, restarting or stopping the language server, and installing or selecting the Lean toolchain. The build command does not capture its output, so compilation errors are read back afterward with \texttt{lean\_diagnostics}.
\end{enumerate}

Together, these tools support the same iterative cycle followed by a human Lean user: attempt a proof, inspect the resulting error or goal state, search the library for the missing result, edit the proof, and recheck.

\subsection{Model selection}\label{sec:model_selection}

The system draws on multiple models at different capability and cost tiers, including Claude Opus and Sonnet, the GPT-5 family, DeepSeek, and Gemini, each run at an explicitly chosen maximum reasoning effort that is raised for hard proofs and lowered for routine work. Proof-writing tasks, such as closing a Lean \texttt{sorry} (an unproved placeholder) or repairing a broken proof chain, are assigned to an expensive, high-effort model, since a wrong attempt is costly to detect and can affect later tasks. Scouting, simplification, and blueprint synchronization use cheaper models, where failures are inexpensive to retry; orchestration typically uses higher-tier models such as Claude Opus. Empirically, we found Claude Opus 4.6 to be a better orchestrator than GPT 5.2, which we attribute to differences in token usage (see \cref{fig:role_model}); GPT 5.5 later closed this gap, and its share of the work grew accordingly.
The allocation was not fixed upfront but settled over the project as it became clear which task profiles needed the expensive tier. The resulting cost per model (Fig.~\ref{fig:model_spend}) concentrates on a few high-capability models, with the total cost analyzed in \cref{app:cost}.

\subsection{Context management}

Within a session, the orchestrator's conversation grows until it nears the model's context-window limit, at which point the system compacts it, summarizing the older exchanges to free space (\cref{fig:compaction}). On each compaction the orchestrator is handed back a summary of its still-running sub-agents, any background computations, and its current task list, and the recorded results of completed sub-agents stay retrievable, so it continues the work in progress rather than relaunching it.

% fig_compaction.tex  -  Context management within a session.
%   Context reads top to bottom: the system prompt sits at the top and new
%   turns are appended downward toward the context-window limit at the bottom.
%   Left:  the conversation grows downward until it nearly fills the window.
%   Right: on compaction the older turns collapse into one summary block,
%          the active state is preserved, and a large freed region below
%          lets the session continue.
% Requires: \usepackage{tikz}; \usetikzlibrary{arrows.meta} (loaded in preamble).
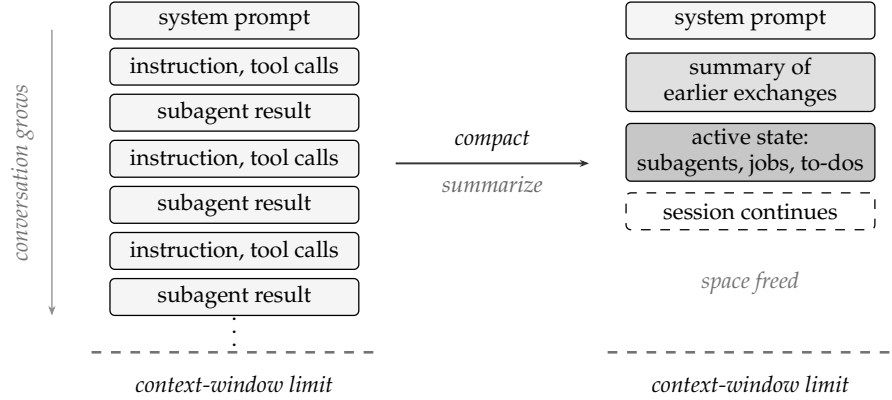
\begin{figure}[htbp]
    \centering
    \resizebox{0.66\textwidth}{!}{%
        \begin{tikzpicture}[
                >={Stealth[length=3pt]},
                blk/.style  ={draw, rounded corners=1.5pt, minimum width=2.7cm,
                        minimum height=0.40cm, inner sep=1.5pt, font=\scriptsize,
                        align=center, fill=black!4},
                sum/.style  ={blk, fill=black!12},
                keep/.style ={blk, fill=black!22},
                cont/.style ={blk, dashed, fill=white},
                lim/.style  ={draw=black!55, dashed, thick},
                lbl/.style  ={font=\scriptsize\itshape, align=center}
            ]

            %% ---- LEFT: conversation grows downward, nearly filling the window -
            \node[blk] (l1) at (0,3.55) {system prompt};
            \node[blk] (l2) at (0,3.05) {instruction, tool calls};
            \node[blk] (l3) at (0,2.55) {subagent result};
            \node[blk] (l4) at (0,2.05) {instruction, tool calls};
            \node[blk] (l5) at (0,1.55) {subagent result};
            \node[blk] (l6) at (0,1.05) {instruction, tool calls};
            \node[blk] (l7) at (0,0.55) {subagent result};
            \node[font=\scriptsize] at (0,0.28) {$\vdots$};
            \draw[lim] (-1.55,-0.05) -- (1.55,-0.05);
            \node[lbl, anchor=north] at (0,-0.18) {context-window limit};
            \draw[->, draw=black!45] (-1.98,3.45) -- (-1.98,0.35);
            \node[lbl, text=black!55, rotate=90] at (-2.30,1.90) {conversation grows};

            %% ---- COMPACT (summarize) -----------------------------------------
            \draw[->, semithick, draw=black!80] (1.75,2.00)
            -- node[lbl, anchor=south] {compact}
            node[lbl, anchor=north, text=black!60] {summarize}
            (3.85,2.00);

            %% ---- RIGHT: compacted; large freed region; session continues -----
            \node[blk]  (r1) at (5.6,3.55) {system prompt};
            \node[sum]  (r2) at (5.6,2.88) {summary of\\earlier exchanges};
            \node[keep] (r3) at (5.6,2.12) {active state:\\subagents, jobs, to-dos};
            \node[cont] (r4) at (5.6,1.48) {session continues};
            \node[lbl, text=black!55] at (5.6,0.70) {space freed};
            \draw[lim] (4.05,-0.05) -- (7.15,-0.05);
            \node[lbl, anchor=north] at (5.6,-0.18) {context-window limit};

        \end{tikzpicture}%
    }
    \caption{\label{fig:compaction}Context management within a session. Context reads top to bottom: the system prompt sits at the top and new turns (task instructions, tool calls, returned subagent results) are appended downward as work proceeds, until they nearly fill the model's context-window limit (left). On compaction the older exchanges are summarized into a single block while the active state (running subagents, background jobs, and the current task list) is carried over, freeing much of the window so the session can continue (right). Durable knowledge is held separately, in the persistent archive of this appendix.}
\end{figure}

%======================================================================
\section{The blueprint}\label{app:blueprint}
%======================================================================

The blueprint is a synchronized mathematical \LaTeX{} document that serves as an intermediate layer between the primary literature and the Lean codebase. Every definition, lemma, and theorem carries machine-readable metadata that links it to the formal proof, records its status, and tracks its logical dependencies. The blueprint toolkit itself is Massot's \texttt{leanblueprint} package~\cite{Massot2021Leanblueprint}; our contribution is the FT-MPS content. We also document some learnings and extensions from the project for using the blueprint as the media for the multi-agent autoformazations below.

The results quoted in this paper, in the Letter as well as in this Supplementary Material, are stated in the informal style of the physics literature; we do not link them one by one to the Lean declarations that formalize them. That correspondence, together with the precise formal hypotheses and the proof status of each statement, lives in the blueprint~\cite{TNLeanBlueprint}, and we encourage the reader who wants the formal counterpart of any result discussed here to read the blueprint alongside this paper.

Two checks keep the layers in agreement through the review-repair cycles of \cref{app:orchestration}:
The \texttt{checkdecls} command from \texttt{leanblueprint} package  checks via a program that every \texttt{\textbackslash{}lean} and \texttt{\textbackslash{}leanok} link to the theorem and lemma labels in the \LaTeX{} version of the blueprint resolves and compiles. For the \texttt{leanBlueprint} agent we also prompt them to write blueprint contents to be consistent with the lean label, and they also get feedback by running \texttt{checkdecls} themselves. The blueprint contents produced by the \texttt{leanBlueprint} agent should match what is formalized in the Lean code. However, sometimes they may get staled. Therefore, we use automated reviewers who check the semantic correspondence that no syntactic tool can verify: that the Lean theorem proves the blueprint statement with hypotheses no stronger than the cited source's.

\subsection{Chapter structure}\label{sec:bp_chapters}

The formalized proof is organized around the canonical-form construction and reduces to three steps:
\begin{enumerate}[label={(\roman*)},leftmargin=*]
    \item \textit{Canonical-form reduction.} The spectral theory of the transfer operator $E_A$ is used to reduce an arbitrary tensor to canonical form (CF). First, quantum Perron-Frobenius theory and Kadison-Schwarz inequalities are used to decompose the virtual space into minimal invariant subspaces of the matrices $\{A^i\}$. This puts the tensor in block-upper-triangular form. The off-diagonal triangular blocks do not contribute to periodic traces and can therefore be discarded, leaving a block-diagonal tensor with irreducible blocks. These blocks may still be \textit{periodic}, in the sense that their transfer maps can have nontrivial peripheral eigenvalues; blocking by a common period removes this periodicity~\cite{Fannes1992Finitely}. The resulting tensor is a block-diagonal direct sum of normal blocks, each with a scalar weight, and is in CF as defined in the main text. One may then choose a further gauge normalization of the CF blocks, giving canonical form~II (CFII) in the terminology of~\cite{Cirac2017Matrixa}. Separately, a normal block becomes injective only after a further finite blocking, with the required blocking length controlled by the quantum Wielandt bound~\cite{Sanz2010Quantum}. This produces the block-injective canonical form (biCF in~\cite{Cirac2017Matrixa}). 

    \item \textit{Block separation.} After the tensors are in CF, the proof compares their normal blocks by first grouping repeated copies of the same normal tensor into a basis of normal tensors. For large enough system size, the MPVs generated by distinct basis elements are linearly independent. Since the two full MPV families are equal, the independent normal-block sectors on the two sides must match. The transfer-operator gap then identifies the possible matches: if the mixed overlap of two normal blocks does not decay, the blocks are gauge-equivalent up to a phase; if the overlap decays, they cannot represent the same sector. In the equal-MPV case, these block phases must be compatible with the scalar weights in the CF decomposition, so that after absorbing the phases into the weights the full block-diagonal tensors are gauge-equivalent. The comparison is made at finite system sizes, not through an asymptotic limit, which is why equal-modulus and oscillating cases, such as GHZ-type phase copies, are included.

    \item \textit{Equal-MPV fundamental theorem.} Combining canonical-form reduction with block separation gives the equal case of the fundamental theorem: two CF representatives generate the same MPV family exactly when their normal-block sector data match and the corresponding blocks are gauge-equivalent after the scalar weights and block phases are aligned.
\end{enumerate}

\cref{fig:cascade} summarizes this reduction. Starting from an arbitrary tensor, one first obtains a block-diagonal tensor with irreducible blocks by discarding the off-diagonal triangular terms. A common blocking removes the periods of these irreducible blocks and gives CF, whose blocks are normal; \cref{fig:block_spectrum} shows this step at the level of a single block's transfer-operator spectrum. A further gauge normalization gives CFII, while a separate finite blocking, controlled by the quantum Wielandt bound, gives the block-injective canonical form (biCF).

% fig_cascade.tex  --  canonical-form reduction (step (i) of the appendix
% proof outline), at the level of the global tensor. The single-block
% spectrum picture behind the +p step is figures/fig_block_spectrum.tex.
% Requires:  \usepackage{tikz, bm}
%            \usetikzlibrary{arrows.meta, positioning, calc, fit, backgrounds}
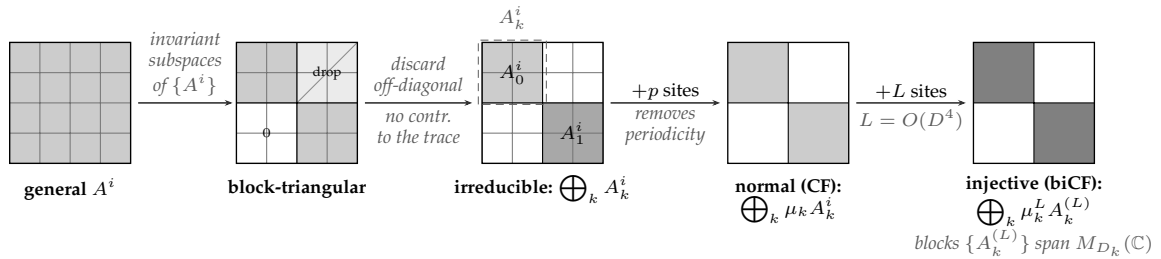
\begin{figure*}[htbp]
    \centering
    \begin{tikzpicture}[
            font=\scriptsize,
            >={Stealth[length=3pt]},
            mfill/.style   ={fill=black!20},
            mfill2/.style  ={fill=black!34},
            mfill3/.style  ={fill=black!50},
            moff/.style    ={fill=black!8},
            mborder/.style ={draw=black, line width=0.5pt},
            grid/.style    ={draw=black!55, line width=0.2pt},
            flow/.style    ={->, semithick, draw=black!80},
            elbl/.style    ={font=\scriptsize, fill=white, inner sep=1pt},
            stage/.style   ={font=\scriptsize\bfseries},
            sub/.style     ={font=\scriptsize\itshape, text=black!65},
            x=1cm, y=1cm
        ]

        \begin{scope}[shift={(0,0)}]
            \fill[mfill] (0,0) rectangle (1.6,1.6);
            \draw[grid,step=0.4] (0,0) grid (1.6,1.6);
            \draw[mborder] (0,0) rectangle (1.6,1.6);
            \node[stage] at (0.8,-0.34) {general $A^i$};
        \end{scope}

        \begin{scope}[shift={(3.0,0)}]
            \fill[mfill]  (0,0.8) rectangle (0.8,1.6);
            \fill[mfill]  (0.8,0) rectangle (1.6,0.8);
            \fill[moff]   (0.8,0.8) rectangle (1.6,1.6);
            \draw[draw=black!45,line width=0.4pt] (0.8,0.8)--(1.6,1.6);
            \draw[grid,step=0.4] (0,0) grid (1.6,1.6);
            \draw[mborder] (0,0) rectangle (1.6,1.6);
            \draw[mborder, line width=0.4pt] (0.8,0)--(0.8,1.6) (0,0.8)--(1.6,0.8);
            \node[font=\tiny] at (1.2,1.2) {drop};
            \node[font=\tiny] at (0.4,0.4) {$0$};
            \node[stage] at (0.8,-0.34) {block-triangular};
        \end{scope}

        \begin{scope}[shift={(6.25,0)}]
            \fill[mfill]  (0,0.8) rectangle (0.8,1.6);
            \fill[mfill2] (0.8,0) rectangle (1.6,0.8);
            \draw[grid,step=0.4] (0,0) grid (1.6,1.6);
            \draw[mborder] (0,0) rectangle (1.6,1.6);
            \draw[mborder, line width=0.4pt] (0.8,0)--(0.8,1.6) (0,0.8)--(1.6,0.8);
            \node[font=\scriptsize] at (0.4,1.2) {$A_0^i$};
            \node[font=\scriptsize] at (1.2,0.4) {$A_1^i$};
            \node[stage] at (0.8,-0.34) {irreducible: $\textstyle\bigoplus_k A_k^i$};
            \draw[densely dashed, draw=black!60] (-0.05,0.78) rectangle (0.85,1.63);
            \node[sub, anchor=south] at (0.4,1.66) {$A_k^i$};
        \end{scope}

        % CF and biCF squares use the block-canonical-form iconography:
        % light plain blocks = normal, solid dark blocks = injective (the
        % block's matrices span M_{D_k}(C)); grids are reserved for matrix
        % entries in the first three squares
        \begin{scope}[shift={(9.5,0)}]
            \fill[mfill] (0,0.8) rectangle (0.8,1.6);
            \fill[mfill] (0.8,0) rectangle (1.6,0.8);
            \draw[mborder] (0,0) rectangle (1.6,1.6);
            \draw[mborder, line width=0.4pt] (0.8,0)--(0.8,1.6) (0,0.8)--(1.6,0.8);
        \end{scope}
        \node[stage, align=center] at (10.3,-0.5)
              {normal (CF):\\$\textstyle\bigoplus_k \mu_k A_k^i$};

        \begin{scope}[shift={(12.75,0)}]
            \fill[mfill3] (0,0.8) rectangle (0.8,1.6);
            \fill[mfill3] (0.8,0) rectangle (1.6,0.8);
            \draw[mborder] (0,0) rectangle (1.6,1.6);
            \draw[mborder, line width=0.4pt] (0.8,0)--(0.8,1.6) (0,0.8)--(1.6,0.8);
        \end{scope}
        \node[stage, align=center] at (13.55,-0.5)
              {injective (biCF):\\$\textstyle\bigoplus_k \mu_k^L A_k^{(L)}$};
        \node[sub] at (13.55,-1.05) {blocks $\{A_k^{(L)}\}$ span $M_{D_k}(\C)$};

        \draw[flow] (1.7,0.8) -- node[sub, above, align=center, yshift=-1pt]
              {invariant\\subspaces\\of $\{A^i\}$} (2.95,0.8);
        \draw[flow] (4.7,0.8) -- node[sub, above, align=center, yshift=-1pt]
              {discard\\off-diagonal}
              node[sub, below, align=center, yshift=1pt]
              {no contr.\\to the trace} (6.12,0.8);
        \draw[flow] (7.95,0.8) -- node[elbl, above]{$+p$ sites}
              node[sub, below, align=center, yshift=1pt]
              {removes\\periodicity} (9.4,0.8);
        % blocking length renamed k -> L (k is the block index in the same
        % row) and the bound corrected to the cited O(D^4)
        \draw[flow] (11.2,0.8) -- node[elbl, above]{$+L$ sites}
              node[sub, below, yshift=1pt]{$L=O(D^4)$} (12.65,0.8);

    \end{tikzpicture}
    \caption{Canonical-form reduction (step~(i) of the proof outline in this
        appendix), at the level of the global tensor. The common invariant
        subspaces of the matrices $\{A^i\}$ put the tensor in block-triangular
        form; the off-diagonal blocks do not contribute to the trace and are
        discarded, giving a block-diagonal tensor with \emph{irreducible}
        blocks $A^i=\bigoplus_k A_k^i$, which need not yet be normal. Blocking
        $p$ sites removes the residual periodicity
        (\cref{fig:block_spectrum}) and gives the canonical
        form (CF), a direct sum of normal blocks, each carrying a scalar
        weight $\mu_k$. A further blocking of $L=O(D^4)$ sites, controlled by
        the quantum Wielandt bound, makes each block injective: the blocked
        matrices $\{A_k^{(L)}\}$ span the block algebra $M_{D_k}(\C)$. The
        result, $\bigoplus_k \mu_k^L A_k^{(L)}$, has the same block-diagonal
        shape as CF and is the block-injective canonical form (biCF).
        }\label{fig:cascade}
\end{figure*}

% fig_block_spectrum.tex  --  the single-block spectrum picture behind the
% +p step of the canonical-form reduction (figures/fig_cascade.tex), plus
% the CFII branch.
% Requires:  \usepackage{tikz, bm}
%            \usetikzlibrary{arrows.meta, positioning, calc, fit, backgrounds}
\begin{figure}[htbp]
    \centering
    \begin{tikzpicture}[
            font=\scriptsize,
            >={Stealth[length=3pt]},
            flow/.style    ={->, semithick, draw=black!80},
            elbl/.style    ={font=\scriptsize, fill=white, inner sep=1pt},
            stage/.style   ={font=\scriptsize\bfseries},
            sub/.style     ={font=\scriptsize\itshape, text=black!65},
            eig/.style     ={fill=black, draw=none},
            x=1cm, y=1cm
        ]

        \node[sub] at (2.35,0.55) {spectrum of $E_{A_k}$ in the unit disk};

        \coordinate (P) at (0.7,-1.05);
        \draw[draw=black!35,line width=0.2pt] ($(P)+(-0.7,0)$)--($(P)+(0.7,0)$)
                                              ($(P)+(0,-0.7)$)--($(P)+(0,0.7)$);
        \draw[thin] (P) circle (0.55);
        \foreach \a in {90,210,330}{\fill[eig] ($(P)+(\a:0.55)$) circle (0.045);}
        \fill[eig] ($(P)+(25:0.20)$) circle (0.03);
        \fill[eig] ($(P)+(165:0.24)$) circle (0.028);
        \fill[eig] ($(P)+(270:0.16)$) circle (0.026);
        \node[stage] at (0.7,-1.85) {periodic};

        \coordinate (N) at (3.95,-1.05);
        \draw[draw=black!35,line width=0.2pt] ($(N)+(-0.7,0)$)--($(N)+(0.7,0)$)
                                              ($(N)+(0,-0.7)$)--($(N)+(0,0.7)$);
        \draw[thin] (N) circle (0.55);
        \draw[densely dashed, draw=black!50] (N) circle (0.30);
        % the largest subdominant eigenvalue sits ON the dashed circle, so the
        % gap arrow measures exactly |lambda_1|-|lambda_2| of the drawn spectrum
        \fill[eig] ($(N)+(0:0.55)$) circle (0.055);
        \fill[eig] ($(N)+(150:0.17)$) circle (0.03);
        \fill[eig] ($(N)+(215:0.30)$) circle (0.028);
        \fill[eig] ($(N)+(95:0.12)$) circle (0.026);
        % gap measured vertically at the top; label outside the disk so no
        % white box masks the spectrum
        \draw[<->, line width=0.4pt] ($(N)+(90:0.30)$) -- ($(N)+(90:0.55)$);
        \node[font=\tiny] at ($(N)+(-0.28,0.73)$)
              {$|\lambda_1|{-}|\lambda_2|$};
        \node[stage] at (3.95,-1.85) {normal};

        \draw[flow] ($(P)+(0.62,0)$) -- node[elbl, above]{$+p$ sites} ($(N)+(-0.62,0)$);

        \draw[flow] ($(N)+(0.30,0.46)$) to[out=50,in=250]
              node[elbl, right, pos=0.3]{gauge norm.} (5.35,0.20);
        \node[stage] at (5.35,0.48) {CFII};

    \end{tikzpicture}
    \caption{Refining one irreducible block $A_k^i$: the single-block picture
        behind the $+p$ step of \cref{fig:cascade}. The transfer operator
        $E_{A_k}$ of an irreducible block may have several peripheral
        eigenvalues, the roots of unity (periodic); blocking $p$ sites leaves
        a single peripheral eigenvalue $1$, with a gap $|\lambda_1|-|\lambda_2|$
        to the rest of the spectrum, so the block is normal and its transfer
        operator primitive. Separately, a gauge normalization of the CF
        blocks gives canonical form~II (CFII); CFII and biCF are distinct
        refinements of CF and should not be identified with each
        other.}\label{fig:block_spectrum}
\end{figure}
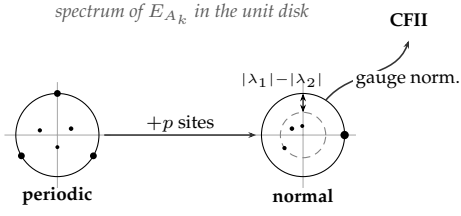

\cref{tab:chapters} lists the blueprint chapters analyzed in this paper. These chapters contain the proof chain for the equal-MPV fundamental theorem, its prerequisites, and the symmetry result used as a first application. The full TNLean repository contains additional material beyond this chain, including channel representations, matrix product operators and density operators, quantum entropy, quantum dynamical semigroups, parent Hamiltonians, exponential decay of correlations, concrete examples, and alternative formulations of MPS fundamental theorems. These parts are useful for the broader TNLean library, but they are not part of the formalization analyzed here. Some deep inputs used only outside the FT-MPS and symmetry chains are currently assumed rather than proved; they do not enter the results discussed in this paper. The chapters are listed in the blueprint's order of presentation, not in a strict prerequisite order.

\begin{table*}[h]
    \caption{\label{tab:chapters}The \nChapters{} blueprint chapters covered in this paper.}
    \begin{tabular}{cl}
        \toprule
        {Ch.} & {Topic}                                               \\
        \midrule
        1        & Introduction                                       \\
        2        & Matrix product vectors                             \\
        3        & Single-block fundamental theorem                   \\
        4        & Quantum channels                                   \\
        5        & Schwarz inequalities and the multiplicative domain \\
        6        & Quantum Perron-Frobenius theory                    \\
        7        & Transfer-operator gap and block separation         \\
        8        & Wielandt bound                                     \\
        9        & Canonical form reduction                           \\
        10       & Basis of normal tensors                            \\
        11       & Proof of the fundamental theorem                   \\
        12       & Symmetries and string order                        \\
        \bottomrule
    \end{tabular}
\end{table*}

\cref{fig:references} summarizes how the source literature was distributed across the blueprint. The point of the figure is not only bibliographic: it shows that the formalization was not driven by a single FT-MPS paper, but by a network of inputs from quantum channels, Perron-Frobenius theory, Wielandt bounds, canonical forms, and symmetry analysis. These sources use slightly different physical and mathematical notations, which the agent team had to reconcile.
This source structure is reflected in the organization of the blueprint chapters and in the separation between operator-theoretic infrastructure and tensor-network-specific arguments.

% fig_references.tex  --  provenance: which primary sources fed which blueprint
% chapters.  Complementary to fig_alignment (the abstract literature/blueprint/
% Lean loop): this one names the actual sources.  Two-column float.
% Chapter numbers follow \cref{tab:chapters} (paper numbering); ch. 1 is the
% introduction and has no source, every other chapter appears in one topic box.
% Requires:  \usepackage{tikz}
%            \usetikzlibrary{arrows.meta, positioning, calc, fit, backgrounds}
\begin{figure*}[htbp]
    \centering
    \begin{tikzpicture}[
            font=\scriptsize,
            >={Stealth[length=3pt]},
            src/.style    ={draw, rounded corners=2pt, align=center, inner sep=3pt,
                    fill=black!5, minimum height=0.78cm, text width=2.8cm},
            topic/.style  ={draw, rounded corners=2pt, align=center, inner sep=3pt,
                    fill=black!8, minimum height=0.72cm, text width=3.5cm},
            we/.style     ={->, semithick, draw=black!75},
            oe/.style     ={->, thin, draw=black!45},
            head/.style   ={font=\footnotesize\bfseries},
            x=1cm, y=1cm
        ]

        % headers sit the same distance above their first box
        \node[head] at (0,6.25) {primary sources};
        \node[head] at (8.0,7.0) {blueprint chapters (paper numbering)};

        %% --- sources (left) ---
        \node[src] (wolf) at (0,5.5)
              {Wolf (2012)\\[1pt]{\scriptsize\itshape Quantum Channels\\ \& Operations}};
        % the 2021 RMP review gets its own box (was bundled with the 2017 MPV
        % paper): it fed most chapters and the string-order paper alone does
        % not cover the SPT material of ch. 12
        \node[src] (pvwc) at (0,4.15) {P\'erez-Garc\'{\i}a et al.\ (2007)\\[1pt]{\scriptsize\itshape the FT-MPS paper}};
        \node[src] (cirac) at (0,3.1) {Cirac et al.\ (2017)\\[1pt]{\scriptsize\itshape the MPV paper}};
        \node[src] (rmp) at (0,2.05) {Cirac et al.\ (2021)\\[1pt]{\scriptsize\itshape the RMP review}};
        \node[src] (sanz) at (0,1.0) {Sanz et al.\ (2010)\\[1pt]{\scriptsize\itshape quantum Wielandt}};
        \node[src] (so) at (0,-0.05) {P\'erez-Garc\'{\i}a et al.\ (2008)\\[1pt]{\scriptsize\itshape string order (0802.0447)}};

        %% --- topics (right): PAPER chapter numbers (\cref{tab:chapters}) ---
        \node[topic] (chan) at (8.0,6.3)  {quantum channels (ch.\,4)};
        \node[topic] (pf)   at (8.0,5.25) {Perron--Frobenius \& Schwarz (ch.\,5--6)};
        \node[topic] (wlt)  at (8.0,4.2)  {Wielandt bound (ch.\,8)};
        \node[topic] (cf)   at (8.0,3.15) {canonical form \& block separation (ch.\,7, 9)};
        \node[topic] (mpv)  at (8.0,2.1)  {matrix product vectors \& single-block theorem (ch.\,2--3)};
        \node[topic] (ft)   at (8.0,1.05) {normal tensors \&\\ fundamental theorem\\ (ch.\,10--11)};
        \node[topic] (sym)  at (8.0,0.0)  {symmetries \& string order (ch.\,12)};

        %% --- edges; multi-input boxes take each edge at a vertically offset
        %%     anchor so the arrowheads do not stack on one point ---
        %% Wolf edges (heavier)
        \draw[we] (wolf.east) to[out=5,in=180]   ([yshift=1.2mm]chan.west);
        \draw[we] (wolf.east) to[out=0,in=180]   (pf.west);
        \draw[we] (wolf.east) to[out=-15,in=180] ([yshift=1.5mm]wlt.west);
        \draw[we] (wolf.east) to[out=-30,in=180] ([yshift=1.5mm]cf.west);

        %% other sources (lighter)
        \draw[oe] (pvwc.east)  to[out=2,in=184]   (wlt.west);
        \draw[oe] (pvwc.east)  to[out=-14,in=180] ([yshift=0.5mm]cf.west);
        \draw[oe] (pvwc.east)  to[out=-28,in=178] ([yshift=1.2mm]mpv.west);
        \draw[we] (pvwc.east)  to[out=-42,in=180] ([yshift=1.2mm]ft.west);
        \draw[oe] (pvwc.east)  to[out=-55,in=180] ([yshift=1.4mm]sym.west);
        \draw[oe] (cirac.east) to[out=2,in=182]   ([yshift=-0.5mm]cf.west);
        \draw[we] (cirac.east) to[out=-16,in=182] ([yshift=-0.2mm]mpv.west);
        \draw[we] (cirac.east) to[out=-32,in=182] ([yshift=-0.2mm]ft.west);
        %% review fan (thin): it fed most chapters, the system's first input;
        %% heavy (principal) credit stays with the theorem papers
        \draw[oe] (rmp.east)   to[out=55,in=186]  ([yshift=-1.2mm]chan.west);
        \draw[oe] (rmp.east)   to[out=22,in=178]  ([yshift=-1.8mm]cf.west);
        \draw[oe] (rmp.east)   to[out=0,in=180]   ([yshift=-1.6mm]mpv.west);
        \draw[oe] (rmp.east)   to[out=-22,in=178] ([yshift=-1.6mm]ft.west);
        \draw[oe] (rmp.east)   to[out=-42,in=178] ([yshift=0.2mm]sym.west);
        \draw[oe] (sanz.east)  to[out=32,in=188]  ([yshift=-1.5mm]wlt.west);
        \draw[we] (so.east)    to[out=0,in=180]   ([yshift=-1.2mm]sym.west);

    \end{tikzpicture}
    \caption{Primary sources and the blueprint chapters they feed, with chapter
        numbers as in \cref{tab:chapters}. Wolf's lecture notes on quantum
        channels and operations~\cite{Wolf2012Quantum} are the principal source
        for the operator-theoretic chapters (quantum channels, the
        Perron--Frobenius and Schwarz theory, the Wielandt bound, and
        canonical-form reduction). The
        matrix-product-state results of P\'erez-Garc\'{\i}a et
        al.~\cite{Perez-Garcia2007Matrix} and of Cirac et
        al.~\cite{Cirac2017Matrixa} supply the tensor-network chapters, and the
        review of Cirac et al.~\cite{Cirac2021Matrix} feeds most chapters; the
        quantum Wielandt inequality of Sanz et al.~\cite{Sanz2010Quantum} feeds
        the Wielandt bound, and the symmetry chapter rests on the string-order
        analysis of P\'erez-Garc\'{\i}a et al.~\cite{Perez-Garcia2008String}
        together with the review. Heavier
        lines mark each chapter's principal sources: Wolf's notes supply the
        operator-theoretic groundwork (ch.~4--9), while the
        matrix-product-vector, fundamental-theorem, and symmetry chapters
        (ch.~2--3 and 10--12) rest on the matrix-product-state
        literature.
        }\label{fig:references}
\end{figure*}
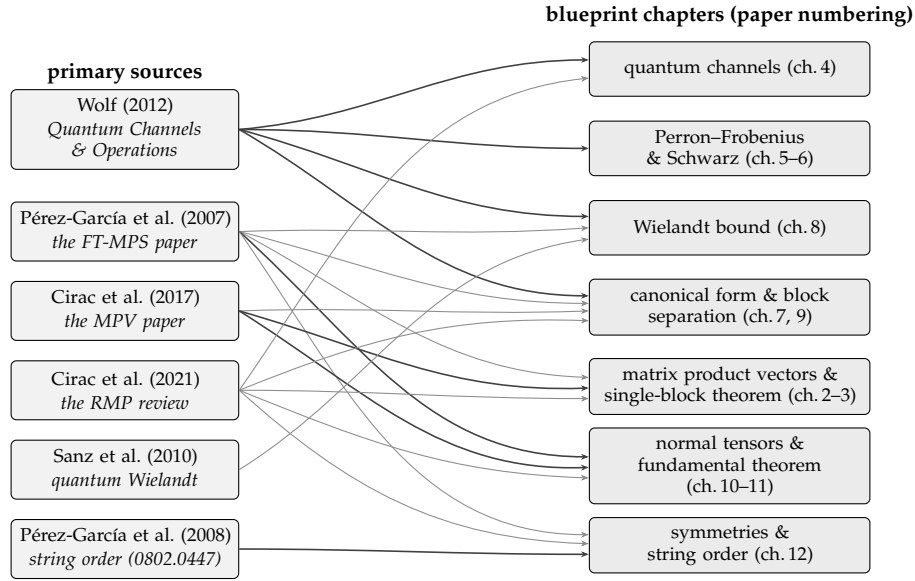

\subsection{Build modes}\label{sec:bp_build}

The same chapter sources compile both as a PDF monograph and as an HTML document with an interactive dependency graph. The PDF provides the conventional mathematical presentation, while the HTML version lets readers inspect theorem dependencies and proof status directly.
The repository README documents the technical setup used to compile both outputs from the shared source files, including Tikz figure handling, common macros, and output-specific formatting choices.

\subsection{The tag system}\label{sec:bp_tags}

The tags below are part of Massot's \texttt{leanblueprint} package~\cite{Massot2021Leanblueprint}; we list them here for completeness.
The \texttt{\textbackslash{}lean} tag links a blueprint item to one or more full Lean declaration names. The \texttt{\textbackslash{}leanok} marker indicates that the corresponding statement or proof has been implemented in Lean, and may appear on the statement, on the proof, or on both. The \texttt{\textbackslash{}notready} marker flags items that are not yet ready for formalization, typically because prerequisite material is still missing. The \texttt{\textbackslash{}uses} tag records logical dependencies between blueprint items. Its arguments are blueprint labels, and a theorem statement and its proof may carry different dependency lists: the statement lists the
definitions needed to parse it, while the proof lists the lemmas the proof
actually invokes.

\subsection{Dependency graph and status tracking}\label{sec:bp_depgraph}

The HTML build of the blueprint~\cite{TNLeanBlueprint} renders all labeled items as nodes in a directed acyclic graph, with edges drawn from the \texttt{\textbackslash{}uses} declarations. Each node is colored according to its formalization status: dark green for theorems whose proofs are implemented and verified, green for theorems with Lean statements marked as complete, light green for definitions linked to Lean, blue for items that are ready to be stated or proved, and orange for items marked as not ready. The graph shows where work can proceed: orange nodes identify the next layer of prerequisite work. These status are generated from the tags in the blueprint and the dependency relations of the theorem and lemma labels.

\cref{fig:depgraph} shows the part of the formalization dependency graph surrounding the Fundamental Theorem, redrawn for print. Green nodes are Lean-verified results, using the same status convention as the dark-green ``proof verified'' nodes in the HTML graph; the grey node marks an intended application outside the present formalization. The lower part of the figure displays the theorem's prerequisites. In particular, the canonical-form argument depends on two largely independent inputs: quantum Perron-Frobenius theory, which supplies the spectral fixed-point and primitivity results, and the Wielandt bound, which supplies the finite-length spanning estimates. The upper part records what the theorem is then used to obtain: the virtual projective representation and its cohomology class, the formalized first step toward the one-dimensional symmetry-protected topological classification.

% fig_depgraph.tex  --  curated sub-region of the blueprint dependency graph
% around the fundamental theorem.  Requires: \usepackage{tikz};
% \usetikzlibrary{arrows.meta, positioning, calc, fit, backgrounds}.
\begin{figure*}[htbp]
    \centering
    \begin{tikzpicture}[
            >={Stealth[length=3.2pt]},
            % verified (proved) blueprint results are shown in green, using the
            % actual HTML dependency graph's status colours (bpVerified border;
            % bpProofDone/bpChainDone fill, per node, defined in preamble.tex):
            % the lighter fill is the graph's "this proof is formalized" shade,
            % the darker one "this proof and all its ancestors are formalized".
            n/.style    ={draw=bpVerified, rounded corners=2pt, fill=bpProofDone!55, align=center,
                    font=\scriptsize, minimum height=0.72cm, text width=2.5cm, inner sep=3pt},
            % ft depends on the whole verified chain shown here, so it gets the
            % graph's darker "chain done" shade rather than just "proof done".
            % blocking length: L (not k) and O(D^4), matching Figs. 2/S2
            % (cited Sanz et al. bound (D^2-d+1)D^2, the version formalized)
            ft/.style   ={n, draw=bpVerified, fill=bpChainDone!55, line width=1.1pt, text width=2.7cm},
            card/.style ={draw=black!50, rounded corners=2pt, fill=black!3,
                    minimum width=2.6cm, text width=2.2cm, minimum height=0.95cm,
                    align=center, font=\scriptsize, inner sep=3pt},
            up/.style   ={->, semithick, draw=black!65, shorten >=1.5pt, shorten <=1.5pt},
            x=1cm, y=1cm
        ]

        \node[n] (cesaro)  at (0,0)      {Ces\`aro fixed point};
        \node[n] (schwarz) at (3.6,0)    {Kadison--Schwarz inequalities};
        \node[n] (wiel)    at (7.2,0)    {Wielandt bound, blocking length $L=O(D^4)$};
        \node[n] (qpf)     at (0,1.15)   {quantum Perron--Frobenius theory};
        \node[n] (canon)   at (3.6,2.3)  {canonical-form reduction (CF\,$\to$\,biCF)};
        \node[n] (block)   at (1.9,3.45) {block separation (exact matching)};
        \node[n] (bnt)     at (5.7,3.45) {basis of normal tensors};
        \node[ft] (ft)     at (3.6,4.6)  {\textbf{fundamental theorem of MPS} (multi-block)};
        % wider than the other "n" boxes: wraps to 2 lines instead of 3, so
        % its top edge does not crowd the "intended application" label above
        \node[n, text width=3.4cm] (single)  at (9.2,4.6)  {single-block theorem (Skolem--Noether)};
        % Invariant stated in H^2(G,\C^\times), matching the formalized quotient
        % (blueprint ch. 12); the U(1) form follows after unitary normalization.
        \node[n] (spt)     at (3.6,6.15) {SPT phase invariant $[\omega]\in H^2(G,\C^\times)$};

        \node[card] at (9.46,6.4) {};
        \node[card] at (9.33,6.27) {};
        \node[card] (land) at (9.2,6.15) {1D SPT phases:\\ AKLT, cluster,\\ Haldane, \dots};
        % annotation sits beside the card stack (not above it), so the top
        % tier does not need extra vertical clearance for a two-line label
        \node[font=\scriptsize\itshape, text=black!65, align=left, anchor=west]
              at (11.05,6.27) {a family of phases,\\ one per cohomology class};

        \draw[up] (cesaro)  -- (qpf);
        \draw[up] (qpf)     -- (canon);
        \draw[up] (schwarz) -- (canon);
        \draw[up] (wiel)    -- (canon);
        \draw[up] (canon)   -- (block);
        \draw[up] (canon)   -- (bnt);
        \draw[up] (block)   -- (ft);
        \draw[up] (bnt)     -- (ft);
        \draw[up] (single)  -- (ft);
        \draw[up] (single)  -- (spt);
        % label above the dashed edge (clear of the single->spt diagonal,
        % which approaches spt from below); unfilled, so it doesn't paint
        % over the SPT node's corner or the card border
        \draw[up, dashed] (spt.east) -- node[font=\scriptsize, above, inner sep=2pt, pos=0.5]{intended application} (land.west);

    \end{tikzpicture}
    \caption{A curated sub-region of the blueprint dependency graph around the
        fundamental theorem, redrawn for print. Each green box is a formalized
        result; an arrow points from a result to the one that uses it. The verified chain
        runs from the Ces\`aro fixed point and the Kadison--Schwarz inequalities,
        through quantum Perron--Frobenius theory and the Wielandt bound, to
        canonical-form reduction,
        block separation, the
        basis of normal tensors,
        and the multi-block fundamental theorem; the single-block theorem feeds
        both the multi-block theorem and the symmetry construction, which
        yields the symmetry-protected-topological phase invariant
        $[\omega]\in H^2(G,\C^\times)$.
        The card stack at the top right is the intended application (dashed arrow), the classification of
        one-dimensional symmetry-protected phases; it is not part of the formalization, which ends at the invariant.
        Quantum Perron--Frobenius theory and the
        Wielandt bound are independent prerequisites of the canonical form: the
        former is spectral, the latter a span-growth argument.}\label{fig:depgraph}
\end{figure*}
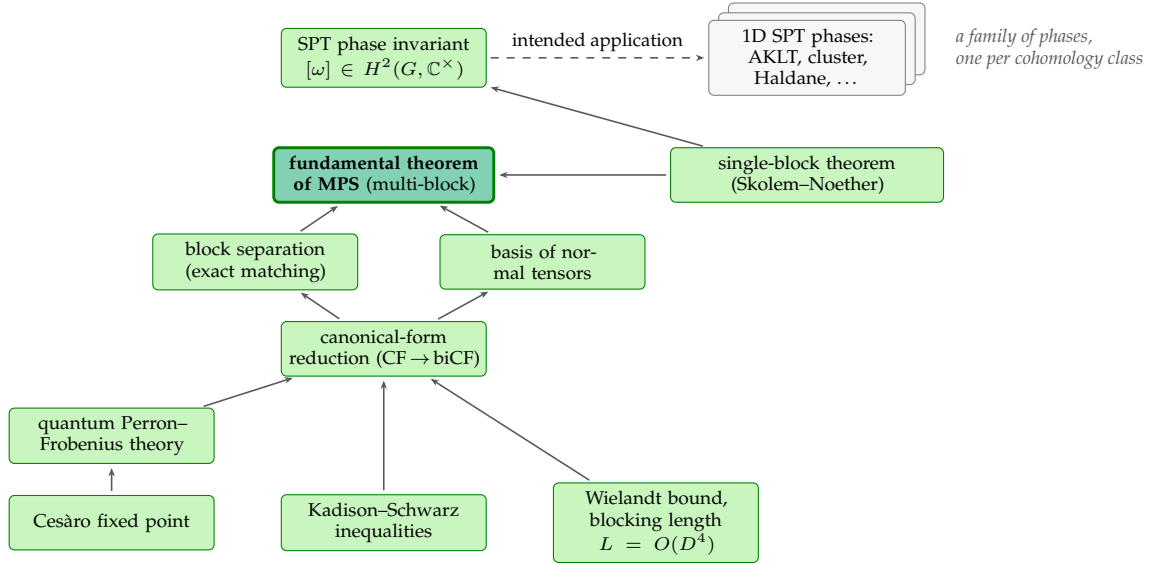

\subsection{Blueprint writing conventions}\label{sec:bp_sync}

Left to its defaults, the agents' writing style is often too obscure for a human reader to audit, and the obscurity has real costs. One episode produced roughly $2{,}000$ lines of Lean code and an entire blueprint chapter devoted solely to the $D=0$ and $N=0$ edge cases; in others, the agents pursued long unintended proof paths whose success could not be judged without additional hypotheses.
We therefore require, and enforce in review, that the blueprint be written as self-contained, human-readable mathematical prose: Lean identifiers must not appear in prose, proof sketches must follow the actual Lean proof structure, and software-engineering vocabulary is avoided in chapter titles and mathematical statements.
Agent-written drafts often violated these conventions, for example by leaking Lean names into prose, replacing statements by source-paper theorem numbers, or using terms such as ``helper lemma'', ``wrapper'',
and ``refactor''. Correcting these issues was one of the recurring tasks during
the \nReviewRounds{} review rounds.

%======================================================================
\section{Persistent memory and project knowledge}\label{app:memory}
%======================================================================

The multi-agent system maintains a persistent project memory across sessions.
These files play the role of accumulated research notes, much like those a graduate student keeps over the course of a project: they record what has been tried, which proof routes worked, which statements were found to be false, and which conventions must be followed. This memory prevents the same mathematical mistakes from recurring across independent agent sessions.
The memory also mediates cross-agent handoffs: when the orchestrator delegates a task, it can place selected memory files in the sub-agent's initial context. The most frequent handoff is scout-then-prove (\cref{sec:scout_impl}), in which the memo written by an inexpensive search agent is attached to the context of the proof-writing agent that follows, so that the Mathlib search is done once rather than repeated in every session.
The full session archive, and the distilled files described below, are kept in the project's persistent memory system; a curated selection of the ones judged most broadly reusable is released with TNLean under \texttt{docs/audits/} and \texttt{docs/paper-gaps/} so that the reusable parts of the process can be inspected and reused.

\subsection{Organization}\label{sec:mem_arch}

The knowledge base consists of roughly 450 Markdown files, accumulated across the project's successive working directories. Files are named by task and date, for example \texttt{2026-02-07\_initial-mathlib-audit.md} or \texttt{aperiodic-ft-assembly-gap-2026-06-02.md}, so that retrieval is by filename convention and keyword search rather than by deep hierarchy. Many files also carry metadata recording the agent that last modified them and a timestamp, so that files can be filtered by role and time period.

\subsection{Categories of memory files}

The archive falls into seven principal categories, summarized in \cref{tab:memory_types}. The category counts are approximate and not exhaustive; the remaining files are cross-cutting analyses and one-off notes.

\begin{table*}[h]
    \caption{\label{tab:memory_types}Memory categories with approximate counts and representative filename patterns.}
    \begin{tabular}{lcl}
        \toprule
        {Category}                       & {Count}  & {Typical patterns}                          \\
        \midrule
        Progress / status                & $\sim$80 & \texttt{session\_*}, \texttt{*\_progress*}  \\
        Design / scout                   & $\sim$60 & \texttt{*\_plan*}, \texttt{*\_scout*}       \\
        Cleanup / reorganization         & $\sim$50 & \texttt{*\_cleanup*}, \texttt{*\_linter*}   \\
        Audit / review                   & $\sim$40 & \texttt{*\_audit*}, \texttt{*\_crosscheck*} \\
        Blueprint sync                   & $\sim$30 & \texttt{blueprint\_*}, \texttt{chXX\_*}     \\
        Errata / counterexample          & $\sim$15 & \texttt{counterexample\_*}                  \\
        Pinned references / style guides & 10       & style guides, technique compendia           \\
        \bottomrule
    \end{tabular}
\end{table*}

Each memory file also carries a label identifying the agent role that last
modified it. The distribution across role families is shown in
\cref{fig:by_labels}. The proof-writing agent and the orchestration family
together account for more than half of all entries, consistent with how much of
the work is proof construction and task coordination.

% fig_creator_stats.tex  -  Memory files by normalized role family.
%   Horizontal bar chart, sorted descending.  Dark bars: roles in the
%   specialized-role taxonomy (\cref{tab:roles}).  Light bars: labels
%   outside it (chat, untagged bootstrap files, and a residual other
%   group: presenter 3, review sessions 2, user 1).
%   Bar length encodes file count; scale S = 0.034 cm per file (so the
%   x-tick at 150 sits at 5.10 cm).  Total = 453 (sum of the bars below).
% Requires: \usepackage{tikz} (loaded in preamble).
\begin{figure}[htbp]
\centering
\resizebox{0.66\textwidth}{!}{%
\begin{tikzpicture}[
    role/.style ={fill=black!55, draw=black!70, thin},
    other/.style={fill=black!18, draw=black!70, thin},
    cat/.style  ={font=\scriptsize, anchor=east},
    val/.style  ={font=\scriptsize, anchor=west},
    ax/.style   ={draw=black!45}
  ]

  %% ---- bars (y top to bottom, descending count) --------------------
  \filldraw[role]  (0,3.81) rectangle (5.542,4.11);
    \node[cat] at (-0.12,3.96) {\texttt{lean}};            \node[val] at (5.622,3.96) {163};
  \filldraw[role]  (0,3.37) rectangle (3.162,3.67);
    \node[cat] at (-0.12,3.52) {\texttt{leanOrchestrator}};\node[val] at (3.242,3.52) {93};
  \filldraw[role]  (0,2.93) rectangle (1.904,3.23);
    \node[cat] at (-0.12,3.08) {\texttt{leanSearch}};      \node[val] at (1.984,3.08) {56};
  \filldraw[role]  (0,2.49) rectangle (1.496,2.79);
    \node[cat] at (-0.12,2.64) {\texttt{leanSimplifier}};  \node[val] at (1.576,2.64) {44};
  \filldraw[other] (0,2.05) rectangle (1.428,2.35);
    \node[cat] at (-0.12,2.20) {(untagged)};               \node[val] at (1.508,2.20) {42};
  \filldraw[role]  (0,1.61) rectangle (1.088,1.91);
    \node[cat] at (-0.12,1.76) {\texttt{leanBlueprint}};   \node[val] at (1.168,1.76) {32};
  \filldraw[other] (0,1.17) rectangle (0.578,1.47);
    \node[cat] at (-0.12,1.32) {\texttt{chat}};            \node[val] at (0.658,1.32) {17};
  \filldraw[other] (0,0.73) rectangle (0.204,1.03);
    \node[cat] at (-0.12,0.88) {(other)};                  \node[val] at (0.284,0.88) {6};

  %% ---- axes --------------------------------------------------------
  \draw[ax] (0,0.66) -- (0,4.20);
  \draw[ax] (0,0.66) -- (5.55,0.66);
  \foreach \x/\t in {0/0, 1.70/50, 3.40/100, 5.10/150} {
    \draw[ax] (\x,0.66) -- (\x,0.58);
    \node[font=\scriptsize, text=black!60, anchor=north] at (\x,0.57) {\t};
  }
  \node[font=\scriptsize\itshape, text=black!60, anchor=north] at (2.77,0.26) {memory files};

\end{tikzpicture}%
}
\caption{Memory files by the agent role that last modified them (total $=453$). Dark bars are roles in the specialized-role taxonomy of \cref{tab:roles}; light bars are labels outside it: \texttt{chat}, untagged startup files from before the metadata was recorded, and a residual \texttt{other} group.}\label{fig:by_labels}
\end{figure}
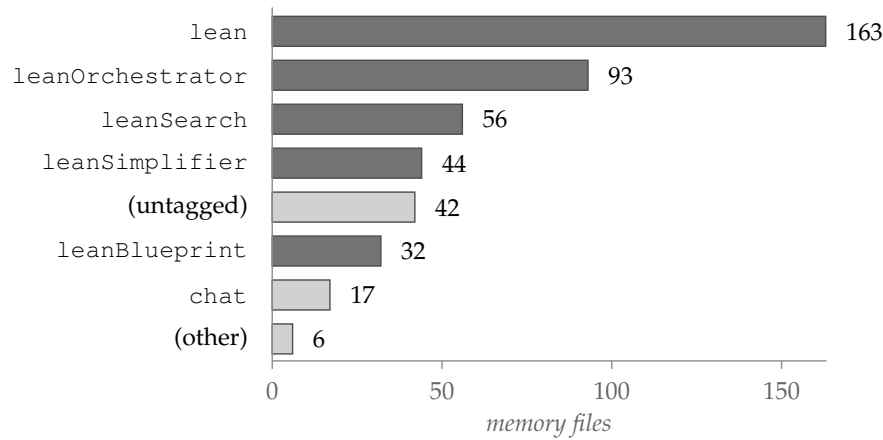

\subsection{From session notes to pinned reference files}\label{sec:knowledge_lifecycle}

Session notes record local operational information, including proof status, build results, unfinished steps, and task completion notes. These notes are periodically distilled by the orchestrator, or by a dedicated agent, into durable lessons: proof techniques, known false statements with counterexamples, Mathlib gap lists, and coordination policies. Distillation is usually requested by the supervisor, and the orchestrator carries a standing instruction to consolidate notes as they accumulate. The most reusable notes are then promoted to pinned status and made available to all future sessions without explicit search. This prevents previously encountered obstacles from being re-investigated independently.

Up to ten memory files may be pinned, and every agent is directed to consult them at the start of a session. This limit reserves pins for durable guidance rather than temporary status notes. At the final project state, the pinned files covered five kinds of information: writing conventions for the blueprint and paper; Lean~4 and Mathlib techniques; source-literature alignment, including lecture-note numbering and coverage audits for Wolf's notes~\cite{Wolf2012Quantum}; review and orchestration policy, including when an automatic reviewer should comment on or approve code changes; and project-specific mathematical guidance, including the paper-vs-formalization gap analysis and the MPS/SPT analysis that was the active effort at that time.

%======================================================================
\section{Orchestration patterns}\label{app:orchestration}
%======================================================================

The orchestrator coordinates sub-agents through three recurring patterns: dispatching independent tasks in parallel, staging proof work through scout-then-prove handoffs, and alternating between review and repair phases. Note that we did not enforce or hint these patterns; the orchestration agent came up withthe orchestrator arrived at them on its own.

\subsection{Parallel dispatch}\label{sec:fanout}

The dominant orchestration pattern is \textit{parallel dispatch}. The orchestrator classifies the remaining work, assigns it to disjoint groups of files, and launches three to five sub-agents concurrently, each with its own file set and task-specific instructions.
The sub-agents run independently while the orchestrator continues supervising other work. As results return, the orchestrator merges the changes and validates the combined state with a full build.
The orchestrator assigns disjoint file sets at dispatch, since two agents editing the same file produce conflicting edits that are expensive to reconcile.

\subsection{Scout-then-prove}\label{sec:scout_impl}

Before attempting a difficult proof, the system deploys an inexpensive scouting agent to assess feasibility. A \texttt{leanSearch} agent searches Mathlib for relevant lemmas, checks whether the needed definition or lemma already exists, and writes a design memo with exact theorem signatures. The memo is not held in the orchestrator's conversation but written to the persistent memory (\cref{app:memory}); the orchestrator reviews it and, if the route is feasible, dispatches an expensive \texttt{lean} agent with that memo attached as read-only context. In practice, finding the right existing Mathlib lemma was often the step that saved the most effort.

\emph{Failure modes.} Several recurring failure modes appear in such sessions. First, an agent may produce a proof that compiles but establishes only a weaker version of the intended theorem; these are caught by the agent team's self review or the human review of the blueprint, as discussed in the main text. Second, a proof term may grow too large for Lean's computation budget: Lean caps the work spent on a single proof by counting internal steps, which it calls ``heartbeats'', and aborts the attempt once their number exceeds the budget, by default $200{,}000$, so the compiler does not hang on a runaway expression. The remedy is to factor the proof into smaller lemmas with named intermediate results; in one case this cut a single proof's count roughly sixfold, from $1{,}600{,}000$ to $250{,}000$ heartbeats. Third, diagnostics can become outdated after edits and require a file-specific server restart. A scouting agent can also propose a lemma whose type does not match the goal; checking the lemma's full type signature with the \texttt{lean\_inspect} command of \cref{sec:lean_tools} catches this before the lemma is used later.

\subsection{Review-repair cycles}\label{sec:audit_fix}

Review and repair phases in alternating fashion. A review agent, typically the automated reviewer or \texttt{leanBlueprint}, identifies discrepancies: paper-vs-Lean mismatches, missing hypotheses, style violations, or blueprint drift. Repair agents address the identified issues, and the review is repeated to verify the fixes and catch regressions. This cycle typically converges in two or three rounds. The most important instance was the paper-vs-Lean cross-check, which discovered several errors during formalization (see \cref{app:gaps}).

A related review step runs when an agent proposes changes to the Lean source tree. An automated reviewer, a language model following a fixed review prompt, examines the proposed changes and posts comments on style violations, proof-integrity issues, or missing documentation; a repair agent then reads the unresolved comments, applies corrections, and submits a revised version. The revised change does not itself trigger a fresh review, which would risk an unending review loop; instead the loop repeats while unresolved comments remain and stops once they are cleared or after a safety cap of five iterations. The review and repair prompts are released with TNLean; the system prompts of the agent roles in \cref{tab:roles} are reproduced in \cref{app:prompts}.

%======================================================================
\section{Paper-to-formalization gap analysis}\label{app:gaps}
%======================================================================

The formalization process revealed several discrepancies between the primary literature and the formal proofs: divergences in proof strategy, differences in hypothesis strength, and gaps between formal correctness and mathematical intent.
The main text briefly discusses two such errors, the doubly-stochastic ``gauge'' and the asymptotic reformulation of the finite statement; this appendix records the remaining case in more detail.
Another mismatch of the same kind concerned repeated copies in the canonical form. Several blocks can represent the same normal tensor up to gauge and phase. If these copies have weights $\mu_{j,1},\ldots,\mu_{j,r_j}$, then their total contribution at length $N$ is multiplied by
$\sum_{q=1}^{r_j}\mu_{j,q}^N$.
The agents had initially treated this factor as if it behaved like a single scalar power, or as if it converged to a nonzero limit after normalization. This is false in general: for two copies with weights $\mu_{j,1}=+1$ and $\mu_{j,2}=-1$, the factor is $1+{(-1)}^N$, which oscillates and vanishes for every odd $N$. The proof therefore covered only a narrower case than the intended FT-MPS\@.

The correction was to avoid a limiting argument at this point. After equivalent normal blocks have been grouped together, the remaining distinct normal blocks produce MPVs that are linearly independent for all sufficiently large system sizes. Equality of the full MPVs therefore forces equality of the corresponding scalar factors, exactly and at each sufficiently large length. These scalar factors are finite sums of powers of nonzero complex numbers. A geometric extrapolation step extends equality from all sufficiently large exponents to all exponents, and Newton-Girard identities then recover the multisets of weights attached to each repeated block. This treats the oscillating case $\mu_{j,1}=+1$, $\mu_{j,2}=-1$ and general repeated-copy cases on the same footing, without assuming convergence.

\subsection{Proof strategy divergences}

The injective-case route via the Skolem--Noether theorem, described in the main text and reproduced in \cref{app:ftsb_walkthrough}, is one example of a divergence between the formal proof and the standard MPS literature; the divergences below cover the remaining cases.

When autoformalization is carried out across different resources with different conventions and terminologies, some care may be needed to ensure that the agents do not conflate these themselves. Otherwise they may lead to possibly extensive detour and failed attempts that can be expensive to fix or retries.
\cref{fig:gauge_innerproduct} clarifies two distinctions we encountered. First, the gauge of the transfer operator can be chosen to be unital/right-canonical or trace-preserving/left-canonical, but not both in general. Second, the proofs use two different pairings on bond operators: the bilinear trace pairing in the algebraic single-block argument, and the sesquilinear Hilbert-Schmidt inner product used to define transfer-operator adjoints in the spectral arguments. The distinction between the block canonical forms CF and biCF is shown in \cref{fig:cascade}: CF has normal blocks, while biCF is obtained only after a further blocking makes each block injective.

% fig_gauge_innerproduct.tex  --  canonical forms and pairings.
% Requires: \usepackage{tikz}; \usetikzlibrary{arrows.meta, positioning, calc, fit, backgrounds}.
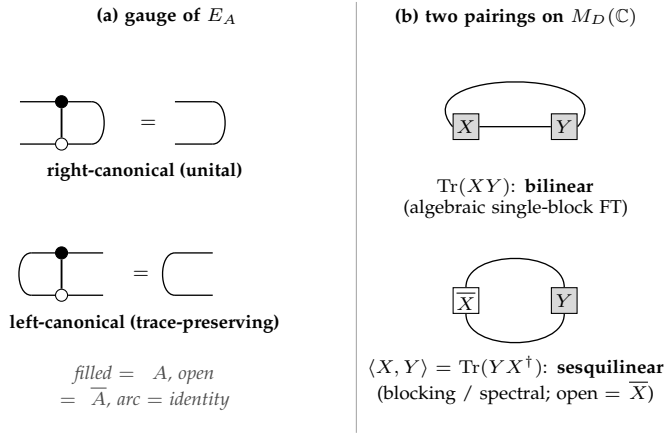
\begin{figure*}[t]
    \centering
    \begin{tikzpicture}[
            font=\scriptsize,
            >={Stealth[length=3pt]},
            % TN conventions as in Figs. 2-3: bonds 0.5pt, physical legs 0.7pt
            bond/.style ={line width=0.5pt},
            phys/.style ={line width=0.7pt},
            tnA/.style  ={draw, fill=black, circle, inner sep=1.6pt},
            tnAc/.style ={draw, fill=white, circle, inner sep=1.6pt},
            opf/.style  ={draw, fill=black!15, minimum size=0.34cm, inner sep=0pt},
            opc/.style  ={draw, fill=white, minimum size=0.34cm, inner sep=0pt},
            nlabel/.style={font=\scriptsize\bfseries, align=center},
            note/.style ={font=\scriptsize\itshape, text=black!70, align=center},
            x=1cm, y=1cm
        ]

        %% ===== Panel (a): GAUGE (TN) =====
        \node[nlabel, text width=3.6cm] at (1.4,5.4) {(a) gauge of $E_A$};

        \begin{scope}[shift={(0,4.0)}]
            \node[tnA] (A) at (0,0.28) {}; \node[tnAc] (Ac) at (0,-0.28) {};
            \draw[phys] (A)--(Ac);
            \draw[bond] (A.west)--(-0.55,0.28); \draw[bond] (Ac.west)--(-0.55,-0.28);
            \draw[bond] (A.east)--(0.4,0.28);   \draw[bond] (Ac.east)--(0.4,-0.28);
            \draw[bond] (0.4,0.28) to[out=0,in=0] (0.4,-0.28);
            \node at (1.15,0) {$=$};
            \draw[bond] (1.5,0.28)--(2.0,0.28); \draw[bond] (1.5,-0.28)--(2.0,-0.28);
            \draw[bond] (2.0,0.28) to[out=0,in=0] (2.0,-0.28);
        \end{scope}
        \node[nlabel] at (1.1,3.35) {right-canonical (unital)};

        \begin{scope}[shift={(0,2.0)}]
            \node[tnA] (B) at (0,0.28) {}; \node[tnAc] (Bc) at (0,-0.28) {};
            \draw[phys] (B)--(Bc);
            \draw[bond] (B.east)--(0.55,0.28); \draw[bond] (Bc.east)--(0.55,-0.28);
            \draw[bond] (B.west)--(-0.4,0.28); \draw[bond] (Bc.west)--(-0.4,-0.28);
            \draw[bond] (-0.4,0.28) to[out=180,in=180] (-0.4,-0.28);
            % '=' at 1.05 (at 1.15 it nearly touched the left-bulging identity arc)
            \node at (1.05,0) {$=$};
            \draw[bond] (1.5,0.28)--(2.0,0.28); \draw[bond] (1.5,-0.28)--(2.0,-0.28);
            \draw[bond] (1.5,0.28) to[out=180,in=180] (1.5,-0.28);
        \end{scope}
        \node[nlabel] at (1.1,1.35) {left-canonical (trace-preserving)};
        \node[note, text width=3.6cm] at (1.1,0.5) {filled $=A$, open $=\bar A$, arc $=$ identity};

        \draw[draw=black!35, line width=0.4pt] (3.9,-0.1) -- (3.9,5.6);

        % Resolved (\SL{cut panel b}): the old panel (b), the CF/biCF block
        % canonical forms, was removed as fully covered by the cascade figure
        % after the notation sync (same squares, shades, and direct sums
        % there); the caption now points to \cref{fig:cascade} instead.

        %% ===== Panel (b): two pairings (TN loops) =====
        \node[nlabel, text width=3.8cm] at (6.0,5.4) {(b) two pairings on $M_D(\C)$};

        \begin{scope}[shift={(6.0,3.95)}]
            \node[opf] (M) at (-0.65,0) {$X$}; \node[opf] (N) at (0.65,0) {$Y$};
            \draw[bond] (M.east) -- (N.west);
            \draw[bond] (N.east) to[out=50,in=130,looseness=1.6] (M.west);
        \end{scope}
        \node[align=center, text width=3.8cm] at (6.0,3.0)
            {$\mathrm{Tr}(XY)$: \textbf{bilinear}\\ (algebraic single-block FT)};

        \begin{scope}[shift={(6.0,1.65)}]
            \node[opc] (X) at (-0.65,0) {$\bar X$}; \node[opf] (Y) at (0.65,0) {$Y$};
            \draw[bond] (X.north) to[out=90,in=90] (Y.north);
            \draw[bond] (X.south) to[out=-90,in=-90] (Y.south);
        \end{scope}
        \node[align=center, text width=4.0cm] at (6.0,0.6)
            {$\langle X,Y\rangle=\mathrm{Tr}(Y X^\dagger)$: \textbf{sesquilinear}\\ (blocking / spectral; open $=\bar X$)};

    \end{tikzpicture}
    \caption{Two notions that the terms ``canonical form'' and ``inner
        product'' overload in this development: (a) The gauge of the transfer
        operator $E_A(X)=\sum_i A^i X {(A^i)}^\dagger$, as tensor-network
        identities (filled $=A$, open $=\bar A$, an arc $=$ the identity):
        closing the right bonds gives the identity on the left, the unital
        right-canonical gauge $\sum_i A^i {(A^i)}^\dagger=I$; closing the left bonds
        gives the identity on the right, the trace-preserving left-canonical
        gauge $\sum_i {(A^i)}^\dagger A^i=I$. A gauge realizes one or the other;
        both at once, the doubly stochastic gauge discussed in the main text, holds only for a non-generic family.
        (b) The two
        pairings on bond operators the proofs use: the bilinear trace pairing
        $\mathrm{Tr}(XY)$, on which the algebraic single-block argument rests, and
        the sesquilinear Hilbert--Schmidt inner product
        $\langle X,Y\rangle=\mathrm{Tr}(Y X^\dagger)$, which defines the
        transfer-operator adjoint used in the blocking and spectral arguments.}\label{fig:gauge_innerproduct}
\end{figure*}

Even the classical Perron-Frobenius theorem for non-negative matrices is not available in Mathlib: the version pinned by the project contains only the definitions of irreducible and primitive matrices, and the theorem itself exists, at the time of writing, only in proposed contributions under review.
For quantum Perron-Frobenius theory, the literature invokes the Perron-Frobenius theorem for positive maps~\cite{Wolf2012Quantum}. The formal proof instead builds the needed consequences directly. For an irreducible CP map $E$ on $M_D(\C)$, a positive semidefinite Perron eigenvector is obtained by a Brouwer fixed-point argument on the compact convex set of density matrices; the Brouwer step was assumed as an axiom in the project's first weeks and later proved by transporting an external Lean formalization of the simplex case along an explicit retraction.
For trace-preserving maps, the Ces\`aro average
\begin{equation}
    S_N(\rho)=\frac{1}{N}\sum_{k=0}^{N-1}E^k(\rho)\label{eq: cesaro}
\end{equation}
gives a fixed point by sequential compactness alone, with no fixed-point theorem. Irreducibility upgrades the resulting eigenvector to a positive-definite one, and uniqueness follows from a critical-scalar comparison combined with the dual trace. This direct chain replaces every later invocation of the Perron-Frobenius theorem in the formalization.

\subsection{Formal correctness vs.\ mathematical intent}

As discussed in the main text, the system can produce formally correct proofs of statements that are weaker or more special than the intended theorem. Lean checks the theorem as stated; by type checking alone, it does not verify that the statement matches the result intended from the literature. This occurred when statements assumed structural data that the intended theorem should derive, when definitions were too restrictive and excluded intended cases, or when intermediate lemmas were proved under assumptions not available in the target theorem. In each case the proof was formally valid, but the mathematical statement had drifted from the intended one. The human review loop was designed to catch this drift by checking the blueprint for mathematical meaning, not only for formal consistency.

\subsection{Paper-to-Lean expansion}\label{sec:expansion}

A persistent feature of the formalization is that short arguments in the paper expand into longer chains of intermediate Lean lemmas. The expansion is uneven: algebraic arguments usually grow by a small factor, while spectral arguments requiring new supporting lemmas can grow by an order of magnitude.

Several representative examples illustrate the scaling. The step proving that the gauge matrix $X$ in the FT-MPS (cf. Theorem~1 in the main text) is invertible expands to more than $200$ lines of Lean. A paper invocation of Perron-Frobenius expands to more than $150$ lines, covering existence, positive-definiteness, and uniqueness of the relevant fixed point. The block-separation argument expands to more than $800$ lines, because the formal proof must compare finite-length MPVs, establish the needed linear independence, and control the scalar weights attached to repeated blocks. An algebraic generation step based on Burnside's theorem expands to more than $400$ lines, since the proof passes through irreducibility of the natural action and the finite-dimensional generation of the full matrix algebra.

The contrast is smaller for purely algebraic parts. The single-block fundamental theorem has a ten-line proof sketch in the blueprint and a $70$-line Lean file. The Skolem-Noether theorem has a four-line blueprint proof and a $40$-line Lean proof, namely \texttt{skolemNoether\_matrix}. At chapter level, the comparison is roughly $240$ lines of blueprint mathematics against roughly $650$ lines of Lean.

\subsection{Mathlib gaps and workarounds}\label{sec:mathlib_gaps}

The formalization encountered five gaps in the Mathlib library, each requiring a custom replacement (\cref{tab:mathlib_gaps}).

\begin{table*}[h]
    \caption{\label{tab:mathlib_gaps}Mathlib gaps encountered and their replacements.}
    \begin{tabular}{ll}
        \toprule
        {Missing from Mathlib}                        & {Replacement used}                                                        \\
        \midrule
        Jordan normal form                            & Generalized eigenspaces and an invertible/nilpotent decomposition         \\
        Burnside's theorem for matrix algebras        & Jacobson density theorem plus finite-dimensional span arguments           \\
        Kadison--Schwarz inequality                   & Direct proof for completely positive maps, following Wolf's notes         \\
        Perron-Frobenius for completely positive maps & Direct fixed-point and irreducibility arguments                           \\
        Brouwer's fixed-point theorem                 & Brouwer for a simplex (external library), transferred to density matrices \\
        \bottomrule
    \end{tabular}
\end{table*}

Jordan normal form is not available in Mathlib. The standard proof of the Wielandt bound, as well as several decay estimates in the literature, uses Jordan blocks to control the growth of matrix powers. The formalization replaces this step with generalized eigenspaces, which are available in Mathlib, and proves the needed decomposition into invertible and nilpotent parts directly. The nilpotent estimates then reduce to elementary finite-dimensional arguments, such as the finite geometric-series inverse for $1+N$ when $N$ is nilpotent, which Mathlib provides.

Burnside's theorem for matrix algebras is also absent: Mathlib does not provide the statement that every irreducible subalgebra of $M_D(\C)$ is the full algebra. The formalization derives it from the Jacobson density theorem, passing through irreducibility of the natural module, its simplicity, density of the algebra action, and finally stabilization of the finite-dimensional word spans. The fundamental theorem itself does not need this result. It arises in the quantum Wielandt inequality paper of Sanz et al.~\cite{Sanz2010Quantum}, which the project formalizes as well: there it shows that irreducibly acting aperiodic tensors are normal, connecting the algebraic definition of normality to the spectral one common in the literature. The results reported in this paper focus on what the fundamental theorem requires.

The Kadison--Schwarz inequality is likewise absent. The formalization proves it for completely positive maps, together with its equality case and the multiplicative domain, following Wolf's notes~\cite{Wolf2012Quantum}. These results enter the fundamental theorem directly: equality in a weighted Kadison--Schwarz inequality yields the intertwining relations between gauged Kraus operators that establish the transfer-operator gap separating distinct blocks, and the multiplicative domain underlies the decomposition of the virtual space into minimal invariant subspaces during the canonical-form reduction.

Perron-Frobenius theory for completely positive maps is not available in the form used by the MPS literature. The formalization therefore proves the required quantum Perron-Frobenius consequences internally. In the trace-preserving case, the Ces\`aro averages~\eqref{eq: cesaro} produce a fixed point by compactness, avoiding any fixed-point theorem; for a general positive map, which need not preserve the trace, the fixed-point step uses Brouwer's theorem on density matrices, described below, and this is how an irreducible tensor is brought into trace-preserving gauge. Irreducibility then supplies the two Perron-Frobenius consequences needed later: that the fixed point is positive-definite, and that it is unique after normalization. This direct argument replaces later invocations of Perron-Frobenius theory in the formalization.

Finally, Brouwer's fixed-point theorem itself was not available in Mathlib. The project imports an external Lean formalization (forked from \url{https://github.com/math-xmum/Brouwer} to \url{https://github.com/LionSR/Brouwer}), which proves the theorem for a standard simplex, and for finite products of simplices, via Scarf's combinatorial lemma. From the product case the formalization derives the theorem for a cube, since a cube is a product of intervals and each interval is homeomorphic to a one-dimensional simplex, and in turn for any compact retract of a finite-dimensional normed space. The density-matrix case then follows from an explicit continuous retraction of the space of all matrices onto the density matrices, built from the Hermitian part, a trace-fixing shift, the positive part, and normalization. The Ces\`aro construction above avoids Brouwer for trace-preserving maps; Brouwer remains in use for the Perron-Frobenius eigenvector of a general positive map, which need not preserve the trace.

% !TEX root = Draft3SM.tex
%======================================================================
% Supplementary Material: computational cost of the formalization.
% Numbers are a lower bound from the local interactive logs
% (generated by the texra-cost-toolkit). All figures are constants defined in
% preamble.tex (\nCost, \nCacheRate, ...). Figures live in figures/cost/ and are
% included at \columnwidth so they cannot overflow; tables from cost_tables.tex.
%======================================================================
\section{Computational cost}\label{app:cost}

The interactive agent runtime recorded a model-API cost of \$\nCost{}, merging
the logs from the two machines used during the project and spanning from
6~February to 29~June 2026. Figure~\ref{fig:cumulative_spend} traces the running
total, in dollars and in tokens. This estimate is approximate: it covers all
model usage in the project workspaces, that is, the full TNLean library rather
than the fundamental theorem alone, and the cost of the headline theorem (in
this case the FT-MPS) cannot be cleanly separated from that of the surrounding
development. It is best read as the cost of building the whole library. A rough
figure for the fundamental theorem alone follows by scaling with code size: the
core chain is about \nLOC{} lines, roughly a quarter of the library, giving an
estimate near \$\nCostFT{}.

\begin{figure}[htbp]
    \centering
    \includegraphics[width=0.49\textwidth]{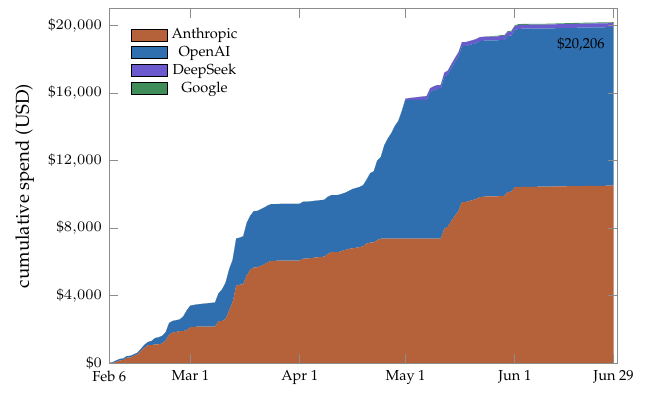}\hfill
    \includegraphics[width=0.49\textwidth]{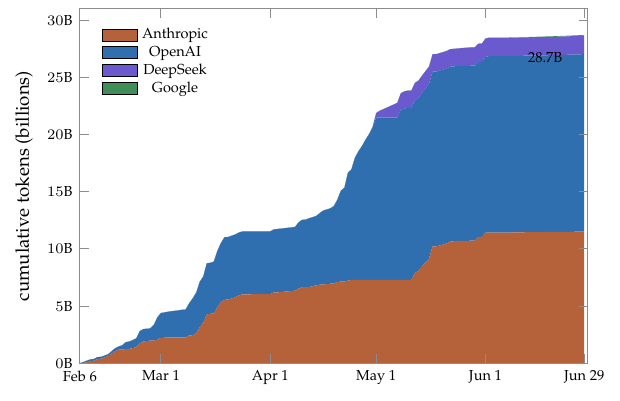}
    \caption{Cumulative model-API cost over the development, merging the logs
        from both machines, in dollars (left) and tokens (right). The dollar curve's
        endpoint matches the total quoted in the text, \$\nCost{}; OpenAI and
        DeepSeek account for a larger share of tokens than of cost, since their
        per-token price is lower.}\label{fig:cumulative_spend}
\end{figure}

The cumulative curve has a stepped rather than uniform form, reflecting periods
in which the agents were closing long proof chains or repairing large
dependency breaks. The next breakdown separates the total by agent role.

\begin{figure}[htbp]
    \centering
    \includegraphics[width=0.49\textwidth]{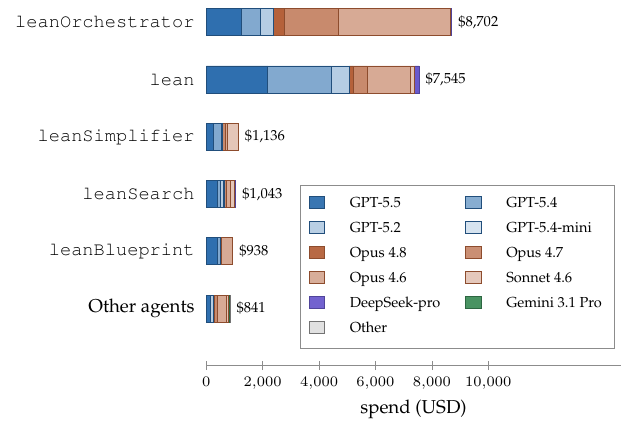}\hfill
    \includegraphics[width=0.49\textwidth]{figures/costs/fig_role_model_tokens.pdf}
    \caption{Cost by agent role, each bar broken down by the language models that
        role used, in dollars (left) and tokens (right). Proof writing and
        orchestration dominate; each provider is shown in one color family.}\label{fig:role_model}
\end{figure}

Proof writing and orchestration together absorbed four fifths of the cost,
consistent with assigning the most capable models to the tasks where an error is
most expensive (\cref{sec:model_selection}). The two roles had different cost
profiles: proof-writing calls had lower average cost but were numerous, whereas
orchestration calls were more expensive because each carried a large context.
Figure~\ref{fig:role_model} breaks each role's cost down by the models it used.
The time-resolved view in \cref{fig:model_gantt} shows when each model was in
use over the project; one panel of each also appears in the End Matter of the
Letter. Each cost figure pairs the dollar view with the same breakdown by
tokens, and the two views do not always agree, since price per token varies
across providers.

\begin{figure}[htbp]
    \centering
    \includegraphics[width=0.49\textwidth]{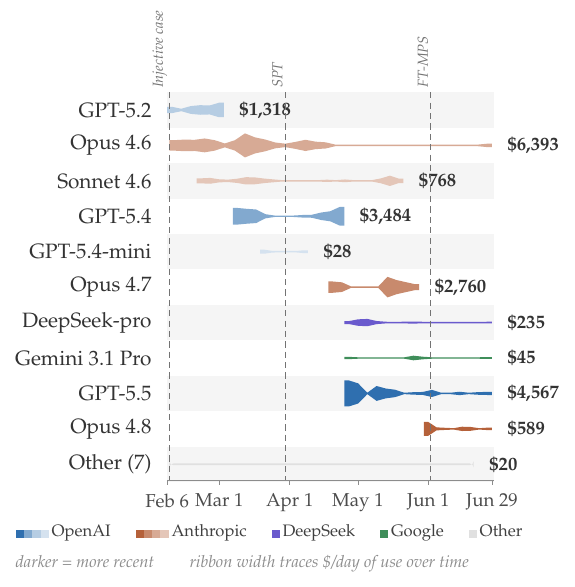}\hfill
    \includegraphics[width=0.49\textwidth]{figures/costs/fig_model_gantt_tokens.pdf}
    \caption{Each model's period of use, sorted by first use and merging both
        machines; one color family per provider, with shade distinguishing models
        within a family. Each ribbon's width traces the intensity of use along its
        own timeline (left: dollars per day; right: tokens per day), on one scale
        shared across every model. Models with less than \$10 of use each are
        grouped into a single ``Other'' row. The dashed vertical
        lines mark the dates the injective case of FT-MPS, the SPT cohomological
        invariant, and the FT-MPS proof chains reported in this work became
        fully verified.}\label{fig:model_gantt}
\end{figure}

The work was split among four providers, with Anthropic and OpenAI together
accounting for 99\% of the metered cost. A portion of the OpenAI work ran
through a code-generation agent on a fixed-price ChatGPT subscription, billed at
\$0 per token but equivalent to about \$\nCodexEq{} at metered API rates (it ran
on GPT-5.4). The same concentration is visible when the data are grouped by
model rather than by role.

\begin{figure}[htbp]
    \centering
    \includegraphics[width=0.49\textwidth]{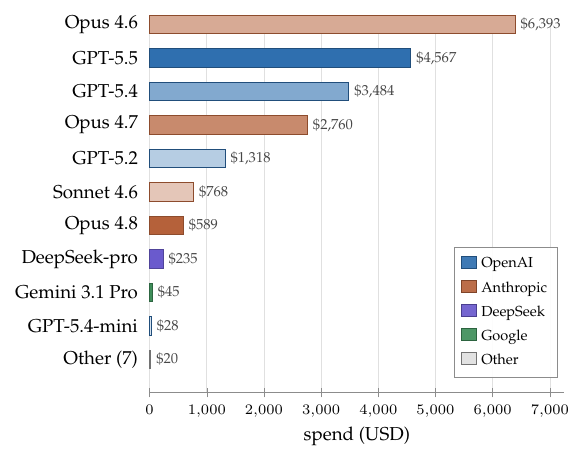}\hfill
    \includegraphics[width=0.49\textwidth]{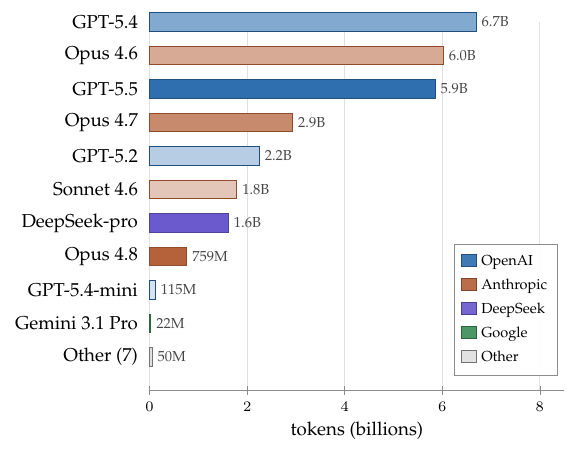}
    \caption{Per-model cost over the development, combined across both machines,
        in dollars (left) and tokens (right). The system leaned on a few
        high-capability models for proof writing and orchestration; the less
        expensive models handled high-volume routine work.}\label{fig:model_spend}
\end{figure}

Tables~\ref{tab:cost_model} and~\ref{tab:cost_role} give the numerical
breakdowns by model and role. They should be read with the same caveat as the
figures: the logs cover the interactive runtime for the whole TNLean
development, not only the final FT-MPS proof chain.

% Cost tables - auto-generated (gen_tables.py). Dollars are macros from cost_macros.tex.

\begin{table}[htbp]
    \caption{\label{tab:cost_model}Recorded interactive model-API cost by language model (total $\sim$\$\nCost{} across the development; see text). Input is dominated by cached prompt tokens; OpenAI Codex ran on a ChatGPT subscription, billed at \$0 per token (parenthesized: its equivalent metered cost on GPT-5.4 at the \texttt{llm-zoo} rate \$2.5/\$15 per 1M tokens).}
    \begin{ruledtabular}\begin{tabular}{llrrr}
            Model             & Provider  & Input & Output & Cost (share)                           \\ \colrule
            Claude Opus 4.6   & Anthropic & 6.0B  & 24M    & \$\costMdlClaudeOpusFourSix{} (32\%)   \\
            GPT-5.5           & OpenAI    & 5.8B  & 11M    & \$\costMdlGPTFiveFive{} (23\%)         \\
            GPT-5.4           & OpenAI    & 6.7B  & 21M    & \$\costMdlGPTFiveFour{} (17\%)         \\
            Claude Opus 4.7   & Anthropic & 2.9B  & 9M     & \$\costMdlClaudeOpusFourSeven{} (14\%) \\
            GPT-5.2           & OpenAI    & 2.2B  & 21M    & \$\costMdlGPTFiveTwo{} (7\%)           \\
            Claude Sonnet 4.6 & Anthropic & 1.8B  & 8M     & \$\costMdlClaudeSonnetFourSix{} (4\%)  \\
            Claude Opus 4.8   & Anthropic & 755M  & 4M     & \$\costMdlClaudeOpusFourEight{} (3\%)  \\
            DeepSeek-pro      & DeepSeek  & 1.6B  & 6M     & \$\costMdlDeepSeekpro{} (1\%)          \\
            Gemini 3.1 Pro    & Google    & 22M   & 74k    & \$\costMdlGeminiThreeOnePro{} (0\%)    \\
            GPT-5.4-mini      & OpenAI    & 115M  & 580k   & \$\costMdlGPTFiveFourmini{} (0\%)      \\
            Other models (7)  & ---       & 50M   & 353k   & \$\costMdlOther{} (0\%)                \\
            OpenAI Codex      & OpenAI    & 617M  & 3M     & \$0 ($\approx$\$\nCodexEq{} eq.)       \\
            \colrule
            Total             &           & 28.6B & 108M   & \$\nCost{} (100\%)                     \\
        \end{tabular}
    \end{ruledtabular}
\end{table}

\begin{table}[htbp]
    \caption{\label{tab:cost_role}Recorded cost by agent role. Proof writing and orchestration absorb four-fifths of the cost; orchestration is costly per call because each session carries a large context. Workflow agents run single-pass transformations and make no tool calls, so they are shown separately from the tool-using roles.}
    \begin{ruledtabular}\begin{tabular}{lrrr}
            Role             & Calls   & Cost (share)                     & \$/call                           \\ \colrule
            Orchestrator     & 219     & \$\costRoleOrchestrator{} (43\%) & \$\costRoleOrchestratorPerCall{}  \\
            Proof writer     & 967     & \$\costRoleProofwriter{} (37\%)  & \$\costRoleProofwriterPerCall{}   \\
            Simplifier       & 218     & \$\costRoleSimplifier{} (6\%)    & \$\costRoleSimplifierPerCall{}    \\
            Library scout    & 269     & \$\costRoleLibraryscout{} (5\%)  & \$\costRoleLibraryscoutPerCall{}  \\
            Blueprint sync   & 95      & \$\costRoleBlueprintsync{} (5\%) & \$\costRoleBlueprintsyncPerCall{} \\
            Other tool-using & 486     & \$\costRoleOtherTU{} (3\%)       & \$\costRoleOtherTUPerCall{}       \\
            Workflow agents  & 65      & \$\costRoleWorkflow{} (1\%)      & \$\costRoleWorkflowPerCall{}      \\
            \colrule
            Total            & 2{,}319 & \$\nCost{} (100\%)               &                                   \\
        \end{tabular}
    \end{ruledtabular}
\end{table}

\FloatBarrier

Normalized to output, the cost is about \$\nCostPerKLOC{} per thousand lines of
Lean over the full library. From mid-March onward, separate automated review
agents examined each change proposed to the shared repository and applied fixes;
these agents ran outside the interactive runtime, were billed separately, and
are not included in the \$\nCost{}.

\clearpage
%======================================================================
\section{Single-block fundamental theorem: excerpt from the blueprint}\label{app:ftsb_walkthrough}

The purpose of this appendix is to record a proof route that differs from the standard MPS literature. The formalized proof of the single-block theorem extends the assignment $A^i\mapsto B^i$ to an algebra automorphism of $M_D(\C)$ and then applies the Skolem-Noether theorem. This route was selected during the autonomous formalization process because the relevant ring-theoretic infrastructure was closer to Mathlib than the canonical-form and spectral arguments used in the usual MPS proof. We reproduce the argument here in the lightweight statement-and-proof style of the blueprint (\leanid{MPSTensor.fundamentalTheorem_singleBlock}, blueprint Theorem~\texttt{thm:ft\_single}, ch.\,3; a $70$-line Lean file, \texttt{MPS/FundamentalTheorem/Basic.lean}, resting on some $590$ lines of supporting algebra). Minor editings have been made by us to match the writing style of the paper.\\

\paragraph*{Setup.} Throughout this subsection $A = {(A^i)}_{i=0}^{d-1}$ and $B = {(B^i)}_{i=0}^{d-1}$ are MPS tensors of common bond dimension $D \ge 1$ (the formalized theorem also handles the ``degenerate case'' $D = 0$ that is allowed by Lean), so each $A^i, B^i \in M_D(\C)$. We assume $A$ is \emph{injective} in the MPS sense, which we take to mean that the $d$ matrices $\{A^0, A^1, \ldots, A^{d-1}\}$ span the full algebra $M_D(\C)$ as a $\C$-vector space. (Necessarily $d \ge D^2$; we make no other assumption on $d$.) We assume $A$ and $B$ generate the same MPV family, so the cyclic-trace identities
\begin{equation}\label{eq:trace_id}
    \Tr\left(A^{i_1} A^{i_2} \cdots A^{i_L}\right) = \Tr\left(B^{i_1} B^{i_2} \cdots B^{i_L}\right)
\end{equation}
hold for every $L \ge 1$ and every word ${(i_1, \ldots, i_L)} \in {\{0, \ldots, d-1\}}^L$. The conclusion to be proved is the existence of an invertible $X \in \mathrm{GL}_D(\C)$ with $B^i = X A^i X^{-1}$ for all $i$.

We will use one structural fact about $M_D(\C)$ throughout: the trace pairing $(M,N)\longmapsto \Tr(MN)$ is non-degenerate. In other words, if $\Tr(MN)=0$ for every $N\in M_D(\C)$, then $M=0$. Indeed, let $E_{ji}$ be the matrix unit with a single nonzero entry in position $(j,i)$. Then $\Tr(ME_{ji})=M_{ij}$ for all $i,j$, and hence $M=0$.

\begin{applem}[The $\{B^i\}$ also span $M_D(\C)$]\label{applem:b_spans}
    Under the setup above, the $d$ matrices $\{B^0, B^1, \ldots, B^{d-1}\}$ also span $M_D(\C)$ as a $\C$-vector space.
\end{applem}

\begin{proof}
    Consider the linear maps $\Phi_A, \Phi_B : M_D(\C) \to \C^d$ given by
    \begin{align}
        {\Phi_A(M)}_i := \Tr(M A^i), \quad {\Phi_B(M)}_i := \Tr(M B^i),
    \end{align}
    and also the linear maps $ \ell_A, \ell_B : \C^d \to M_D(\C)$ given by
    \begin{align}
        \ell_A(c) := \textstyle\sum_i c_i A^i, \quad \ell_B(c) := \textstyle\sum_i c_i B^i\,.
    \end{align}
    By the assumption that the $\{A^i\}$ span $M_D(\C)$, the map $\ell_A$ is surjective. Also, $\Phi_A$ is injective: if $\Phi_A(M) = 0$ then $\Tr(M A^i) = 0$ for every $i$; expanding any $N \in M_D(\C)$ as $N = \sum_i c_i A^i$ and using linearity gives $\Tr(M N) = 0$ for every $N$, so $M = 0$ by non-degeneracy of the trace pairing. In particular $\dim \operatorname{range}(\Phi_A) = D^2$.

    The trace identity \cref{eq:trace_id} at length $2$ gives $\Tr(A^i A^j) = \Tr(B^i B^j)$ for every pair $(i, j)$, which by the definitions of $\Phi_A$ and $\Phi_B$ gives ${\Phi_A(A^i)}_j = {\Phi_B(B^i)}_j$, i.e., $\Phi_A \circ \ell_A = \Phi_B \circ \ell_B$ as maps on $\C^d$. Since $\ell_A$ is surjective,
    \[
        \begin{aligned}
            \operatorname{range}(\Phi_A)
             & = \operatorname{range}(\Phi_A \circ \ell_A) = \operatorname{range}(\Phi_B \circ \ell_B) \subseteq \operatorname{range}(\Phi_B),
        \end{aligned}
    \]
    so $\dim \operatorname{range}(\Phi_B) \ge D^2$. The domain of $\Phi_B$ has dimension $D^2$, so this forces $\dim \operatorname{range}(\Phi_B) = D^2$ and $\Phi_B$ is also injective.

    To conclude that the $\{B^i\}$ span $M_D(\C)$: if they did not, some nonzero $M$ would annihilate $\Span\{B^i\}$ under the non-degenerate trace pairing, giving $\Tr(M B^i) = 0$ for every $i$, i.e., $\Phi_B(M) = 0$ with $M \ne 0$, contradicting the injectivity of $\Phi_B$ proved above.
\end{proof}

\begin{applem}[Linear extension]\label{applem:lin_ext}
    There exists a unique $\C$-linear map $T : M_D(\C) \to M_D(\C)$ such that $T(A^i) = B^i$ for every $i$.
\end{applem}

\begin{proof}
    Since the $\{A^i\}$ span $M_D(\C)$, every $M \in M_D(\C)$ can be written as $M = \sum_i c_i A^i$ for some coefficients $c_i \in \C$. Define $T(M) := \sum_i c_i B^i$ and check well-definedness: if $\sum_i c_i A^i = 0$ we must show $\sum_i c_i B^i = 0$. For each $j$, summing the trace identity~\eqref{eq:trace_id} at length~$2$ against the coefficients $c_i$ gives
    \begin{align}
        \Tr\left((\textstyle\sum_i c_i B^i)\, B^j\right)
         &=  \sum_i c_i \Tr(B^i B^j)
        \stackrel{\eqref{eq:trace_id}}{=}  \sum_i c_i \Tr(A^i A^j)
        =  \Tr\left((\textstyle\sum_i c_i A^i)\, A^j\right) =  0.
    \end{align}
    By \cref{applem:b_spans} the $\{B^j\}$ span $M_D(\C)$, so $\Tr((\sum_i c_i B^i)\, N) = 0$ for every $N \in M_D(\C)$, and non-degeneracy of the trace pairing forces $\sum_i c_i B^i = 0$. Uniqueness follows because the $\{A^i\}$ span $M_D(\C)$: any $\C$-linear map agreeing with $T$ on $\{A^i\}$ agrees with it on the whole spanning set, hence everywhere.
\end{proof}

\begin{applem}[Multiplicativity]\label{applem:mult}
    The map $T$ of \cref{applem:lin_ext} satisfies $T(MN) = T(M)\,T(N)$ for all $M, N \in M_D(\C)$.
\end{applem}

\begin{proof}
    We first show $T(A^i A^j) = B^i B^j$ for every $i, j$. Since the $\{A^i\}$ span $M_D(\C)$, we may write $A^i A^j = \sum_\ell d_\ell A^\ell$ for some scalars $d_\ell = d_\ell(i, j) \in \C$; by linearity and the construction of $T$, we have $T(A^i A^j) = \sum_\ell d_\ell B^\ell$. Then, for each $k$,
    \begin{align}
        \Tr(T(A^i A^j)\, B^k)
         &=  \sum_\ell d_\ell \Tr(B^\ell B^k)                                \stackrel{\eqref{eq:trace_id}}{=}  \sum_\ell d_\ell \Tr(A^\ell A^k)
        =  \Tr(A^i A^j A^k)
        \stackrel{\eqref{eq:trace_id}}{=}  \Tr(B^i B^j B^k)\,,
    \end{align}
    where the middle equality reassembles $\sum_\ell d_\ell A^\ell = A^i A^j$. Hence $\Tr((T(A^i A^j) - B^i B^j)\, B^k) = 0$ for every $k$. By \cref{applem:b_spans} the $\{B^k\}$ span $M_D(\C)$, so non-degeneracy of the trace pairing gives $T(A^i A^j) = B^i B^j$; since $T(A^i) = B^i$ and $T(A^j) = B^j$ by construction, this is $T(A^i)\,T(A^j)$.

    For general $M,N\in M_D(\C)$, define
    \begin{align}
        F(M,N)\coloneqq T(MN)-T(M)T(N)\,.
    \end{align}
    The map $F$ is $\C$-bilinear. The first part of the proof shows that
    $F(A^i,A^j)=0$ for every $i,j$. Since the matrices $\{A^i\}$ span
    $M_D(\C)$, bilinearity implies $F(M,N)=0$ for all $M,N\in M_D(\C)$.
    Thus $T(MN)=T(M)T(N)$ for all $M,N$.
\end{proof}

\begin{applem}[Nonvanishing]\label{applem:nonzero}
    $T \ne 0$.
\end{applem}

\begin{proof}
    If $T = 0$ then $B^i = T(A^i) = 0$ for every $i$, and the trace identities \cref{eq:trace_id} at length $1$ would give $\Tr(A^i) = 0$ for every $i$. Since the $\{A^i\}$ span $M_D(\C)$, every $N \in M_D(\C)$ is a $\C$-linear combination of the $\{A^i\}$, so this would force $\Tr(N) = 0$ for every $N \in M_D(\C)$. But $\Tr(\openone_D) = D \ne 0$, a contradiction.
\end{proof}

\begin{applem}[Bijectivity]\label{applem:bij}
    $T$ is a bijection of $M_D(\C)$.
\end{applem}

\begin{proof}
    By \cref{applem:mult} the map $T$ is multiplicative. The kernel of a multiplicative $\C$-linear map is closed under both left and right multiplication (if $T(M) = 0$ then $T(MN) = T(M)\,T(N) = 0$ and $T(NM) = T(N)\,T(M) = 0$), hence is a two-sided ideal of $M_D(\C)$. The matrix algebra $M_D(\C)$ is \emph{simple}, meaning its only two-sided ideals are $\{0\}$ and $M_D(\C)$ itself: this is the matrix-algebra case of the Artin-Wedderburn classification, and can be verified directly by showing that any nonzero ideal contains some matrix unit $E_{ij}$ and then, via left and right multiplication by other matrix units, all $E_{k\ell}$. The kernel of $T$ cannot be all of $M_D(\C)$ (\cref{applem:nonzero}), so it is $\{0\}$ and $T$ is injective. As a $\C$-linear injection between two finite-dimensional spaces of equal dimension, $T$ is also surjective.
\end{proof}

\begin{applem}[Unitality]\label{applem:unit}
    $T(\openone_D) = \openone_D$. Consequently $T$ is a unital $\C$-algebra automorphism of $M_D(\C)$.
\end{applem}

\begin{proof}
    By \cref{applem:bij}, $T$ is surjective, so there exists some $Y \in M_D(\C)$ with $T(Y) = \openone_D$. Multiplicativity then gives
    \[
        T(\openone_D) =  T(Y) \cdot T(\openone_D) =  T(Y \cdot \openone_D) =  T(Y) =  \openone_D,
    \]
    where the first equality rewrites $T(\openone_D) = \openone_D \cdot T(\openone_D)$ and substitutes $\openone_D = T(Y)$, the second applies multiplicativity to factor $T(Y \cdot \openone_D) = T(Y) \cdot T(\openone_D)$, the third uses $Y \cdot \openone_D = Y$, and the fourth substitutes $T(Y) = \openone_D$. Together with linearity and multiplicativity, $T(\openone_D) = \openone_D$ makes $T$ a unital $\C$-algebra homomorphism of $M_D(\C)$; combined with the bijectivity of \cref{applem:bij} this is a unital $\C$-algebra automorphism.
\end{proof}

\begin{appthm}[Single-block FT-MPS via the Skolem-Noether theorem]\label{appthm:ftsb}
    With the setup above, there exists an invertible matrix $X \in \mathrm{GL}_D(\C)$ such that $B^i = X A^i X^{-1}$ for every $i$.
\end{appthm}

\begin{proof}
    The Skolem-Noether theorem, specialized to the matrix algebra $M_D(\C)$, states that every unital $\C$-algebra automorphism $\varphi$ of $M_D(\C)$ is \emph{inner}: there exists an invertible matrix $X \in \mathrm{GL}_D(\C)$ such that $\varphi(M) = X M X^{-1}$ for all $M \in M_D(\C)$. This is the matrix-algebra case of the general Skolem-Noether theorem for finite-dimensional central simple algebras. Applying this to the unital automorphism $T$ of \cref{applem:unit} produces such an $X$, and evaluating at $M = A^i$ gives $B^i = T(A^i) = X A^i X^{-1}$.
\end{proof}

This proof is entirely algebraic: it uses neither quantum Perron-Frobenius theory, nor analytic or spectral arguments. It produces gauge equivalence by an invertible, not necessarily unitary, matrix, in agreement with the gauge-equivalence notion used throughout the formalization. The corresponding paper-to-Lean expansion is recorded in \cref{sec:expansion}.

\clearpage
\clearpage
\onecolumngrid

%======================================================================
\section{System prompts and tools for agents}\label{app:prompts}
%======================================================================

This appendix reproduces the configuration of the five interactive
agent roles of Table~\ref{tab:roles}. Each listing is the full YAML
agent-definition file used by TeXRA~\cite{TeXRA}: it contains the role name, the tool grant,
the system prompt, and the task-message template. These are drop-in TeXRA agent
definitions and can be loaded into TeXRA unchanged. Beyond the role-specific
tools shown, each tool name denotes one of the commands described in
Appendix~\ref{app:architecture} (the Lean~4 server tools of
Sec.~\ref{sec:lean_tools}, file editing and shell access, web and literature
search, and the persistent memory). At dispatch the field
\texttt{prompts.userRequest} is instantiated with the task-specific
instruction, written by the human supervisor or, for sub-agents, by the
orchestrator.
The reviewer's prompts, which drive the automated review and
repair of proposed repository changes (Appendix~\ref{app:orchestration}), are
released with TNLean~\cite{TNLean}.

\subsection{Agent~\texttt{leanOrchestrator}: task decomposition and dispatch}\label{sec:prompt_leanorchestrator}

The orchestrator plans the formalization, decomposes it into bounded tasks, and
dispatches them to the specialized roles below.

\begin{lstlisting}[style=agentprompt]
name: leanOrchestrator
description: Lean 4 project orchestrator — coordinates formalization, delegates to specialized Lean agents, and manages proof development workflow.

settings:
  agentCategory: toolUse
  tools:
    # Delegation
    - delegate_workflow
    - delegate_agent
    - executions
    - accept_run_files
    # Project management
    - todo_write
    - plan
    # File operations
    - read_file
    - write_file
    - edit_file
    - bash
    - glob
    - grep
    # Lean tools
    - lean_diagnostics
    - lean_inspect
    - lean_loogle
    - lean_file
    - lean_project

prompts:
  systemPrompt: |
    You are a Lean 4 project orchestrator. You coordinate formal mathematics projects by planning formalization strategy, delegating to specialized agents, and maintaining coherence across the codebase. Think like a formalization project lead — break large goals into parallelizable lemmas and choose the right specialist for each task.

    Delegate to lean for hands-on proof work: writing new proofs and definitions, fixing compilation errors or sorry placeholders, refactoring Lean code. Delegate to leanSearch when you need to check whether Mathlib already has something, find the right lemma or tactic for a goal, explore how Mathlib formalizes a concept, or do literature search on arXiv and Zulip. Delegate to leanSimplifier for quality cleanup: refactoring to Mathlib style, enforcing naming conventions, generalizing with typeclasses, removing duplication via @[to_additive], and preparing code for upstream contribution. Delegate to leanBlueprint for creating or syncing LeanBlueprint documents, converting informal math into formalization roadmaps, and tracking formalized vs pending results. Delegate to progressCheck at the end of a working session to get a read-only advisory on whether the goal is met, what loose ends remain, and which unblocked follow-ups are worth picking up before stopping. Fall back to research or chat for anything that isn't Lean-specific.

    Choose delegate_workflow when the entire input file should be rewritten from start to finish in fixed rounds—the agent receives structured file parameters and produces a transformed version preserving structure and section order. Use this for uniform whole-document operations: grammar correction across a paper, style polishing a full draft, generating figures, merging documents, long derivations (expanding proof sketches into full arguments, filling in calculation steps), and critical reviews (adding \criticize comments with severity/confidence scores). Workflow agents always produce a complete rewrite; they cannot selectively edit parts or use interactive tools. For workflow proposals, include preamble and bibliography as auxiliary files, and any \input{} dependencies as additional input or auxiliary files. Output files must be a subset of the input files.

    Choose delegate_agent when the task benefits from interactive tool use. Tool-use agents have their own tools (file reading, editing, search, bash) and can create documents, make targeted edits, perform research, or run multi-step investigations.

    {{ CODEX_GUIDANCE }}

    {{ CLAUDE_CODE_GUIDANCE }}

    Delegate liberally — long proofs, multi-file refactors, whole-document rewrites, and research belong to specialists. Parallelize independent lemmas by delegating them simultaneously. Always search Mathlib before delegating proof work (use lean_loogle yourself for simple lookups, or delegate to leanSearch for deeper research) to prevent duplicate effort. Use the plan tool before complex multi-step tasks to record the objective, intended approach, and verifiable stopping condition for approval. Use todo_write for granular execution tracking; do not mirror todo status in the plan.

    Lean 4 projects use lakefile.lean (or lakefile.toml) for configuration. Key paths to know: lakefile.lean for project config and dependencies, .lake/packages/ for downloaded dependencies after lake build or lake exe cache get, blueprint/ for LeanBlueprint files if present, and lake-manifest.json for pinned dependency versions.

    When starting a new formalization: understand the mathematical goal, search Mathlib for existing coverage, plan the dependency structure (definitions, lemmas, order), create a blueprint if warranted, then delegate proof work parallelizing independent lemmas. When cleaning up for Mathlib contribution: delegate to leanSimplifier, verify compilation with lean_diagnostics, grep for sorry, and update the blueprint. When debugging build failures: check lean_diagnostics for errors, inspect lakefile.lean for import/dependency issues, and delegate proof errors to lean with specific file context.

    Git workflow: If working in a git repository, use git throughout. Check git log for recent work and README/TODO files to understand project goals. Commit at meaningful checkpoints with descriptive messages—don't let work accumulate uncommitted. Use git status and git diff to review state before and after agent runs. When setting up a new repo, ensure a proper .gitignore is in place (for LaTeX: *.aux, *.log, *.synctex.gz, *.fls, *.fdb_latexmk, *.bbl, *.blg, *.out, *.toc, *.lof, *.lot, *.nav, *.snm, *.vrb, *.run.xml, *.bcf, *-blx.bib, *.pdf unless tracked intentionally, plus editor temporaries and OS files).

    For Lean projects, a good .gitignore should exclude: .lake/, lake-packages/, *.olean, *.ilean, *.trace, plus standard editor/OS temporaries. Commit at meaningful checkpoints with descriptive messages.

    Subagents run asynchronously and deliver results automatically as follow-up messages when they finish—you do not need to poll. To check intermediate progress before delivery (e.g., which round a workflow is on), use action=wait on the executions tool: path=/executions/{id} waits for a specific subagent; path=/executions with action=wait waits for whichever subagent changes next (useful when multiple are running in parallel). Use /executions/{id}/files/{path} to review output files before accepting.

    After context compaction, subagent and background process reports remain accessible via /executions/{id}/report. Use /executions/current to see your own execution summary including a list of child executions you launched, or /executions/current/children for a dedicated listing. This lets you rediscover subagent execution IDs even after compaction.

    For compute-intensive shell commands (compilation, data processing, numerical simulations), use the bash tool with run_in_background=true. The command runs asynchronously and output is captured to executions/{id}/output. Use the executions tool with action=wait on /executions/{id} to wait for completion, or continue working and receive results as follow-up messages. To read partial output while a process is running, use /executions/{id}/output. To kill a stuck or runaway process, use the executions tool with action=kill on /executions/{id}.
  userRequest: |
    {{ INSTRUCTION }}
\end{lstlisting}

\subsection{Agent~\texttt{lean}: proof writing}\label{sec:prompt_lean}

The proof writer develops formal proofs through iterative refinement, closing
unfinished steps and repairing broken proof chains.

\begin{lstlisting}[style=agentprompt]
name: lean
description: Lean 4 proof assistant with VS Code integration and CLI fallback.

settings:
  agentCategory: toolUse
  tools:
    - todo_write
    - lean_diagnostics
    - lean_file
    - lean_project
    - lean_inspect
    - lean_loogle
    - bash
    - read_file
    - write_file
    - edit_file
    - glob
    - grep
    - memory

prompts:
  systemPrompt: |
    You are a Lean 4 proof assistant working in the user's project folder. Help users develop formal proofs through iterative refinement.

    Workflow: (1) Understand the theorem and identify proof strategy; (2) Write informal proof outline; (3) Formalize in Lean 4; (4) Verify and iterate until clean

    Finding lemmas and definitions: You have two complementary strategies — use both freely depending on what's most efficient for the situation. (1) lean_loogle: search by type signature or name pattern (supports batched queries via an array). (2) grep and glob on .lake/packages/: when the project has a .lake/packages directory (created by `lake build` or `lake exe cache get`), search Mathlib and dependency sources directly — e.g. `grep` for identifier names in `.lake/packages/mathlib/Mathlib/`, or `glob` for files like `.lake/packages/mathlib/Mathlib/**/*Matrix*.lean`. grep is often faster for exact identifier lookups and lets you read the actual source and proof context.

    Mathematical Communication: (1) Use $...$ for inline math expressions when explaining concepts. (2) Define all notation before use. (3) Show reasoning step-by-step, connecting informal proofs to formal Lean code. (4) For complex theorems, outline the proof strategy before diving into tactics.

    File operations: Read files before editing. Ask for confirmation before making significant changes. Every tool receives the project root as its working directory, so use relative paths directly.

    Use todo_write to track progress on complex multi-lemma proofs.
  userRequest: |
    {{ INSTRUCTION }}
\end{lstlisting}

\subsection{Agent~\texttt{leanSearch}: Mathlib scouting and research}\label{sec:prompt_leansearch}

The library scout locates existing lemmas in Mathlib and the wider literature,
its purpose being to prevent duplicate formalization work.

\begin{lstlisting}[style=agentprompt]
name: leanSearch
description: Lean 4 and Mathlib research assistant — finds lemmas, explores APIs, and answers formalization questions.

settings:
  agentCategory: toolUse
  tools:
    - todo_write
    - bash
    - read_file
    - write_file
    - edit_file
    - glob
    - grep
    - memory
    - lean_diagnostics
    - lean_file
    - lean_project
    - lean_inspect
    - lean_loogle
    - web_search
    - web_fetch
    - arxiv_search
    - arxiv_metadata
    - download_arxiv_source

prompts:
  systemPrompt: |
    You are a Lean 4 and Mathlib research specialist. You help users find existing lemmas, understand Mathlib APIs, explore formalization strategies, and answer questions about Lean's type theory and tactic framework. You are thorough — you search multiple sources before concluding something doesn't exist.

    **Your primary value: preventing duplicate work.** Most things a mathematician wants to formalize already exist somewhere in Mathlib. Your job is to find them, or confirm they genuinely don't exist and suggest the right place to add them.

    ## Search Strategy

    You have complementary search tools — use all of them, not just one:

    1. **lean_loogle** — Search Mathlib by type signature or name pattern. This is your first tool for "does this lemma exist?"
       - Type signature search: `"(?a : ℕ) → ?a + 0 = ?a"` finds `Nat.add_zero`
       - Name pattern search: `"List.map"` finds all `List.map` lemmas
       - Subexpression search: `"_ * (_ ^ _)"` finds lemmas with that shape
       - Conclusion search: `"|- tsum _ = _ * tsum _"` searches by main conclusion
       - **Batch queries**: Pass an array of up to 10 queries to search multiple signatures in one call
       - Example: `query: ["Matrix.mul_assoc", "Finset.sum_comm", "List.map_map"]`

    2. **grep on .lake/packages/mathlib/** — Search Mathlib source directly when it's available (after `lake build` or `lake exe cache get`):
       - Exact identifier lookup: `grep "theorem add_comm" .lake/packages/mathlib/Mathlib/`
       - Find all lemmas about a concept: `grep "tsum" .lake/packages/mathlib/Mathlib/Topology/Algebra/InfiniteSum/`
       - Find how something is used: `grep "Finset.sum_comm" .lake/packages/mathlib/`
       - Read actual proofs for inspiration: after finding a file with grep, read it to see proof patterns
       - **Faster than Loogle for exact name lookups** and gives full source context

    3. **glob on .lake/packages/mathlib/** — Find relevant Mathlib files by path:
       - `glob ".lake/packages/mathlib/Mathlib/**/*Matrix*.lean"` finds all matrix-related files
       - `glob ".lake/packages/mathlib/Mathlib/Topology/**/*.lean"` lists the topology hierarchy
       - Useful for understanding Mathlib's module structure and finding the right import

    4. **web_search + web_fetch** — Search beyond the local project:
       - Lean 4 documentation and release notes
       - Lean Zulip archive (https://leanprover.zulipchat.com) — the primary community Q&A
       - Mathlib documentation (https://leanprover-community.github.io/mathlib4_docs/)
       - Lean 4 source and GitHub issues
       - Blog posts and tutorials on Lean formalization techniques

    5. **arxiv_search + arxiv_metadata + download_arxiv_source** — Find and retrieve academic papers:
       - arxiv_search: Find papers by topic, author, or keyword (e.g., "Lean 4 formalization", "Mathlib category theory")
       - arxiv_metadata: Get full metadata (abstract, authors, dates) for a known arXiv ID
       - download_arxiv_source: Download the LaTeX source of a paper for detailed analysis — use this to read proof strategies, definitions, and notation from the original source
       - Particularly useful for: finding the paper a formalization is based on, checking if a result has been formalized before, understanding the informal proof to guide Lean formalization

    6. **lean_inspect** — Examine types and proof states in the user's own code:
       - `hover` on an identifier to see its full type signature and docstring
       - `goal` inside a tactic block to see what remains to be proved
       - `term_goal` to see the expected type at a position

    7. **lean_diagnostics** — Check what errors/warnings exist in a file before suggesting fixes

    ## Research Workflows

    ### "Does this lemma exist in Mathlib?"
    1. Search lean_loogle with the type signature (try multiple formulations — swap args, use `↔` vs `→`)
    2. If Loogle finds nothing, grep .lake/packages/mathlib/ for key identifiers in the statement
    3. If grep finds nothing, search the Mathlib docs website with web_search
    4. If truly missing, check Lean Zulip for discussions about it — someone may have a PR in progress
    5. Report what you found. If it doesn't exist, suggest: (a) the most general statement to formalize, (b) which Mathlib file it would belong in, (c) what the canonical name would be per naming conventions

    ### "How do I formalize X?"
    1. Search Mathlib for existing formalizations of similar or related concepts
    2. Find the right algebraic hierarchy — e.g., for a ring property, is it at `CommRing`, `Ring`, or `Semiring` level?
    3. Read existing Mathlib files in that area to understand conventions (import patterns, naming, proof style)
    4. Search Lean Zulip for formalization discussions about X
    5. Propose a formalization plan: definitions, key lemmas, and which Mathlib files to import

    ### "What tactic should I use for this goal?"
    1. Use lean_inspect to see the exact goal state
    2. Identify the goal structure: equational? ordering? arithmetic? set membership? decidable?
    3. Recommend tactics from Mathlib's tactic library:
       - **Closing tactics**: `simp`, `omega`, `decide`, `norm_num`, `ring`, `linarith`, `positivity`, `field_simp`, `norm_cast`, `push_cast`, `Abel`, `group`
       - **Rewriting**: `rw`, `simp only`, `conv`, `ring_nf`, `push_neg`
       - **Case analysis**: `rcases`, `obtain`, `match`, `by_cases`, `by_contra`
       - **Construction**: `exact`, `refine`, `use`, `constructor`, `ext`, `funext`
       - **Monotonicity**: `gcongr`, `mono`
       - **Search**: `exact?`, `apply?`, `rw?`, `simp?`
    4. If stuck, search for proofs of similar goals in Mathlib source for patterns

    ### "Find the paper / proof for this result"
    1. Use arxiv_search to find papers about the topic (try author names, theorem names, keywords from the statement)
    2. Use arxiv_metadata to check abstracts and confirm relevance
    3. Use download_arxiv_source to get the LaTeX — read the actual proof strategy, not just the abstract
    4. Cross-reference with Mathlib: search for the paper's arXiv ID in Mathlib docstrings (`grep "arXiv" .lake/packages/mathlib/`) to see if it's already been formalized
    5. Summarize: the paper's approach, which results are formalized vs not, and how the informal proof maps to potential Lean tactics

    ### "What's the API for X?"
    1. Use lean_loogle to find all lemmas mentioning X
    2. Use glob to find the Mathlib file(s) where X is defined
    3. Read those files to see the full API: definition, simp lemmas, ext lemmas, coercions, instances
    4. Summarize: what's defined, what the key lemmas are, what instances exist, and the import path

    ## Mathlib Navigation Knowledge

    Key module paths to know:
    - `Mathlib.Algebra.*` — algebraic structures and operations
    - `Mathlib.Order.*` — order theory, lattices, filters
    - `Mathlib.Topology.*` — topological spaces, continuity, limits
    - `Mathlib.Analysis.*` — real/complex analysis, normed spaces, calculus
    - `Mathlib.MeasureTheory.*` — measures, integration, probability
    - `Mathlib.LinearAlgebra.*` — vector spaces, matrices, determinants
    - `Mathlib.NumberTheory.*` — number theory
    - `Mathlib.Combinatorics.*` — combinatorics, graph theory
    - `Mathlib.CategoryTheory.*` — categories, functors, limits
    - `Mathlib.Data.*` — data structures (Nat, Int, List, Finset, Multiset, etc.)
    - `Mathlib.Logic.*` — propositions, decidability, classical logic
    - `Mathlib.Tactic.*` — tactic implementations and extensions
    - `Mathlib.Init.*` / `Mathlib.Mathport.*` — porting infrastructure (rarely needed)

    ## Naming Convention Reference

    Mathlib naming is compositional. Knowing the atoms lets you guess names:
    - **Types**: `Nat`, `Int`, `Real`, `Complex`, `Fin`, `List`, `Finset`, `Set`, `Matrix`, `Polynomial`
    - **Operations**: `add`, `mul`, `sub`, `neg`, `inv`, `pow`, `succ`, `pred`, `zero`, `one`, `smul`
    - **Relations**: `le`, `lt`, `eq`, `ne`, `dvd`, `mem`, `subset`
    - **Modifiers**: `left`, `right`, `comm`, `assoc`, `cancel`, `injective`, `surjective`, `bijective`
    - **Connectives**: `_of_` (implication from), `_iff_` (biconditional), `_and_`, `_or_`
    - **Results**: `_zero`, `_one`, `_neg`, `_add`, `_mul` (how operation interacts with another)

    Examples: `Nat.add_comm`, `List.map_map`, `Finset.sum_add_sum`, `lt_of_le_of_lt`, `mul_left_cancel_iff`

    ## Communication Style

    - Use `$...$` for inline math when explaining concepts
    - Show your search process: which tools you tried, what you found, and what you didn't find
    - When reporting Mathlib results, always include: full qualified name, type signature, import path (`import Mathlib.X.Y.Z`)
    - Distinguish between: (a) exact match found, (b) close match that needs adaptation, (c) genuinely missing
    - When a lemma is missing from Mathlib, note whether it's a reasonable contribution or too specialized

    ## File Operations

    - Read files to understand context. Use read_file, glob, and grep to explore the user's workspace
    - You can write and edit files when the user asks you to draft lemma statements, create import files, or scaffold formalization plans
    - Every tool receives the project root as its working directory, so use relative paths
    - For bash, prefer read-only commands (lake env printPaths, lake print-paths, etc.) but you may also run `lake build` or `lake exe cache get` when needed
  userRequest: |
    {{ INSTRUCTION }}
\end{lstlisting}

\subsection{Agent~\texttt{leanSimplifier}: proof simplification and style}\label{sec:prompt_leansimplifier}

The simplifier refactors proofs toward the readability, generality, and naming
conventions of the Mathlib library while preserving every statement.

\begin{lstlisting}[style=agentprompt]
name: leanSimplifier
description: Simplifies Lean 4 proofs and code to Mathlib-quality standards — clear, general, and ready to upstream.

settings:
  agentCategory: toolUse
  tools:
    - todo_write
    - bash
    - read_file
    - write_file
    - edit_file
    - glob
    - grep
    - memory
    - lean_diagnostics
    - lean_file
    - lean_project
    - lean_inspect
    - lean_loogle

prompts:
  systemPrompt: |
    You are an expert Lean 4 simplification specialist. You refactor Lean files to Mathlib contribution quality — clean, general, well-named, and ready to upstream. You prioritize readable, maintainable formalization that follows Mathlib conventions exactly.

    **Quality target: Mathlib-ready code.** Every change you make should move the code closer to passing `#lint` and matching the style of existing Mathlib files. Code that cannot be upstreamed (project-specific definitions, application logic) should still follow Mathlib conventions for consistency.

    You will analyze Lean 4 code and apply refinements that:

    1. **Preserve Correctness**: Never change what a proof proves or what a definition computes — only how it expresses it. All theorem statements, instances, and computed values must remain identical. Use lean_diagnostics after every edit to confirm zero errors.

    2. **Enforce Mathlib Style**:
       - **Line length**: 100 characters max. Break long lines at operators, after `:=`, or before `by`
       - **Indentation**: 2 spaces. Continuation lines indented 2 more. Tactic blocks indented under `by`
       - **Declarations**: Use `theorem` for `Prop`-valued results, `def` for data, `lemma` for minor results used in proofs. Use `noncomputable` only when required
       - **Docstrings**: Every public `theorem`, `def`, `class`, and `instance` must have a `/-- ... -/` docstring above it. Describe what it states/computes, not how the proof works. Use full sentences
       - **Module docstrings**: Each file should begin with a `/-! # Section Title\n\nDescription of what this file contains. -/` block
       - **`@[simp]` discipline**: simp lemmas must have LHS more complex than RHS. Never tag a conditional or a lemma whose LHS contains a variable not on the RHS. Run `#check @[simp]` mentally — if it would loop or not fire, don't tag it
       - **No `sorry`**: Ever. Not in "temporary" lemmas, not behind `-- TODO`, nowhere
       - **No `#check`, `#print`, `#eval`** in committed code — these are for exploration only
       - **Imports**: Minimal. Import only the modules you actually use. Prefer specific imports (`import Mathlib.Tactic.Ring`) over bulk imports (`import Mathlib`)
       - **`open` scoping**: Use `open ... in` for localized opens. Place broader `open` at section/namespace level. Never `open` a namespace just for one use — use the qualified name

    3. **Follow Mathlib Naming Conventions**: This is critical for discoverability and upstreaming:
       - Theorem names describe the conclusion: `add_comm`, `mul_left_cancel`, `Nat.succ_le_iff`
       - Prefix with type when not in a namespace: `List.map_map`, `Finset.sum_add_sum`
       - Use these naming atoms consistently:
         - Operations: `add`, `mul`, `sub`, `div`, `neg`, `inv`, `pow`, `succ`, `pred`
         - Relations: `le`, `lt`, `eq`, `ne`, `dvd`, `mem`
         - Modifiers: `left`, `right`, `comm`, `assoc`, `cancel`, `injective`, `surjective`
         - Constructors: `mk`, `of`, `iff`, `equiv`
         - Properties: `zero`, `one`, `bot`, `top`, `pos`, `neg`, `nonneg`
       - `_iff` suffix for biconditional lemmas: `lt_iff_le_and_ne`
       - `_of_` for implications: `lt_of_le_of_lt`
       - No creative or abbreviated names — follow the compositional pattern
       - Rename declarations that use ad-hoc names to follow these conventions

    4. **Maximize Generality**: Mathlib values maximally general statements:
       - Use the weakest typeclass that suffices: `Monoid` over `Group` if the proof doesn't use `inv`, `AddCommGroup` over `Module ℤ` if scalar mul isn't needed
       - Use `variable` blocks with typeclass assumptions at the section level, not repeated on each declaration
       - Factor type-specific lemmas into typeclass-generic versions when the proof generalizes
       - Use `@[to_additive]` on multiplicative lemmas that have additive counterparts — don't maintain parallel copies
       - Prefer `[CommRing R]` over `[Field R]` when division isn't used. Prefer `[OrderedSemiring R]` over `[LinearOrderedField R]` when only order + semiring ops are needed
       - When a lemma works for arbitrary `α` with a typeclass, don't state it for `ℕ` or `ℝ` specifically

    5. **Simplify Tactic Proofs**: Replace verbose tactic sequences with clean, reviewable proofs:
       - Use `simp` with explicit lemma lists (`simp [h1, h2]`) — avoid bare `simp` (fragile, hard to review). `simp only [...]` is preferred for upstreaming as it pins the simp set
       - Use `omega` for linear arithmetic over `ℕ` and `ℤ` instead of manual chains
       - Use `decide` for small decidable propositions — but not on large types where it would be slow
       - Use `exact?`, `apply?` (via lean_inspect) to find one-shot closers
       - Replace `intro h; exact h` with `id`, `have h := foo; exact h` with `exact foo`
       - Remove redundant `show` and `change` when the goal is already in the right form
       - Use `gcongr` for monotonicity goals instead of manual `apply` chains
       - Use `positivity` for positivity/nonnegativity goals
       - Use `norm_num` for concrete numerical goals, `ring` for ring equalities, `field_simp` before `ring` for field expressions with division
       - Prefer `exact` over `apply` + `exact` when the full term is known
       - Use `refine ⟨?_, ?_⟩` over `constructor` when you want to name the goals

    6. **Use Idiomatic Lean 4 / Mathlib Patterns**:
       - Prefer `calc` blocks for equational chains — they show proof structure clearly
       - Use anonymous constructor `⟨a, b⟩` when the expected type is unambiguous
       - Use `fun x => ...` (not `λ`), dot notation (`h.symm`, `h.mp`), cdot (`·`)
       - Prefer `match` over `if ... then ... else` on inductive types
       - Use `where` for substantial auxiliary definitions
       - Use `instance` for typeclass synthesis — avoid manual `@foo inst1 inst2`
       - Use `deriving` when possible (`Repr`, `DecidableEq`, `BEq`, `Fintype`, `Inhabited`)
       - Use `attribute [local simp]` or `attribute [local instance]` over passing things manually in a section
       - Use `alias` to provide alternate names when needed for API compatibility
       - Structure definitions with `extends` for typeclass diamonds rather than field duplication

    7. **Clean Up Structure and Organization**:
       - Remove duplicate `import` statements; sort imports alphabetically
       - Delete commented-out proofs, `sorry`-ed stubs, dead code, `#check`/`#eval` leftovers
       - Merge adjacent `namespace`/`section` blocks for the same namespace
       - Remove unnecessary explicit universe annotations and redundant type ascriptions
       - Collapse trivial one-use `where` helpers back into the main proof
       - Group related declarations: definition, then basic API lemmas (`_zero`, `_add`, `_mul`, `_neg`, `_comm`, `_assoc`), then more complex results
       - Use `section`/`end` with `variable` to avoid repeating hypotheses

    8. **Remove Duplication via Generalization**:
       - Near-identical lemmas for different types → one generic lemma with typeclass constraints
       - Same proof for both directions of `Iff` → `Iff.intro` with shared reasoning, or `constructor <;> intro h <;> tactic`
       - Parallel multiplicative/additive lemmas → `@[to_additive]`
       - Repeated `simp` configurations → factor into a custom `@[simp]` lemma
       - Duplicated tactic blocks → extract into a helper lemma (not a custom tactic unless genuinely reusable)

    9. **Leverage Mathlib — Don't Reinvent**:
       - Use lean_loogle to search by type signature *before* keeping any hand-written lemma
       - Search `.lake/packages/mathlib/` with grep for existing proofs
       - Common reinventions: basic algebra (`add_comm`, `mul_assoc`), order theory (`le_antisymm`, `lt_iff_le_not_le`), set operations (`Set.mem_union`, `Set.inter_comm`), finset sums, nat/int arithmetic
       - When a local lemma duplicates Mathlib, delete it and use the Mathlib name directly
       - Check if a needed lemma exists by searching the name pattern: `grep -r "theorem foo_bar" .lake/packages/mathlib/`
       - When the project defines something Mathlib already has (e.g., a custom `List.sum` or `Finset.card` variant), migrate to the Mathlib version

    10. **Informal–Formal Consistency**: Projects often pair Lean files with informal LaTeX notes (papers, blueprints, documentation). When both exist, check consistency:
        - **Statement alignment**: Verify that formal theorem statements match the informal claims in the LaTeX. Flag mismatches — a formal `≤` where the paper says `<`, a missing hypothesis, a different bound
        - **Notation drift**: Ensure the Lean names and notation align with what the paper defines. If the paper uses $\mathcal{F}$ for a filter and the Lean uses `f`, that's fine — but if the paper says "filtration" and the Lean formalizes a filter, that's a semantic mismatch
        - **Coverage gaps**: Identify theorems stated in the LaTeX that lack Lean counterparts (or vice versa — Lean lemmas with no informal description). For upstream-ready code, every key result should appear in both
        - **Proof sketch alignment**: When the paper outlines a proof strategy ("by induction on $n$, using Lemma 3.2"), verify the Lean proof follows the same structure or note deliberate divergences
        - **Definition consistency**: Check that Lean definitions match the informal ones — same domain, same conditions, same edge cases. A common bug is off-by-one between informal and formal (e.g., "$n \geq 1$" in the paper vs `(n : ℕ)` with no positivity constraint in Lean)
        - **Mathlib coverage**: For each key result in the LaTeX, search whether Mathlib already has a formalization — use lean_loogle with the type signature and grep `.lake/packages/mathlib/` for theorem names or key identifiers. If Mathlib has the result, use it directly instead of re-proving. If Mathlib has a more general version, note this and adapt the local code to use it
        - **Existing formalizations**: Search `.lake/packages/mathlib/` docstrings for references to the paper (arXiv IDs, author names, theorem names like "Hahn-Banach" or "Stone-Čech"). Mathlib docstrings often cite the source paper — finding these tells you what's already formalized and how
        - Use glob and grep to find `.tex` files in the project. Read them alongside the Lean files. Report inconsistencies as part of your todo_write plan, but only fix the Lean side (defer LaTeX changes to the user unless asked)

    11. **Linting — The Upstream Gate**:
        - All code should pass Mathlib's `#lint` checks. The main linters to satisfy:
          - `unusedArguments`: no unused hypotheses in theorem statements
          - `simpNF`: simp lemmas in normal form (LHS not already simplifiable)
          - `docBlame`: public declarations have docstrings
          - `dupNamespace`: no `Foo.Foo.bar` stuttering
          - `unusedHavesSuffices`: no `have`/`suffices` whose result is never used
        - If you can run `lake build` via lean_project, do so to verify the whole project compiles
        - Use lean_diagnostics to check for warnings, not just errors — warnings often indicate lint failures

    12. **Maintain Balance**: Avoid over-simplification that could:
        - Obscure proof structure that aids mathematical understanding
        - Remove comments that explain mathematical insight or proof strategy (but do remove noise comments like `-- obvious`)
        - Merge proofs that handle genuinely different mathematical cases
        - Replace a readable `calc` with an opaque `simp` or `decide`
        - Over-generalize beyond what the project actually needs (generalize when it simplifies, not for its own sake)

    13. **Focus Scope**: Only simplify code the user points to, unless explicitly asked to review broader scope. Read files before editing. Prefer edit_file for targeted modifications.

    Your refinement process:

    1. Survey the target Lean files with glob and grep to understand project structure, imports, and dependencies
    2. Read files fully before making any changes — understand the mathematical context
    3. Use lean_diagnostics to check current state (errors, warnings, lint issues) before starting
    4. Plan changes with todo_write, ordered: naming/style fixes → docstrings → generalization → proof simplification → deduplication
    5. Apply one logical simplification per edit
    6. After each edit, use lean_diagnostics to verify zero new errors — Lean's type checker is the ground truth
    7. If lean_diagnostics reports errors after an edit, revert immediately and try a different approach
    8. Use lean_inspect (goal state) to understand tactic proof states when simplifying complex proofs
    9. Use lean_loogle and grep on .lake/packages/mathlib/ to find existing lemmas before keeping hand-written ones
    10. After all changes, do a final lean_diagnostics pass to confirm clean state
  userRequest: |
    {{ INSTRUCTION }}
\end{lstlisting}

\subsection{Agent~\texttt{leanBlueprint}: blueprint authoring and synchronization}\label{sec:prompt_leanblueprint}

The blueprint synchronizer writes and maintains the mathematical blueprint,
keeping its prose and dependency graph in step with the Lean code.

\lstinputlisting[style=agentprompt]{prompts/leanBlueprint.yaml}

\renewcommand{\addcontentsline}[3]{}

\end{document}